\documentclass[11pt]{article}
\usepackage{graphicx} 

\usepackage{mathpazo}
\usepackage{pdfpages}
\usepackage{bm}
\usepackage[margin=1in]{geometry}
\usepackage{amsmath,amssymb}

\usepackage[utf8]{inputenc}
\usepackage{amsmath,amsfonts,fullpage,amsthm,mathrsfs,fancyhdr,verbatim,framed,hyperref,xspace,mdframed,tabularx,booktabs,enumitem,comment,float}
\usepackage{xcolor}
\usepackage{cleveref}
\usepackage{color}
\usepackage{graphicx}
\usepackage{amsthm}
\usepackage{mathpazo}
\usepackage{amsmath}
\usepackage{amssymb}
\usepackage{mathtools}
\usepackage{thmtools,thm-restate}
\usepackage{fec} 
\usepackage{authblk}

\usepackage{algorithm}
\usepackage{algpseudocode}
\usepackage{amsmath}

\newtheorem{theorem}{Theorem}
\numberwithin{theorem}{section}

\newtheorem{lemma}[theorem]{Lemma}
\newtheorem{condition}[theorem]{Condition}

\newtheorem{proposition}[theorem]{Proposition}
\newtheorem{assumption}[theorem]{Assumption}
\newtheorem{corollary}[theorem]{Corollary}

\newtheorem{definition}[theorem]{Definition}

\usepackage{tikz}
\usetikzlibrary{arrows.meta,positioning,fit,calc,backgrounds}
\newcount\Comments
\newcommand{\kibitz}[2]{\ifnum\Comments=1{\leavevmode\color{#1}{#2}}\fi}

\newcommand{\review}[1]{\ifnum\Comments=1\par\noindent\textcolor{red}{\textbf{Review:} #1}\par\fi}

\definecolor{lightteal}{rgb}{0.5, 0.8, 0.8}
\usepackage[style=alphabetic, maxbibnames=99, backref=true, backrefstyle=three]{biblatex}
\bibliography{biblio}

\title{Super-Quadratic Quantum Speedups for Combinatorial Optimization via Tilted Walks}
\author[1]{Guneykan Ozgul\thanks{\texttt{guneykan.ozgul@jpmchase.com}.}}
\author[1]{Shouvanik Chakrabarti\thanks{\texttt{shouvanik.chakrabarti@jpmchase.com}.}}
\affil[1]{Global Technology Applied Research, JPMorganChase, New York, NY 10001, USA}
\date{\today}

\begin{document}
\maketitle

\begin{abstract}
We introduce \emph{quantum tilted walks}, a quantum algorithmic framework for solving exact combinatorial optimization problems. The framework applies an average of powers of a tilted Hamiltonian that biases the discriminant matrix of a base Markov chain (mixer) with the objective function. Our starting point is quantum short-path algorithms, which prepare the ground state of such a Hamiltonian and obtain super-quadratic speedups over exhaustive search for certain combinatorial optimization problems. Recently, Le Gall and Tamaki~(arXiv:2604.12131) developed a classical conditioning-and-search algorithm for weighted \textup{\textsc{MAX-E$k$-LIN2}} and weighted \textup{\textsc{MAX-$k$-CSP}}. Under the same assumptions, their algorithm is only sub-quadratically slower than quantum short-path algorithms. Consequently, existing short-path algorithms do not establish a super-quadratic speedup over this stronger classical baseline. For maximization problems, we give conditions under which tilted walks increase the amplitude on the target state with high objective value when initialized from a starting state with lower objective value. This framework captures conditioning-and-search and yields super-quadratic speedups over it for the same problems. While our framework recovers quantum short-path algorithms as a special case, it neither requires ground-state preparation nor initialization in the ground state of the base mixer. We demonstrate these advantages on a synthetic optimization problem for which tilted walks achieve a super-quadratic speedup whereas the short-path algorithms do not.
\end{abstract}

\section{Introduction}
\subsection{Motivation and background}
Let $f:\mathcal{X}\to\mathbb{R}$ be an objective function on a finite set of feasible configurations. We consider the problem of finding a maximizer of $f$. When $\mathcal{X}=\{0,1\}^n$ and no additional structure is available, exhaustive search costs $O^*(2^n)$\footnote{The notation $O^*$ suppresses polynomial factors.}, whereas quantum maximum finding costs $O^*(2^{n/2})$~\cite{durr1996minimum}, giving the optimal quadratic speedup in the black-box model.

Classical algorithms such as backtracking, branch-and-bound, dynamic programming, and Markov chain search can exploit structure in $\mathcal X$ or $f$ to run significantly faster than exhaustive search. Quantum algorithms can give further speedups by quadratically accelerating certain subroutines of the classical algorithm via quantum amplitude amplification, quantum walk search, or related techniques ~\cite{ambainis2019dynamicprogramming,magniez2011search,montanaro2020branchbound,chakrabarti2022universal, montanaro2016quantumwalkspeedupbacktracking, montanaro2015MonteCarlo}. Since these quantum algorithms obtain their advantage only by quadratically accelerating classical subroutines, their overall speedup over the corresponding classical algorithms is at most quadratic.

On the other hand, a quantum algorithm achieves a \emph{super-quadratic} speedup over a classical algorithm with runtime $T_{\mathrm{cl}}$ if the quantum runtime satisfies $T_q=O^*(T_{\mathrm{cl}}^{1/2-c})$ for some constant $c>0$ independent of the problem size. This stronger scaling is particularly important for fault-tolerant implementations: differences in clock rates, error-correction overhead, and limited quantum parallelism can offset the benefit of a quadratic reduction in runtime. Resource estimates show that improving the scaling beyond $O^*(T_{\mathrm{cl}}^{1/2})$ can substantially reduce the required resources and make quantum advantage more practical~\cite{babbush2021focus}. Beating the quadratic barrier, however, requires new quantum algorithmic ideas beyond quantization of classical subroutines via standard techniques.

Quantum short-path algorithms provide one such idea: before searching for an optimum, they prepare a quantum state with greater overlap with high-value configurations. The algorithm begins in the ground state of a simple Hamiltonian that moves amplitude between feasible configurations, called the mixer Hamiltonian. It adds a diagonal term derived from the objective so that configurations with larger objective values have lower energy, prepares the ground state of the resulting objective-biased Hamiltonian, and then amplifies its overlap with the optimum.  This algorithm was introduced by Hastings~\cite{hastings2018shortpathquantumalgorithm} as an alternative to quantum adiabatic optimization, which continuously tracks the ground state while interpolating from a mixer Hamiltonian to an objective Hamiltonian. Short path algorithms instead prepare the ground state without following the adiabatic interpolation up to the end. 

Subsequent work weakened the assumptions of the original paper and analyzed the short-path algorithm for different problems and settings~\cite{hastings2018weaker,hastings2019toy, Dalzell_2023,chakrabarti2025generalizedshortpathalgorithms}. Dalzell et al.~\cite{Dalzell_2023} proved runtimes of the form $O^*(2^{(1-c)n/2})$, for a constant $c>0$ independent of $n$, for several families of binary optimization problems. In these applications, the mixer Hamiltonian is $-X/n$, where $X=\sum_{i=1}^nX_i$ and $X_i$ flips the $i$-th bit. Thus $X/n$ moves amplitude between bit strings that differ in one coordinate.

Chakrabarti et al.~\cite{chakrabarti2025generalizedshortpathalgorithms} replaced the bit-flip mixer $X/n$ by an operator obtained from a reversible Markov chain $P$ with stationary distribution $\pi$. The resulting symmetric operator, called the discriminant of $P$, has matrix elements
\begin{equation}
    D(P)_{x,y}:=\sqrt{\frac{\pi(x)}{\pi(y)}}P(x,y)=\sqrt{P(x,y)P(y,x)},
\end{equation}
where the second equality follows from reversibility. The coherent stationary state $|\sqrt{\pi}\rangle:=\sum_x\sqrt{\pi(x)}|x\rangle$ is the top eigenstate of $D(P)$ and hence the ground state of the mixer Hamiltonian $-D(P)$. The generalized short-path algorithm starts in this state, adds a diagonal term derived from the objective, prepares the resulting ground state, and amplifies its overlap with the optimum. Under suitable regularity assumptions, this construction gives super-quadratic speedups over the associated Markov chain search algorithms that search directly by sampling from $\pi$.

Le Gall and Tamaki~\cite{gall2026dequantizingshortpathquantumalgorithms} showed that some spectral-density condition used to analyze the short-path algorithm also allows a different classical algorithm, called conditioning-and-search. The algorithm samples from a near-optimal threshold set and searches a local neighborhood around the sampled point. For weighted \textup{\textsc{MAX-E$k$-LIN2}} and weighted \textup{\textsc{MAX-$k$-CSP}}, their analysis identifies an exponentially large successful subset of the threshold set from which local search reaches an optimum. The resulting classical runtime is $2^{(1-c^{\mathrm{cl}})n+o(n)}$, and a standard quadratic quantum acceleration has runtime $2^{(1-c^{\mathrm{cl}})n/2+o(n)}$. The quadratically accelerated runtime has a smaller exponent than those proven for the short-path algorithm by Dalzell et al.~\cite{Dalzell_2023}. Consequently, the short-path algorithm does not establish a super-quadratic speedup over conditioning-and-search.

We argue that this comparison motivates a more fundamental question for quantum algorithms. Consider a classical algorithm that samples a set of candidate configurations and then searches for the optimum via a fixed postprocessing procedure. Let us call a candidate configuration ``successful'' if a desired optimum is successfully found by postprocessing from that candidate. The runtime of such an algorithm is partially dictated by the probability assigned to successful candidates. Conditioning-and-search algorithms naturally fit into this framework, and it also intuitively captures popular classical heuristics such as simulated annealing. A quantum algorithm can improve this runtime if, before postprocessing, it increases the amplitude of successful candidates in a coherent encoding of the candidate distribution.

Existing generalized short-path algorithms do not establish such an amplitude increase. Instead, they starts in the coherent stationary state of the mixer and amplify the amplitude of near-optimal solutions compared to that state. An application-specific candidate distribution need not be stationary for the chosen mixer. Using that distribution within the (generalized) short path algorithm would therefore require constructing a new mixer and establishing the spectral conditions and regularity properties needed for super-quadratic speedup. The existing analysis of short-path algorithms requires a mixer that has inverse-polynomial log-Sobolev constant (hence spectral gap), and requires the objective to be smooth under the action of the mixer. In general, it may be impossible to realize both conditions for the candidate distribution, exposing a structural limitation in the current algorithm.

We introduce \emph{quantum tilted walks}, a quantum algorithmic framework that addresses this shortcoming. We prove conditions under which tilted walks increase the amplitude on high-value target states from application-specific initial states, without requiring ground-state preparation or initialization in the mixer ground state. The framework recovers quantum short-path algorithms as a special case and gives super-quadratic speedups over conditioning-and-search for weighted \textup{\textsc{MAX-E$k$-LIN2}} and weighted \textup{\textsc{MAX-$k$-CSP}} under the same assumptions. We also identify some generic advantages of the framework over short-path algorithms and demonstrate them on a synthetic problem for which tilted walks have a better proven runtime than the standard generalized-short-path implementation using the same mixer. The following sections discuss our results in more detail.

\section{Results}
\label{sec:results}
Our results can be grouped into the following three themes.
\begin{enumerate}[leftmargin=*]
    \item \textbf{Quantum tilted walks.} We introduce a framework that biases a reversible Markov-chain walk toward configurations with larger objective value and transfers amplitude from an application-dependent starting state $|\psi_0\rangle$ to an application-dependent \emph{successful} state $|S\rangle$. Under explicit regularity conditions, the amplitude on $|S\rangle$ increases by an exponential factor compared to the starting amplitude $|\langle \psi_0|S \rangle|$, as long as the successful set has a higher minimum objective value than the initial state. This objective value condition will be made more precise in the subsequent section. Our framework can be viewed as a generalization of discrete time quantum walks, which are obtained when there is no biasing or tilting term. Generalized short-path algorithms arise as a special case when the initial state is the ground state of the mixer and the applied power is large enough to approximate the projector onto the ground state of the tilted Hamiltonian.
    \item \textbf{Applications to exact optimization.} We apply quantum tilted walks to the conditioning-and-search algorithm of Le Gall and Tamaki. For weighted \textsc{MAX-E$k$-LIN2} and weighted \textsc{MAX-$k$-CSP} under the same assumptions as their classical analysis, quantum tilted walks improve on the standard quadratic quantum acceleration and yield a super-quadratic speedup over conditioning-and-search. The comparison uses exactly the classical runtime exponents proved by Le Gall and Tamaki, rather than those of a weaker classical baseline. Moreover, because generalized short-path algorithms are a special case of quantum tilted walks, the applications established in previous short-path work are also included in our framework. Finally, if new conditioning-and-search algorithms are derived from these previous applications of short-path, we expect our results to extend also to super-quadratic speedups over those algorithms.
    \item \textbf{Further advantages of quantum tilted walks.} Quantum tilted walks allow the initial state to be chosen independently of the stationary state of the mixer. We prove an amplitude guarantee for nonnegative initial states whose amplitudes and objective values are sufficiently preserved by one mixer step, while retaining the original mixer. Quantum tilted walks also apply a bounded number of powers instead of preparing the ground state of the tilted Hamiltonian, so the guarantee does not require a lower bound on the spectral gap of the tilted Hamiltonian. We illustrate these advantages with an example in which the initial state has exponentially small overlap with the stationary state and the spectral gaps of both the mixer and the tilted Hamiltonian are exponentially small, yet the tilted walk produces exponential amplitude gain. The example separates the sufficient conditions for tilted walks from those used in the generalized short-path analysis.
\end{enumerate}
A more detailed discussion of these three themes appears in Sections~\ref{sec:quantum-tilted-walks-intro}, \ref{subsec:introduction-exact-optimization}, and~\ref{sec:further-advantages-intro}, respectively. This is followed by Section~\ref{sec:discussion} that discusses some consequences of our results, the connection to classical algorithms, and highlights some open questions for future work. We conclude this section with an overview of Related Work.

\subsection{Overview of Quantum Tilted Walks}
\label{sec:quantum-tilted-walks-intro}
Let $P$ be a reversible Markov chain on $\mathcal X$ with stationary distribution $\pi$. Its discriminant operator is defined by
\begin{equation}
    D_{x,y}
    :=\sqrt{\frac{\pi(x)}{\pi(y)}}P(x,y)
    =\sqrt{P(x,y)P(y,x)},
\end{equation}
where the second equality follows from reversibility: $\pi(x)P(x,y)=\pi(y)P(y,x)$ for all $x,y\in\mathcal X$. The operator $D$ is real symmetric, entrywise nonnegative, and has spectral norm one. Its principal eigenstate is the coherent stationary state $|\sqrt{\pi}\rangle=\sum_{x\in\mathcal X}\sqrt{\pi(x)}|x\rangle$. For the hypercube walk, $D=X/n$, where $X=\sum_{i=1}^nX_i$ is the sum of Pauli $X$ matrices acting on each coordinate.

Let $G$ be a diagonal operator whose entries are scaled objective values in $[1/2,1]$. For $b\geq0$, define
\begin{equation}
    D_b=\frac{D+bG}{E_b},
    \qquad
    E_b=\lambda_{\max}(D+bG),
\end{equation}
and, for an integer $M\geq0$, define the power average
\begin{equation}
    C_{b,M}
    =\frac{1}{M+1}\sum_{m=0}^{M}D_b^m.
    \label{eq:introduction-power-average}
\end{equation}
We call $D_b$ a tilted walk operator and $C_{b,M}$ an averaged tilted walk. The normalization by $E_b$ makes $D_b$ a contraction ($\|D_b\|\leq 1$) but does not make it a stochastic matrix.

Before applying the tilted walk, we choose an initial candidate state $|T\rangle$, which may be any quantum superposition over configurations in $\mathcal X$. Its amplitudes and support are application dependent. For example, $|T\rangle$ may be supported on a near-optimal threshold set, a feasible region, or a neighborhood selected by a classical approximation algorithm. We assume access to a state-preparation unitary $U_T$ satisfying $U_T|0\rangle=|T\rangle$ and to $U_T^\dagger$.  The analysis below considers initial states $|T\rangle$ that correspond to particular classes of starting distributions over candidate configurations, which will be sufficient for our applications. However, these assumptions on the amplitudes of $|T\rangle$ are not needed to implement the algorithm and apply $C_{b,M}$. The cost of $U_T$ and $U_T^\dagger$ is included in the total runtime and need not be polynomial.

Independent of the initial state, we choose a successful set $S\subseteq\mathcal X$. A configuration in $S$ may already be optimal, may have a sufficiently large objective value, or may allow a subsequent procedure to recover an optimum. The tilted walk algorithm then:
\begin{enumerate}[label=\arabic*.,leftmargin=*]
    \item applies $U_T$ to prepare the initial candidate state $|T\rangle$;
    \item applies a block encoding of $C_{b,M}$ to $|T\rangle$; and
    \item marks and amplifies the subspace spanned by the configurations in $S$, followed by any required postprocessing.
\end{enumerate}
The cost of the final amplification step is governed by the amplitude,
\begin{equation}
    \left\|\langle S| C_{b,M}|T\rangle\right\|.
    \label{eq:introduction-central-amplitude}
\end{equation}
The initial candidate state, mixer, and successful set can therefore be chosen independently. In the conditioning-and-search applications, $|T\rangle$ is the coherent threshold-set state $|T_\eta\rangle$, while $S$ consists of threshold configurations from which the application-specific neighborhood search reaches an optimum.

\paragraph{Amplitude gain} Our main amplitude-gain theorem considers coherent encodings of $\pi$ conditioned on sets, including the threshold states used in our applications. Section~\ref{sec:warm-starts} later treats more general nonnegative initial states. For a nonempty set $A$, write $\pi(A):=\sum_{x\in A}\pi(x)$, let $\pi_A$ denote $\pi$ conditioned on $A$, and let $|A\rangle$ denote its coherent encoding. The inverse overlap $|\langle S|C_{b,M}|T\rangle|^{-1}$ determines the runtime of the algorithm and it is larger than $|\langle S|T\rangle|$ by an exponential factor for our applications.

\begin{theorem}[Informal tilted-walk amplitude gain, formal statement in Theorem~\ref{thm:normalized-amplitude-gain}]
\label{thm:normalized-amplitude-gain-informal}
Assume that $1/2\leq G(x)\leq1$ for every $x\in\mathcal X$ and that $\max_xG(x)=1$, and define $\mu:=\mathbb E_\pi[G]$. Let $\eta,\rho\in(0,1)$ and define
\begin{equation}
    T_\eta:=\{x\in\mathcal X:G(x)\geq1-\eta\}.
    \label{eq:introduction-threshold-set}
\end{equation}
Let $T\supseteq T_\eta$ be nonempty, and let $S$ be a nonempty set such that $G(x)\geq1-\rho$ for every $x\in S$. Let $|T\rangle$ and $|S\rangle$ be the coherent encodings of $\pi_T$ and $\pi_S$, respectively. Assume that $G$ is pseudo-Lipschitz with constant $L\in(0,1]$, meaning that
\begin{equation}
    (PG)(x):=\sum_{y\in\mathcal X}P(x,y)G(y)\geq(1-L)G(x)
    \qquad\text{for every }x\in\mathcal X.
    \label{eq:introduction-score-stability}
\end{equation}
Let $\mathcal L_2(G)>0$ satisfy
\begin{equation}
    \max_{x\in\mathcal X}\sum_{y\in\mathcal X}P(x,y)(G(x)-G(y))^2
    \leq\mathcal L_2(G)^2,
\end{equation}
and assume that $D$ satisfies the log-Sobolev inequality in Condition~\ref{cond:log-Sobolev} with constant $\omega>0$. Finally, assume that $L^{-1}=\operatorname{poly}(n)$, that $\eta,\rho=\Theta(1)$, and that
\begin{equation}
    \eta-\rho-L=\Omega(1),
    \qquad
    1-\rho-\mu-L=\Omega(1).
\end{equation}
Then there are choices $M=O(1/L)=\operatorname{poly}(n)$ and $b^*\in(0,1]$ such that
\begin{align}
|\left\langle S\left| C_{b^{\star},M}\right|T \right\rangle|
&\geq
\frac{1}{\operatorname{poly}(n)}
\sqrt{\frac{\pi(S)}{\pi(T)}}
2^{\kappa n/2},
\label{eq:informal-general-warm-gain}
\end{align}
where
\begin{equation}
    \kappa
    =\frac{2b^*(1-\rho-\mu-L)}{3nL\log 2}
    \mathbb{E}_{x\sim\pi_S}\left[
    \log\left(\frac{G(x)}{1-\rho}\right)
    \right].
\end{equation}
\end{theorem}

The bound in the theorem quantifies the increase relative to the amplitude already present on $S$. The separation assumptions imply $S\subseteq T_\eta\subseteq T$, so the initial overlap is
\begin{equation}
    |\langle S|T\rangle|=\sqrt{\frac{\pi(S)}{\pi(T)}}.
\end{equation}
Since $|S\rangle$ lies in the range of $\Pi_S$, the bound in the theorem also yields a lower bound on the marked amplitude $\|\Pi_SC_{b,M}|T\rangle\|$. Thus quantum tilted walks increase the amplitude already present on $S$ by the factor $2^{\kappa n/2}$, up to an inverse-polynomial prefactor.

The two objective-value levels $1-\eta$ and $1-\rho$ make the gain conditions concrete. The initial set $T$ may be any set containing $T_\eta=\{x:G(x)\geq1-\eta\}$, whereas every configuration in the successful set $S$ must have score at least $1-\rho$. Since $L$ is inverse polynomial, the condition $\eta-\rho-L=\Omega(1)$ requires the successful level to lie a constant distance above the cutoff defining $T_\eta$. Similarly, $1-\rho-\mu-L=\Omega(1)$ requires the successful level $1-\rho$ to exceed the mean $\mu$ under the stationary distribution by a constant margin. Finally, because $G(x)\geq1-\rho$ on $S$, the expectation
\begin{equation}
    \mathbb E_{x\sim\pi_S}\left[\log\left(\frac{G(x)}{1-\rho}\right)\right]
\end{equation}
is automatically nonnegative. It is bounded below by a positive constant when configurations in $S$ typically lie a constant multiplicative amount above the minimum value $1-\rho$. For example, it suffices that a constant fraction of the probability mass of $\pi_S$ lies on configurations satisfying $G(x)\geq(1+\epsilon)(1-\rho)$ for some constant $\epsilon>0$. If, in addition, $b^*=\Omega(1)$ and $L=O(1/n)$, then $\kappa=\Omega(1)$ and the amplitude gain is exponential in $n$. The application sections verify these conditions for the concrete successful sets used there. We also note some conditions on the relative objective values are necessary to get the super-quadratic speedup. For example as a sanity check, when $S=T$, these conditions are violated implying that the amplitude cannot be larger than $1$.  

This freedom in choosing $\eta$, $\rho$, $T$, and $S$ makes the guarantee offered by Theorem~\ref{thm:normalized-amplitude-gain-informal} more flexible than the guarantees offered by short-path algorithms. In particular, we can amplify the amplitude of $S$ by an exponential factor relative to any candidate state $|T\rangle$ that satisfies the assumptions of the theorem, including the threshold state $|T_\eta\rangle$ for every admissible choice of $\eta$. In contrast, the short-path algorithm measures amplification relative to the ground state $|\sqrt{\pi}\rangle$ of its mixer. This flexibility is crucial for obtaining our super-quadratic speedups over conditioning-
and-search.

The pseudo-Lipschitz constant $L$ determines how long a large objective value remains useful under the walk. By induction, $(PG)(x)\geq(1-L)G(x)$ implies $P^tG(x)\geq(1-L)^tG(x)$. Thus, for each $x\in S$, the lower bound on the expected objective value remains above the cutoff defining $T_\eta$ for all powers up to the point-dependent time
\begin{equation}
    m(x):=\left\lceil\frac{1}{L}\log\left(\frac{G(x)}{1-\rho}\right)\right\rceil.
\end{equation}
The amplitude lower bound accumulates over this entire interval. A smaller $L$ therefore both extends the useful range of powers and increases the resulting amplitude gain. This explains why $1/L$ appears in both the exponent $\kappa$ and the required maximum power $M=O(1/L)$. The pseudo-Lipschitz condition is imposed on all of $\mathcal X$ because contributing paths need not remain inside $S$ or $T$.

Averaging over powers is necessary because the useful walk length depends on each point in $S$. Intuitively, the starting from particular point in $T$, the walk might enter the set $S$ and escape from it later. A single power can be too short to gain a nontrivial amplitude for one point and too long to keep another point inside of $T$. Instead of controlling this behavior, averaging simply measures the average number of steps the walk spends in the set $T$. By choosing $M\geq\max_{x\in S}m(x)$, the average $C_{b,M}$ contains a useful power for every point in $S$.  Since all entries in the path expansion are nonnegative, these pointwise contributions can be aggregated without cancellation. Averaging costs only the factor $1/(M+1)$, which is inverse polynomial under the assumptions of the theorem. Its purpose here is therefore to preserve a range of useful walk lengths, not to approximate the projector onto the ground state of the tilted Hamiltonian.

The remaining quantities determine whether the amplitude gain survives after normalization by $E_b$. For the useful power $m=m(x)$, the path contribution grows at least as $[1+b(1-\rho-L)]^m$, whereas replacing $D+bG$ by $D_b$ divides this contribution by $E_b^m$. Thus a useful upper bound on $E_b$ must place it strictly below $1+b(1-\rho-L)$. A constant logarithmic separation between these two quantities produces exponential gain when $m=\Theta(1/L)$. Under the quadratic-variation and log-Sobolev assumptions,
\begin{equation}
    E_b\leq1+b\mu+\frac{b^2\mathcal L_2(G)^2}{8\omega^2}.
\end{equation}
Here $\mu$ determines the baseline normalization cost, while $\mathcal L_2(G)^2/\omega^2$ controls the additional cost caused by fluctuations of $G$. The choice $b=b^*$ balances the gain from increasing the tilt against this quadratic penalty. In particular, it ensures
\begin{equation}
    \log\left(\frac{1+b^*(1-\rho-L)}{E_{b^*}}\right)
    \geq\frac{b^*(1-\rho-\mu-L)}{3}>0.
\end{equation}

The bound on $E_b$ is the only place where the log-Sobolev and quadratic-variation assumptions enter the amplitude-gain theorem. They can be replaced by any other argument that yields a sufficiently strong upper bound on $E_b$, while pseudo-Lipschitzness remains necessary for the positive path bound. Generalized short-path algorithms use related concentration estimates both to control normalization and to prove a spectral gap for ground-state projection. Quantum tilted walks do not use the latter conclusion. The assumptions of the amplitude-gain theorem are therefore not substantially stronger than those used in short-path algorithms; its main regularity conditions resemble the stability and concentration conditions used there~\cite{Dalzell_2023,chakrabarti2025generalizedshortpathalgorithms}.

\paragraph{Connection to quantum walks}
The definition of $C_{b,M}$ has a direct connection to the Markov chain from which $D$ is constructed. This connection is exact when $b=0$ and provides a natural interpretation as a sum over weighted paths when $b>0$.

At zero bias, $E_0=1$ and $D_0=D$. By reversibility,
\begin{equation}
    \langle S|C_{0,M}|T\rangle
    =
    \sqrt{\frac{\pi(T)}{\pi(S)}}
    \frac{1}{M+1}
    \sum_{m=0}^{M}
    \mathbb E_{x\sim\pi_T}\!\left[P^m(x,S)\right],
    \qquad
    P^m(x,S):=\sum_{y\in S}P^m(x,y).
\end{equation}
After accounting for the normalization by $\pi(S)$ and $\pi(T)$, this matrix element is the probability that a classical walk initialized from $\pi_T$ lies in $S$ after a number of steps chosen uniformly from $0,\ldots,M$. The initial distribution, the marked set, and the choice of how long to run the walk are also central to quantum algorithms based on walks~\cite{szegedy2004quantum,magniez2011search}. In the present setting, they are specified by $T$, $S$, and the range $0,\ldots,M$.

For $b>0$, define $K_b:=P+bG$. The similarity relation
\begin{equation}
    D+bG
    =
    \operatorname{diag}(\pi)^{1/2}
    K_b
    \operatorname{diag}(\pi)^{-1/2}
\end{equation}
implies
\begin{equation}
    \langle S|C_{b,M}|T\rangle
    =
    \sqrt{\frac{\pi(T)}{\pi(S)}}
    \frac{1}{M+1}
    \sum_{m=0}^{M}
    \mathbb E_{x\sim\pi_T}\!\left[
        \frac{K_b^m(x,S)}{E_b^m}
    \right],
    \qquad
    K_b^m(x,S):=\sum_{y\in S}K_b^m(x,y).
    \label{eq:introduction-tilted-classical-kernel}
\end{equation}
The matrix $K_b$ is not a transition matrix because the entries in its row indexed by $x$ sum to $1+bG(x)$. Its powers can nevertheless be expressed as sums over paths. In the expansion of $(P+bG)^m$, selecting a factor $P$ advances the path according to the Markov chain, while selecting a factor $bG$ leaves the current configuration unchanged and multiplies the path weight by $bG(x)$. Hence $K_b^m(x,S)$ is a sum of nonnegative weights over paths from $x$ to $S$. Equation~\eqref{eq:introduction-tilted-classical-kernel} places these path sums in a uniform average over lengths $0,\ldots,M$. At zero bias, the summands are transition probabilities; at positive bias, the objective values encountered along a path modify its weight.

\paragraph{Short-path algorithms as a special case}
The same tilted operator $D+bG$ underlies short-path algorithms. The correspondence becomes exact after accounting for the filter and score normalization used in those algorithms.

Consider a minimization objective $H$ with optimum energy $E^\star<0$. For $u\geq-1$, define the short-path filter
\begin{equation}
    g_\eta(u):=\min\left\{0,\frac{u+1-\eta}{\eta}\right\},
    \qquad
    F_\eta:=-g_\eta\left(\frac{H}{|E^\star|}\right).
\end{equation}
Then $0\leq F_\eta\leq I$, and the generalized short-path Hamiltonian is
\begin{equation}
    H_b^{\mathrm{sp}}
    =-D+b\,g_\eta\left(\frac{H}{|E^\star|}\right)
    =-D-bF_\eta.
    \label{eq:introduction-short-path-hamiltonian}
\end{equation}
To place the score in the interval $[1/2,1]$ used above, set $\widetilde G_\eta:=(I+F_\eta)/2$. The identity
\begin{equation}
    D+2b\widetilde G_\eta=D+bF_\eta+bI
\end{equation}
shows that $D+2b\widetilde G_\eta$ and $D+bF_\eta$ have the same eigenvectors. Their principal eigenstate is therefore the ground state of $H_b^{\mathrm{sp}}$. This establishes the correspondence without omitting either the short-path filter or the normalization of $G$.

Let $|\psi_b\rangle$ denote this common eigenstate, and define
\begin{equation}
    \widetilde D_b
    :=\frac{D+2b\widetilde G_\eta}{\lambda_{\max}(D+2b\widetilde G_\eta)}.
\end{equation}
Suppose the eigenvalue $1$ of $\widetilde D_b$ is nondegenerate, and let
\begin{equation}
    \delta_b^{\mathrm{abs}}
    :=1-\max_{\substack{\lambda\in\operatorname{spec}(\widetilde D_b)\\\lambda\neq1}}|\lambda|.
\end{equation}
By the spectral decomposition of $\widetilde D_b$,
\begin{equation}
    \left\|\widetilde D_b^m-|\psi_b\rangle\!\langle\psi_b|\right\|
    \leq(1-\delta_b^{\mathrm{abs}})^m.
    \label{eq:introduction-power-projector}
\end{equation}
If $\delta_b^{\mathrm{abs}}$ and $|\langle\sqrt\pi|\psi_b\rangle|$ are inverse polynomial, then a polynomially large $m$ makes the right-hand side sufficiently small compared with this overlap. Applying $\widetilde D_b^m$ to $|\sqrt\pi\rangle$ and normalizing then prepares an approximation to $|\psi_b\rangle$. Short-path algorithms implement the corresponding projection through ground-state reflections rather than by literally applying this power, but the prepared state and the overlap that controls the runtime are the same.

For the hypercube walk, $D=X/n$ and $|\sqrt\pi\rangle=|+\rangle$, so this specialization coincides with the filtered construction analyzed by Dalzell et al.~\cite{Dalzell_2023}. When $D$ is the discriminant matrix of an arbitrary reversible Markov chain, it coincides with the generalized short-path algorithms by Chakrabarti et al.~\cite{chakrabarti2025generalizedshortpathalgorithms}. Taking $T=\mathcal X$ fixes the initial state to $|\sqrt\pi\rangle$, and the final amplification step targets the optimal subspace.

No averaging is required in this special case because we can just take one sufficiently large power to isolate $|\psi_b\rangle$. The amplitude-gain theorem instead uses powers that need not approximate the projector onto the ground state of the tilted Hamiltonian. Its guarantee relies on the average $C_{b,M}$ because the useful power $m(x)$ can vary with $x\in S$. The tilted operator is therefore common to both analyses, but averaging permits the set-to-set amplitude bound proved here and can retain contributions that a sufficiently large single power would suppress.

\subsection{Applications to Exact Optimization: Super-Quadratic Speedups over Conditioning-and-Search}
\label{subsec:introduction-exact-optimization}
We first briefly describe the classical algorithm. Conditioning-and-search separates exact optimization into sampling and local search. Recall the threshold set $T_\eta$ from~\eqref{eq:introduction-threshold-set}. The classical algorithm draws a uniformly random configuration from $\mathcal{X}$ and accepts it when it lies in $T_\eta$. Repeating this rejection-sampling step produces a uniform sample from $T_\eta$. For each accepted configuration, the algorithm searches an application-dependent neighborhood.

A configuration $x\in T_\eta$ is successful if the neighborhood searched from $x$ contains a global optimum. Let $S\subseteq T_\eta$ denote the set of successful configurations. As in the motivation above, $S$ is used only to analyze the probability that local search succeeds; the algorithm does not need to know this set. In both applications below, $S$ is itself a Hamming ball around an optimum, and the size of the neighborhood searched from any center (and thereby the runtime of the local search) is $|S|$ up to polynomial factors.

In the concrete applications discussed below, $\mathcal{X} = \{0,1\}^n$. One accepted sample from $T_\eta$ requires $2^n/|T_\eta|$ rejection-sampling trials in expectation. The probability that this sample belongs to $S$ is $|S|/|T_\eta|$, so the expected number of accepted samples before local search succeeds is $|T_\eta|/|S|$. Multiplying the number of accepted samples by the cost of rejection sampling and local search, the classical runtime is, up to polynomial factors,
\begin{equation}
    \frac{2^n}{|S|}+|T_\eta|.
    \label{eq:introduction-classical-conditioning-runtime}
\end{equation}
If
\begin{equation}
    |T_\eta|\leq2^{(1-\gamma)n},
    \qquad
    |S|\geq2^{\sigma n-o(n)},
\end{equation}
then this runtime is at most
\begin{equation}
    2^{(1-\min\{\gamma,\sigma\})n+o(n)}.
\end{equation}
Here $\gamma$ characterizes the rarity of the threshold set and $\sigma$ measures the number of successful configurations around an optimum. In the applications below, $\min\{\gamma,\sigma\}$ equals the exponent $c^{\mathrm{cl}}$ of the corresponding classical conditioning-and-search algorithm.

The standard quantum acceleration of a fixed conditioning-and-search procedure uses the same threshold set and local neighborhoods as its classical counterpart. We use amplitude amplification to prepare the threshold state $|T_\eta\rangle$, and use quantum maximum finding to search the neighborhood associated with each center. Without the tilted walk, the amplitude on $S$ is $\sqrt{|S|/|T_\eta|}$. A ``Grover-like'' quantum acceleration that does not leverage tilted walks therefore has runtime
\begin{equation}
    \operatorname{poly}(n)
    \left(
        \sqrt{\frac{2^n}{|S|}}+\sqrt{|T_\eta|}
    \right),
\end{equation}
where each term is the square root of its classical counterpart.

We further accelerate conditioning-and-search by inserting $C_{b,M}$ between threshold-state preparation and the coherent neighborhood search. Suppose applying $C_{b,M}$ increases the amplitude on $S$ by $2^{\kappa n/2}$, up to an inverse-polynomial factor, as established by Theorem~\ref{thm:normalized-amplitude-gain}. The number of amplitude-amplification iterations decreases by the same factor, and Theorem~\ref{thm:optimization-from-amplitude} bounds the runtime by
\begin{equation}
    \operatorname{poly}(n)2^{-\kappa n/2}
    \left(
        \sqrt{\frac{2^n}{|S|}}+\sqrt{|T_\eta|}
    \right).
    \label{eq:introduction-quantum-conditioning-runtime}
\end{equation}
Substituting the bounds on $|T_\eta|$ and $|S|$, the quantum runtime is at most
\begin{equation}
    2^{(1-\min\{\gamma,\sigma\}-\kappa)n/2+o(n)},
    \label{eq:introduction-conditioning-exponents}
\end{equation}
where $\kappa$ measures the additional amplitude produced by the tilted walk. Since $c^{\mathrm{cl}}=\min\{\gamma,\sigma\}$, a constant $\kappa>0$ reduces the exponent beyond the ordinary quadratic quantum speedup.

Proving the claimed runtime therefore reduces to bounding $\gamma$, $\sigma$, and $\kappa$. The threshold exponent $\gamma$ follows from concentration of the centered objective $H$. Under a uniformly random assignment, $H$ has mean zero, and the change caused by flipping one variable is controlled by its weighted degree. Applying McDiarmid's inequality bounds the size of the threshold set $\{x:H(x)\geq H_{\max}/2\}$. The exponent $\sigma$ is determined by the cardinality of an explicit set of successful configurations near an optimum.

Both applications construct this successful set in the same way. Fix a global optimum $x^\star$ and choose the coordinates on which local search may move. For \textup{\textsc{MAX-E$k$-LIN2}}, local search may flip any coordinate. For \textup{\textsc{MAX-$k$-CSP}}, it is restricted to coordinates with sufficiently small weighted degree. Let $S$ be the radius-$r$ Hamming ball around $x^\star$ on the allowed coordinates, and search the corresponding radius-$r$ ball around every sampled configuration $x$. By symmetry of Hamming distance, the neighborhood of every $x\in S$ contains $x^\star$. Thus every configuration in $S$ is successful. It is enough to analyze this explicit subset because any additional successful configurations can only increase the success probability and marked amplitude.

In both applications, we choose $S$ so that
\begin{equation}
    G(x)\geq1-\frac{\rho}{2}
    \qquad\text{for every }x\in S,
\end{equation}
and apply Theorem~\ref{thm:normalized-amplitude-gain} with parameter $\rho$. Since $\rho/2<\eta$, this stronger score bound also implies $S\subseteq T_\eta$. Moreover,
\begin{equation}
    \mathbb E_{x\sim\pi_S}
    \left[
        \log\left(\frac{G(x)}{1-\rho}\right)
    \right]
    \geq
    \log\left(\frac{1-\rho/2}{1-\rho}\right)>0.
\end{equation}
The logarithmic factor in $\kappa$ is therefore bounded below by a positive constant. The application-specific estimates below ensure that $S$ remains exponentially large and that the exact conditioning-and-search exponent is preserved. 

We note that the choice of $S$ typically does not significantly affect the runtime, since the second term in~\eqref{eq:introduction-quantum-conditioning-runtime} is dominant. For example, in Max-$k$-CSP, flipping one bit changes $G$ by $O(k/n)$, so $S$ can be chosen as a Hamming ball of radius $\Theta(n/k)$, giving $2^n/|S|=\exp((\log 2-\Theta((\log k)/k))n+o(n))$. On the other hand, McDiarmid's inequality gives $|T_\eta|\leq\exp((\log 2-\Omega(1/k^2))n)$. Hence, for sufficiently large constant $k$, the runtime exponent is controlled by the $\sqrt{|T_\eta|}$ term.

After normalization, the same bit-flip bounds imply $L=O(1/n)$ and $\mathcal L_2(G)=O(1/n)$. The hypercube walk has log-Sobolev constant $\omega=1/n$. Together, these estimates ensure that $\kappa$ is bounded below by a positive constant.

\paragraph{Weighted \textup{\textsc{MAX-E$k$-LIN2}}.}
The conditioning-and-search analysis of Le Gall and Tamaki uses the same threshold $H(x)\geq H_{\max}/2$~\cite{gall2026dequantizingshortpathquantumalgorithms}. The concentration bound has exponent $\gamma_A$, and the resulting classical runtime exponent is exactly $c_A^{\mathrm{cl}}=\gamma_A$. The analysis of the quantum algorithm uses this threshold and the same concentration bound. Let $d_i$ be the total weight of equations containing variable $i$, let $d_{\max}=\max_i d_i$, and define
\begin{equation}
    \Delta:=\frac{kH_{\max}}{nd_{\max}}.
\end{equation}
The parameter $\Delta$ compares the optimum value with the largest change caused by one coordinate. When $\Delta=\Theta(1)$, the tilted-walk analysis uses a Hamming ball of radius $\Theta(n)$ around an optimum, chosen so that every center satisfies $G(x)\geq1-\rho/2$. The same weighted-degree estimates establish the pseudo-Lipschitz and quadratic-variation bounds. The entropy exponent $s_A$ of this ball satisfies $s_A\geq\gamma_A$, so the threshold-set term remains the bottleneck and the classical exponent in the quantum runtime is exactly $c_A^{\mathrm{cl}}$. The tilted walk lowers the quantum exponent by the additional constant $\kappa_A$.

\begin{theorem}[Informal \textup{\textsc{MAX-E$k$-LIN2}} runtime]
For every fixed $k\geq2$, consider weighted \textup{\textsc{MAX-E$k$-LIN2}} instances satisfying the assumptions of Theorem~\ref{thm:tilted-walk-eklin2}. There is a constant $\kappa_A>0$, independent of $n$, for which the quantum tilted walk algorithm finds an optimum in time
\begin{equation}
    2^{(1-c_A^{\mathrm{cl}}-\kappa_A)n/2+o(n)},
\end{equation}
where $2^{(1-c_A^{\mathrm{cl}})n+o(n)}$ is the runtime of the corresponding classical conditioning-and-search algorithm. The resulting speedup over the classical algorithm is therefore super-quadratic.
\end{theorem}

\paragraph{Weighted \textup{\textsc{MAX-$k$-CSP}}.}
The quantum analysis also uses the light-coordinate construction from the analysis of the conditioning-and-search algorithm by Le Gall and Tamaki~\cite{gall2026dequantizingshortpathquantumalgorithms}. If $\Sigma$ is the sum of the weighted degrees, then at least half of the variables have weighted degree at most $2\Sigma/n$. Restricting local search to these coordinates still permits flips on at least $n/2$ coordinates, while every permitted flip changes the objective by a controlled amount.

The tilted-walk analysis uses the same threshold set as the classical algorithm but reduces the radius of the restricted Hamming ball by a constant factor. This smaller radius ensures that every successful center satisfies $G(x)\geq1-\rho/2$. Let $\gamma_B$ denote the threshold-set exponent and $\sigma_B$ the exponent of this smaller restricted ball. When $\sigma_B\geq\gamma_B$, threshold sampling remains the bottleneck in both analyses, and $c_B^{\mathrm{cl}}=\gamma_B$ is exactly the conditioning-and-search exponent. The same weighted-degree estimates establish the pseudo-Lipschitz and quadratic-variation bounds, and the tilted walk lowers the quantum exponent by the additional constant $\kappa_B$. For weighted exact-$k$ CSPs, the inequality $\sigma_B\geq\gamma_B$ holds automatically when $k\geq6$.

\begin{theorem}[Informal \textup{\textsc{MAX-$k$-CSP}} runtime]
For every fixed $k$, consider weighted \textup{\textsc{MAX-$k$-CSP}} instances satisfying the assumptions of Theorem~\ref{thm:tilted-walk-maxkcsp} and the comparison condition $\sigma_B\geq\gamma_B$ stated there. There is a constant $\kappa_B>0$, independent of $n$, for which the quantum tilted walk algorithm finds an optimum in time
\begin{equation}
    2^{(1-c_B^{\mathrm{cl}}-\kappa_B)n/2+o(n)},
\end{equation}
where $2^{(1-c_B^{\mathrm{cl}})n+o(n)}$ is the runtime of the corresponding classical conditioning-and-search algorithm. The comparison condition holds, in particular, for weighted exact-$k$ CSPs with $k\geq6$.
\end{theorem}

Section~\ref{sec:application-to-optimization} contains the complete proofs corresponding to this section.

\subsection{Further Advantages of Quantum Tilted Walks}
\label{sec:further-advantages-intro}
The applications above start from the threshold state $|T_\eta\rangle$, which is the coherent encoding of the stationary distribution conditioned on $T_\eta$. Quantum tilted walks also allow the initial state to be chosen independently of the stationary distribution. This freedom is useful when a classical algorithm, prior computation, or property of the instance provides a promising distribution that is not stationary for the available mixer. Rather than replacing the mixer, we continue to use $D$ and measure directly how it acts on the chosen initial amplitudes.

A classical motivation for this choice is provided by the phenomenon of ``favorable metastability''. It is well known that classical heuristic algorithms based on Markov chains, such as simulated annealing or the metropolis process, can sometimes find a global optimum effectively even when their low-temperature mixing is exponentially obstructed. Chen, Mikulincer, Reichman, and Wein~\cite{chen_et_al:LIPIcs.APPROX-RANDOM.2025.47} give an example of this phenomenon, observing that the Metropolis process for maximum independent set on $K_{n,n}$ has exponential mixing time at every temperature, yet at an appropriate temperature it finds a maximum independent set in polynomial time when initialized at the empty set. Independent sets concentrated on the two sides of $K_{n,n}$ form regions separated by a bottleneck. Moving between these regions is necessary for mixing, but it is unnecessary for optimization because either side is a maximum independent set. The process can therefore benefit from remaining in one region long enough to reach an optimum. This example motivates quantum algorithms that can start from a state concentrated in a favorable region, even when that state has little overlap with the stationary state.

\paragraph{Amplitude gain from a nonstationary initial state}
Theorem~\ref{thm:normalized-amplitude-gain-informal} applies when the initial state is the coherent encoding of $\pi$ conditioned on a set $T$ containing $T_\eta$. To analyze more general initial states, we replace this requirement by conditions that compare the chosen amplitudes directly with their image under $D$. Theorem~\ref{thm:warm-start-gain} itself assumes neither a log-Sobolev inequality nor a spectral-gap bound. Its lower bound depends on $E_b$, so each application must control this normalization separately.

Let
\begin{equation}
    |\varphi\rangle=\sum_{x\in\mathcal X}\varphi(x)|x\rangle,
    \qquad
    \Omega_\varphi=\{x:\varphi(x)>0\},
\end{equation}
be a normalized state with nonnegative amplitudes. Define
\begin{equation}
    \lambda_\varphi
    =\min_{x\in\Omega_\varphi}
      \frac{(D\varphi)(x)}{\varphi(x)}.
    \label{eq:introduction-warm-compatibility}
\end{equation}
Thus $\lambda_\varphi$ is the largest constant for which $D\varphi\geq\lambda_\varphi\varphi$ holds pointwise on $\Omega_\varphi$. The stationary state satisfies $\lambda_{\sqrt\pi}=1$. A nonstationary state can also have $\lambda_\varphi$ close to one when one application of the mixer approximately preserves its amplitudes.

Assume $\lambda_\varphi>0$. Then $(D\varphi)(x)>0$ throughout $\Omega_\varphi$, and the effect of the mixer on objective values is described by the stochastic kernel
\begin{equation}
    K_\varphi(x,y)
    =\frac{D_{x,y}\varphi(y)}{(D\varphi)(x)},
    \qquad x\in\Omega_\varphi.
\end{equation}
We assume that, for some $L_\varphi\in[0,1]$,
\begin{equation}
    (K_\varphi G)(x)
    \geq(1-L_\varphi)G(x)
    \qquad\text{for every }x\in\Omega_\varphi.
    \label{eq:introduction-warm-score-regularity}
\end{equation}
This condition controls how quickly the objective can decrease under the transitions induced by $D$ and $|\varphi\rangle$. Unlike the pseudo-Lipschitz condition in Theorem~\ref{thm:normalized-amplitude-gain-informal}, it is imposed only on $\Omega_\varphi$ and uses the kernel $K_\varphi$ rather than the transition matrix $P$.

\begin{theorem}[Informal amplitude gain from a nonstationary initial state, formal statement in Theorem~\ref{thm:warm-start-gain}]
\label{thm:warm-start-gain-informal}
Assume that $\lambda_\varphi>0$ and that~\eqref{eq:introduction-warm-score-regularity} holds. Let $b\geq0$, let $M\geq0$ be an integer, and let $S\subseteq\Omega_\varphi$ be nonempty. If $g\in[0,1]$ and $G(x)\geq g$ for every $x\in S$, then
\begin{equation}
    \left\|\Pi_SC_{b,M}|\varphi\rangle\right\|
    \geq
    \frac{\left\|\Pi_S|\varphi\rangle\right\|}{M+1}
    \prod_{t=0}^{M-1}
    \frac{\lambda_\varphi+b(1-L_\varphi)^tg}{E_b}.
    \label{eq:introduction-warm-marked-norm-gain}
\end{equation}
The same inequality holds with $C_{b,M}$ replaced by $D_b^M$ and the factor $1/(M+1)$ removed from the right-hand side.
\end{theorem}

The product in~\eqref{eq:introduction-warm-marked-norm-gain} separates the three effects that determine the gain. The factor $\lambda_\varphi$ is the minimum fraction of the initial amplitude retained under one application of $D$. The factor $(1-L_\varphi)^tg$ tracks the objective value retained after $t$ steps of $K_\varphi$, while $E_b$ is the normalization cost. If $M=\Theta(n)$, $L_\varphi=O(1/n)$, and every ratio in the product exceeds one by a constant, then the output marked amplitude exceeds the initial marked amplitude by an exponential factor, up to the factor $1/(M+1)$. The same product bound holds for the single power $D_b^M$, so averaging is not needed for this theorem. Averaging remains useful for the point-dependent walk lengths in Theorem~\ref{thm:normalized-amplitude-gain-informal}.

\paragraph{Conditioned initial states}
The conditioned states used in the preceding subsection provide an important special case. For a nonempty set $A$, let $|A\rangle$ be the coherent encoding of $\pi(\cdot\mid A)$. Then
\begin{equation}
    \lambda_A=\min_{x\in A}P(x,A).
\end{equation}
If $P(x,A)>0$ for every $x\in A$, then
\begin{equation}
    K_A(x,y)=\frac{P(x,y)}{P(x,A)}\mathbf1_{\{y\in A\}}.
    \label{eq:introduction-conditioned-warm-state}
\end{equation}
Thus $\lambda_A$ is the minimum one-step probability of remaining in $A$, and $K_A$ is the original chain conditioned to remain in $A$ for that step. The original mixer can therefore be retained instead of constructing a new reversible chain with stationary distribution $\pi(\cdot\mid A)$. To apply Theorem~\ref{thm:warm-start-gain}, one must have $P(x,A)>0$ for every $x\in A$ and verify $(K_AG)(x)\geq(1-L_A)G(x)$ on $A$ for some $L_A\in[0,1]$. By contrast, Theorem~\ref{thm:normalized-amplitude-gain-informal} can apply to a threshold state with vanishing retention when its threshold-containment and global regularity assumptions hold.

\paragraph{Avoiding ground-state preparation}
The use of bounded powers rather than ground-state projection is a second advantage of quantum tilted walks. The following inequality first shows why nonstationary initial states are particularly relevant when the mixer itself has a small gap. Suppose, for some $\delta>0$, that the largest eigenvalue of $D$ on $|\sqrt\pi\rangle^\perp$ is at most $1-\delta$. Then
\begin{equation}
    1-|\langle\sqrt\pi|\varphi\rangle|^2
    \leq\frac{1-\lambda_\varphi}{\delta}.
    \label{eq:introduction-warm-gap-tradeoff}
\end{equation}
If $\delta$ is bounded away from zero, a state with $\lambda_\varphi$ close to one must also have large overlap with the stationary state. A state can have small stationary overlap while $\lambda_\varphi$ remains close to one only when this top-eigenvalue gap is small. This is precisely the regime in which requiring preparation of the mixer ground state would discard the information contained in the chosen initial state.

Theorem~\ref{thm:warm-start-gain} also imposes no lower bound on the spectral gap of the tilted Hamiltonian. A block encoding of $C_{b,M}$ uses $O(M)$ calls to a block encoding of $D_b$, independent of this spectral gap. The output need not approximate the ground state of the tilted Hamiltonian; only its amplitude on $S$ enters the runtime. In contrast, the implementation of the generalized short-path algorithm by Chakrabarti et al.~\cite{chakrabarti2025generalizedshortpathalgorithms} uses ground-state reflections whose analyzed cost depends on the inverse spectral gap of the tilted Hamiltonian.

\paragraph{Choosing the tilt without empirical tuning}
The absence of ground-state projection provides an additional advantage. In the two applications above, the analysis determines $b^*$ and $M$ explicitly from bounds on the objective and mixer. More generally, suppose the analysis identifies a polynomial-size candidate set containing a useful pair $(b,M)$ and a common running-time bound sufficient for that pair. The algorithm may run every pair up to this bound, evaluate the returned objective values, and retain the best solution. An unsuccessful pair may fail to increase the marked amplitude, but it cannot exceed the common budget. The useful pair suffices for the overall procedure to succeed, so enumeration changes the stated running time by only a polynomial factor. For tilted walks, the cost of each trial is fixed by $M$ and the common cutoff without requiring an estimate of the spectral gap of the tilted Hamiltonian. Certifying a generalized short-path trial at a particular value of $b$ additionally requires a lower bound on the spectral gap of the corresponding tilted Hamiltonian. This advantage is particularly important when running these algorithms in practice. It has been observed numerically~\cite{Dalzell_2023,chakrabarti2025generalizedshortpathalgorithms}, that short-path algorithms provide the best performance when $b$ is chosen to be significantly larger than that predicted by the theory. We expect the same to be true of tilted walks. The advantage of the tilted walk framework is that we can obtain the performance corresponding to the best choice of $b$ with only polynomial overhead, without the need for expensive hyperparameter tuning or other heuristics.

\subsubsection{An example with exponentially small spectral gaps}
We now use the new amplitude-gain guarantee to construct an example that combines the two main advantages discussed above. The initial state has squared overlap $2^{-n}$ with the stationary state of the mixer, yet Theorem~\ref{thm:warm-start-gain} proves exponential amplitude gain using $M=\Theta(n)$. At the same time, the spectral gap of the tilted Hamiltonian is exponentially small and therefore violates the inverse-polynomial gap condition used in the analysis of the generalized short-path algorithm by Chakrabarti et al.~\cite{chakrabarti2025generalizedshortpathalgorithms}. The example thus establishes a concrete separation between the analytical guarantees for quantum tilted walks and generalized short-path algorithms built from the same mixer.

The construction assigns separate roles to a sector register and a work register. Transitions between sector states are exponentially weak, while the work register carries the objective and the successful set. The initial state is confined to one sector and is uniform on the work register. It has exponentially small stationary overlap, but one mixer step decreases its amplitudes by only an exponentially small amount. Because the objective acts only on the work register, tilting does not remove the bottleneck between sectors.

\begin{theorem}[Informal example with a nonstationary initial state]
For every sufficiently large $n$ divisible by $4$, there are a reversible-chain discriminant $D_n$, an objective $G_n$, a nonnegative initial state $|\varphi_n\rangle$, and a successful set $S_n$ with the following properties. The initial state satisfies
\begin{equation}
    |\langle\sqrt{\pi_n}|\varphi_n\rangle|^2=2^{-n},
    \qquad
    \lambda_{\varphi_n}=1-O(2^{-n}/n),
\end{equation}
the spectral gap of the mixer $D_n$ equals $2^{-n}/n$, and the spectral gap of the tilted Hamiltonian $D_n+\frac12G_n$ is at most $2^{-n}/n$. Nevertheless, for $M=\lfloor n/20\rfloor$ and a constant $\Gamma_0>0$,
\begin{equation}
    \left\|\Pi_{S_n}C_{1/2,M}|\varphi_n\rangle\right\|
    \geq
    2^{\Gamma_0n-o(n)}
    \left\|\Pi_{S_n}|\varphi_n\rangle\right\|.
\end{equation}
\end{theorem}
For the minimization objective $H_n=-G_n$ and filter parameter $\eta=1/2$, the generalized short-path Hamiltonian with tilt $1/4$ satisfies
\begin{equation}
    D_n+\frac12G_n=-H_{1/4}^{\mathrm{sp}}+\frac14I.
\end{equation}
The two operators therefore have the same spectral gap. Within the occupied sector, the work-register walk satisfies~\eqref{eq:introduction-warm-score-regularity} with $L_{\varphi_n}=O(1/n)$, so Theorem~\ref{thm:warm-start-gain} proves exponential amplitude gain using only $M=\Theta(n)$ powers. For the same mixer and objective, the exponentially small spectral gap of the tilted Hamiltonian prevents the analysis of Chakrabarti et al.~\cite{chakrabarti2025generalizedshortpathalgorithms} from certifying efficient ground-state reflection.

Equation~\eqref{eq:introduction-warm-gap-tradeoff} also shows that the exponentially small mixer gap is unavoidable for a state with the displayed stationary overlap and value of $\lambda_{\varphi_n}$. The example therefore separates the combined sufficient conditions of the two analyses rather than the condition on the spectral gap of the tilted Hamiltonian alone. We note that the goal of this example is to contrast the tilted-walk and short-path frameworks. An element of $S_n$ can be written down directly, so this is not a quantum speedup for a hard optimization problem. Proposition~\ref{prop:metastable-warm-separation} gives more details about this construction and proves each of the above claims.

\subsection{Discussion}
\label{sec:discussion}

\paragraph{Comparison to Classical Algorithms and the Possibility of Further Dequantization} In this paper, we do not focus on computing the exact numerical values of the exponent in our super-quadratic speedups and instead only show that the exponent is lower-bounded by a constant independent of the problem size. The reasons for this are two-fold: firstly, the actual exponents of the speedup from short-path algorithms tend to be very small. However, it has been observed numerically that the true speedup obtained when the value of $b$ is appropriately tuned can be much larger than the analytical prediction~\cite{Dalzell_2023,chakrabarti2025generalizedshortpathalgorithms}. As we have already argued, the tilted walks algorithm can be augmented to obtain the performance of the best $b$ out of a number of possibilities, at an overhead that is proportional to the number of possibilities. Determining the value of the exponent with this ideal choice of $b$ requires a significant numerical study, which we defer to future work.

Secondly, our main focus in this paper is on the mechanism for quantum speedup afforded by the quantum tilted walk algorithm. We note that the quantum algorithms presented here may not offer a super-quadratic speedup over \emph{all} classical algorithms. In fact, for \textsc{MAX-$k$-CSP}, classical conditioning-and-search is not the best known classical algorithm. The dequantization of Le Gall and Tamaki is notable in part because it was based on the same spectral density assumption used by Dalzell et al to obtain a super-quadratic speedup over unstructured search. It therefore raised the possibility that the apparent mechanism for quantum speedup was fully exploitable classically. We show that this is not the case for quantum tilted walks, since we can offer a super-quadratic speedup over the classical algorithm that is designed to exploit the spectral density condition. In fact, the spectral density condition is never assumed in our analysis as it is a provable consequence of our pseudo-Lipschitzness assumption. In fact, a more important role is played by assumptions that govern the \emph{smoothness} of the objective function under the action of the mixer. This includes the pseudo-Lipschitzness assumption in Theorem~\ref{thm:normalized-amplitude-gain-informal} but such assumptions are also required for Theorem~\ref{thm:warm-start-gain-informal}. We therefore propose that smoothness of the objective function under the action of the mixer is in fact the central mechanism driving the speedup offered by this class of algorithms.

Finally, the authors of~\cite{chakrabarti2025generalizedshortpathalgorithms} identify a setting where the generalized-short path algorithm can be shown to offer a super-quadratic speedup over all known algorithms for a particular class of maximum independent set problems. There is currently no conditioning-and-search version of the classical Markov Chain search algorithm analyzed in~\cite{chakrabarti2025generalizedshortpathalgorithms}, although it is reasonable that one can be derived. If such an algorithm is found, the analysis in this paper offers a natural path to a super-quadratically accelerated quantum version. The identification of such an algorithm and the analysis of the corresponding tilted walk are natural open questions for future work.

\paragraph{Future Work}
Several other questions remain open. The first is to find hard optimization families with useful nonstationary initial states, such as states produced by classical approximation algorithms or states localized in metastable regions. The second is to improve the upper bound on $E_b$. Our present bound uses a log-Sobolev inequality and one-step quadratic variation; problem-specific spectral structure may yield a stronger normalization certificate, especially when the global chain mixes slowly. It is also an interesting direction to identify different structured problem classes where the tilted walk framework yields larger advantage over the classical baselines.

\subsection{Related work}
Short-path algorithms are the main inspiration for this work. Hastings~\cite{hastings2018shortpathquantumalgorithm,hastings2018weaker,hastings2019toy} introduced the framework and subsequently studied weaker assumptions and a toy model. Dalzell et al.~\cite{Dalzell_2023} introduced a filtered objective and proved constant-exponent super-Grover speedups for several families of binary optimization problems. Chakrabarti et al.~\cite{chakrabarti2025generalizedshortpathalgorithms} replaced the hypercube mixer by discriminant operators of reversible Markov chains, allowing the quantum runtime to be compared with search from problem-dependent stationary distributions. Quantum tilted walks use the same class of tilted operators but analyze the amplitude produced by $C_{b,M}$ directly rather than preparing the ground state of the tilted Hamiltonian.

Quantum-walk search provides a second line of related work. Continuous-time spatial search combines a graph walk Hamiltonian with a marking potential \cite{childs2004spatial}. The frameworks of Szegedy~\cite{szegedy2004quantum} and Magniez et al.~\cite{magniez2011search} quantize reversible-chain search in terms of hitting, setup, update, and checking costs. Apers et al.~\cite{apers_et_al:LIPIcs.STACS.2021.6} extended quantum-walk search to arbitrary initial distributions, while Apers and Sarlette~\cite{apers2019fastforwarding} used polynomial transformations of a discriminant operator to accelerate transient dynamics. Related quantum sampling and simulated-annealing algorithms obtain quadratic improvements in mixing or spectral-gap dependence \cite{somma2008annealing,wocjan2008speedup}.

Le Gall and Tamaki~\cite{gall2026dequantizingshortpathquantumalgorithms} introduced conditioning-and-search by extracting an explicit threshold-set and local-search mechanism from the analysis of previous short-path algorithms. Their classical algorithm and its standard quadratic quantum speedup are the central baselines for our concrete applications. Quantum tilted walks insert the transformation $C_{b,M}$ between threshold-state preparation and local search. The resulting amplitude gain is the source of the additional exponent $\kappa$.

\subsection{Organization}
Section~\ref{sec:tilted-walk-algorithm} defines the tilted-walk framework, its assumptions, and the access model. Section~\ref{sec:tilted-walk-analysis} proves the amplitude gain, controls the normalization factor (top eigenvalue of the tilted operator) through a log-Sobolev inequality, and balances these two effects. Section~\ref{sec:application-to-optimization} converts these results into an exact-optimization algorithm and applies it to \textsc{MAX-E$k$-LIN2} and \textsc{MAX-$k$-CSP}. Section~\ref{sec:warm-starts} proves an amplitude bound for nonstationary initial states and analyzes an example with an exponentially small mixer gap. Section~\ref{sec:implementation-details} constructs the required block encodings, accounts for approximation errors, and resolves the unknown normalization and optimum values.

\section{The Quantum Tilted Walk Framework}
\label{sec:tilted-walk-algorithm}
The introduction described the quantum tilted-walk algorithm and stated its main amplitude guarantee. We now fully describe the algorithm and state the assumptions used in the amplitude analysis. We formally define the normalized objective $G$, the tilted operator $D_b$, the operator $C_{b,M}$, and the coherent encodings of conditioned stationary distributions. We then state the assumptions on one-step objective variation, pseudo-Lipschitzness, the log-Sobolev constant of the mixer, and quantum access to the objective and transition rule, that are required for the analysis of the algorithm.

Let $\mathcal{X}=\{0,1\}^n$ index the computational basis, and let $P$ be the transition matrix of a reversible random walk on $\mathcal{X}$ with stationary distribution $\pi$, and let $D=D(P)$ be its discriminant operator, so that
\begin{equation}
    D_{x,y}=\frac{\sqrt{\pi(x)}}{\sqrt{\pi(y)}}P(x,y)
    =\sqrt{P(x,y)P(y,x)}.
\end{equation}
Let $H_c$ be the diagonal operator encoding the objective, i.e., $H_c|z\rangle=f(z)|z\rangle$. Choose a shift $A_c\in\mathbb R$ and a scale $R_c>0$ such that
\begin{equation}
    \widehat G:=\frac{H_c+A_cI}{R_c}
\end{equation}
satisfies $\widehat G(z)\in[0,1]$ and $\max_z\widehat G(z)=1$, and define
\begin{equation}
    G:=\frac{I+\widehat G}{2}.
    \label{eq:shifted-objective}
\end{equation}
Thus $G(z)\in[1/2,1]$, $\max_zG(z)=1$, and $G$ has the same maximizers as $H_c$, equivalently as the objective function $f$. For $b\geq0$, we define the tilted operator
\begin{equation}
    H_b=D+bG
\end{equation}
and let $E_b=\lambda_{\max}(H_b)$. Since $H_b$ is a nonzero entrywise nonnegative matrix, $E_b>0$, and a corresponding eigenvector may be chosen entrywise nonnegative. We define the normalized tilted operator $D_b=H_b/E_b$. Although $D$ is the discriminant of a Markov chain, $D_b$ is not generally a stochastic matrix. Fix $b\geq0$ and an integer $M\geq0$, and define the operator
\begin{equation}
    C_{b,M}=\frac{1}{M+1}\sum_{m=0}^{M}D_b^m,
    \qquad
    D_b=\frac{D+bG}{E_b}.
\end{equation}
so that the operator $C_{b,M}$ has spectral norm at most $1$. For any nonempty set $T\subseteq\mathcal X$, define
\begin{equation}
    \pi(T):=\sum_{z\in T}\pi(z),
    \qquad
    |T\rangle:=\frac{1}{\sqrt{\pi(T)}}
    \sum_{z\in T}\sqrt{\pi(z)}\,|z\rangle.
\end{equation}
In particular, this defines $|T_\eta\rangle$; for the hypercube walk it is the uniform superposition over $T_\eta$. The application theorems specify $b$ and $M$ explicitly. If the analysis instead identifies a polynomial-size set of candidate pairs, the algorithm can run every pair with a fixed time limit and retain the best verified output. For a fixed pair, it applies
\begin{equation}
     C_{b,M} |T \rangle
\end{equation}
and then projects onto a target quantum state $|S\rangle$. Depending on the application, the resulting state may be processed further. When $S$ is a set of points where $G$ achieves relatively larger values than $T$, the overlap 
\begin{equation}
|\left\langle S| C_{b,M} |T \right\rangle|
\end{equation}
can be bigger than $|\langle S|T\rangle|>0$ by a nontrivial factor since the operator $C_{b,M}$ tilts the state onto higher energy subspace of $G$. To understand why we should expect such a nontrivial amplitude gain, consider the hypercube specialization and take $S=\{x^{\star}\}$ and $|T\rangle=|+\rangle$. Consider a single term $D_b^m$ in the sum
\begin{align}
    \langle x^*|D_b^m|+\rangle
    &=\frac{1}{2^{n/2}}\frac{1}{E_b^m}\sum_{s=0}^m\sum_{\substack{p_0,\ldots,p_s\geq0\\p_0+\cdots+p_s=m-s}} b^s 
    \sum_{y, y_1,\ldots,y_s\in\{0,1\}^n}
    P^{p_0}(x^{\star},y_1)G(y_1)P^{p_1}(y_1,y_2)G(y_2)\cdots \\
    &\hspace{4cm}\cdots P^{p_{s-1}}(y_{s-1},y_s)G(y_s)P^{p_s}(y_s,y) 
\end{align}
Since there are exponentially many paths in this sum, it is possible that for problems that are relatively stable under each application of $P$, each path contributes a non-negligible amount to the sum possibly, giving an exponentially large improvement. However, the catch is that there is an $E_b^m$ term that is sitting in the denominator which is exponentially large as well when $m=\text{poly}(n)$. Therefore, the structure should give a non-vacuous tradeoff where the amplitude gain is not dominated by the energy term. For certain cases, in fact it is possible to show that 
\begin{equation}
|\left\langle x^{*}| C_{b,M} |+ \right\rangle| = \Omega^{*}(2^{-(1-c)n/2}).
\end{equation}
Then, quantum maximum finding over the normalized state $C_{b,M}|T\rangle$ gives a super-quadratic speedup over brute force search for a choice of $b=b^{*}$ and $M$. 

For this choice of $S$ and $T$, our algorithm should recover the amplitude gain by the original short path algorithm. In fact a single term in this operator approximates the projector onto the ground state of the tilted Hamiltonian analyzed in~\cite{Dalzell_2023} for large $m$. However, for more general $T$, a single term in the expansion is not sufficient since coupling between every point in $S$ and $T$ requires a  different value of $m$. On the other hand, we cannot set $m$ to a very large value, since a large $m$ value means the walk operator is applied too many times and most points in the expansion will be distributed according to the stationary distribution of $P$.

Since the amplitude gain is not possible for completely unstructured problems, we make the following assumptions on the objective function. We note that these assumptions are proved explicitly for the problem families considered in this paper. 

More specifically, in Theorems~\ref{thm:tilted-walk-eklin2} and~\ref{thm:tilted-walk-maxkcsp}, we derive the required bound on $\mathcal L_2(f)$ and then obtain the corresponding bound for the constructed operator $G$. In that case, the log-Sobolev constant $\omega=1/n$ is well known for the hypercube walk. The same properties are also used in the analysis of  previous quantum short path algorithms and also classical conditioning-and-search.

We first use a bound on the one-step quadratic variation of the input objective $f$:
\begin{condition}[One-step quadratic variation]
    \label{cond:quadratic-variation}
Let $\mathcal L_2(f)$ be an upper bound on the one-step quadratic variation of $f$, so that
\begin{equation}
    \max_{x\in\mathcal X}\sum_{y\in\mathcal X}
    P(x,y)\bigl(f(x)-f(y)\bigr)^2
    \leq \mathcal L_2(f)^2.
\end{equation}
We assume that $0<\mathcal L_2(f)=O(1)$.
\end{condition}

{

\begin{definition}[Pseudo-Lipschitzness]
We say that $G$ is pseudo-Lipschitz with constant $L\in[0,1]$ if
\begin{equation}
    (PG)(x)\geq(1-L)G(x)
    \qquad\text{for every }x\in\mathcal X.
\end{equation}
\end{definition}

\begin{lemma}[Regularity after normalization]
    \label{lem:normalization-regularity}
Let $f$ satisfy Condition~\ref{cond:quadratic-variation}, and let $G$ be defined by~\eqref{eq:shifted-objective}. Assume that the optimum of $f$ is extensive:
\begin{equation}
    H_{\max}:=\max_{x\in\mathcal X}f(x)=\Theta(n).
\end{equation}
Define
\begin{equation}
    \mathcal L_2(G):=\frac{\mathcal L_2(f)}{2R_c},
    \qquad
    L:=2\mathcal L_2(G)=\frac{\mathcal L_2(f)}{R_c}.
\end{equation}
Then $\mathcal L_2(G)=O(1/n)$ is an upper bound on the one-step quadratic variation of $G$, and $G$ is pseudo-Lipschitz with constant $L=O(1/n)$.
\end{lemma}
\begin{proof}
For the affine normalization in~\eqref{eq:shifted-objective}, the extensive optimum implies
\begin{equation}
    R_c=H_{\max}+A_c=\Theta(n).
\end{equation}
Since
\begin{equation}
    G(x)-G(y)=\frac{f(x)-f(y)}{2R_c},
\end{equation}
the stated bound on $\mathcal L_2(G)$ follows directly from Condition~\ref{cond:quadratic-variation} and $R_c=\Theta(n)$. For every $x\in\mathcal X$, Cauchy--Schwarz gives
\begin{align}
    G(x)-(PG)(x)
    &\leq
    \left(\sum_yP(x,y)(G(x)-G(y))^2\right)^{1/2}\\
    &\leq\mathcal L_2(G).
\end{align}
Since $G(x)\geq1/2$, it follows that
\begin{equation}
    (PG)(x)\geq G(x)-\mathcal L_2(G)
    \geq(1-L)G(x).
\end{equation}
Finally, Condition~\ref{cond:quadratic-variation} and $R_c=\Theta(n)$ give $\mathcal L_2(G)=O(1/n)$ and $L=O(1/n)$.
\end{proof}
\par}
Section~\ref{subsec:resolving-unknowns} discusses the unknown optimum and normalization parameters. We will show that the problems we consider satisfies this condition when $D$ is chosen to be Pauli mixer.
\begin{condition}[Log-Sobolev Inequality]
\label{cond:log-Sobolev}
We assume that $D(P)$ satisfies the following log-Sobolev inequality with constant $\omega>0$: for every unit vector $|\psi\rangle$, if
\begin{equation}
    \mu_\psi(z):=|\psi(z)|^2,
    \qquad
    r_\psi:=\frac{\mu_\psi}{\pi},
\end{equation}
then
\begin{equation}
    \langle \psi|D|\psi\rangle
    \leq
    1-\omega\operatorname{Ent}_\pi(r_\psi),
\end{equation}
where $\pi$ is the stationary distribution of $P$. 
\end{condition}
This assumption is used to control the maximum eigenvalue $E_b$ in the denominator. While this can be seen as an assumption for general walk operators, for the hypercube walk this assumption holds true with $\omega = \frac{1}{n}$. We make the following assumption on the access model.
\begin{assumption}[Access model]
\label{ass:access-model}
The algorithm has coherent access to an exact binary-value oracle for $f$ and to the quantum update oracle of $P$. More precisely, controlled versions of the following unitaries and their adjoints are available:
\begin{equation}
\begin{aligned}
    O_f|x\rangle|0^p\rangle
        &=|x\rangle|f(x)\rangle,\\
    U_P|x\rangle|0\rangle
        &=\sum_{y\in\mathcal X}\sqrt{P(x,y)}\,|x\rangle|y\rangle,
\end{aligned}
\label{eq:primitive-access-oracles}
\end{equation}
where $|f(x)\rangle$ is an exact signed $p$-bit binary representation and $p=\operatorname{poly}(n)$.
\end{assumption}

Throughout the paper, we will assume that we have the ideal block-encoding of $C_{b,M}$ with no error and each step of the algorithm can be done ideally with no error. Furthermore, we assume the quantities $E_b, H_{\text{max}}$ are known to the algorithm.  We do this to ease the presentation but Section~\ref{sec:implementation-details} gives the construction, normalization, and error analysis that results from the actual implementation. In practice, $E_b$ and  $H_{\text{max}}$ are not known in advance. We show in~\ref{sec:implementation-details}, this does not pose a problem for the algorithm.  

The next section uses these assumptions to prove the amplitude bound. The positive path expansion supplies a lower bound before normalization, while the log-Sobolev inequality supplies an upper bound on $E_b$.

\section{Amplitude Gain from Quantum Tilted Walks}
\label{sec:tilted-walk-analysis}
The path expansion in Section~\ref{sec:tilted-walk-algorithm} identifies two quantities that determine the normalized amplitude. The numerator depends on the nonnegative path contributions to $(D+bG)^m$, while the denominator contains the factor $E_b^m$. We bound these quantities separately. First, pseudo-Lipschitzness yields a lower bound on the amplitude before normalization. Next, the log-Sobolev inequality and one-step quadratic variation yield an upper bound on $E_b$. We then combine the two estimates, choose $b$, and average over powers so that each configuration in the successful set contributes at an appropriate walk length.

\subsection{Lower Bound on the Amplitude Before Normalization}
For convenience, we make the following definition:
\begin{definition}[$\eta$-filter]
For $\eta\in(0,1)$, define the filter $\bar{g}_\eta:[0,1]\to[0,1]$ by
\begin{equation}
\bar{g}_\eta(u):=
\begin{cases}
0, & 0\leq u\leq 1-\eta,\\
\dfrac{u-1+\eta}{\eta}, & 1-\eta<u\leq 1,
\end{cases}
\end{equation}
This function is nonnegative, nondecreasing, and convex. Moreover, for $T_\eta=\{z:G(z)\geq 1-\eta\}$, we have
\begin{equation}
    \bar{g}_\eta(G(z))\leq \mathbf{1}_{\{z\in T_\eta\}}.
\end{equation}
\end{definition}

As opposed to the algorithm presented in~\cite{Dalzell_2023,chakrabarti2025generalizedshortpathalgorithms}, we will not use any filter in our algorithm and this will only be used for analysis purposes.

We start with the following theorem:
\begin{theorem}
\label{thm:amplitude-gain}
Assume that $G$ is pseudo-Lipschitz with constant $L\in[0,1]$. Let $T\subseteq\mathcal X$ be nonempty, let $\eta\in(0,1)$ satisfy $T_\eta\subseteq T$, and let $S\subseteq \mathcal X$ be nonempty. Then for every integer $m\geq 0$ and every $b\geq0$,
\begin{equation}
|\langle T| (D + bG)^m |S \rangle| \geq
\sqrt{\frac{\pi(S)}{\pi(T)}}
\sum_{s=0}^m K_{m,s}(L) b^s \mathbb{E}_{x\sim \pi_S}\left[G(x)^s\bar{g}_\eta\left((1-L)^{m-s}G(x)\right)\right].
\end{equation}
where 
\begin{equation}
    K_{m,s}(L):= \begin{cases}
 \prod_{j=1}^s \frac{1-(1-L)^{m-s+j}}{1-(1-L)^j} & \text{if } 0< L\leq 1,\\
\binom{m}{s} & \text{if } L=0.
\end{cases}
\end{equation}
and $\pi_S(x)=\pi(x)/\pi(S)$ for $x\in S$ is the stationary distribution conditioned on $S$.
\end{theorem}
\begin{proof}
We expand $(D + bG)^m$ as a sum over all words $w\in \{bG,D\}^m$ where $w = w_1 w_2 \ldots w_m$ and $w_i \in \{bG,D\}$. That is, we can write
\begin{equation}
    \langle T |(D + bG)^m |S \rangle = \sum_{w\in \{bG,D\}^m} \langle T| w|S \rangle
\end{equation}
Fix a word $w$ with $s$ occurrences of $bG$. More explicitly, write the word from right to left as
\begin{equation}
    w=D^{p_s}(bG)D^{p_{s-1}}\cdots D^{p_1}(bG)D^{p_0},
    \qquad \sum_{i=0}^s p_i=m-s,
\end{equation}
and set $r_j=p_0+p_1+\cdots+p_{j-1}$ for $1\leq j\leq s$. Reversibility gives
\begin{equation}
    \langle y|D^t|x\rangle
    =\sqrt{\frac{\pi(x)}{\pi(y)}}P^t(x,y).
\end{equation}
Thus the square-root factors telescope when the matrix element is expanded in the computational basis, giving
\begin{align}
    \langle T|w|x\rangle
    &=b^s\sqrt{\frac{\pi(x)}{\pi(T)}}
    \sum_{y_1,\ldots,y_s\in\mathcal X}\sum_{y\in T}
    P^{p_0}(x,y_1)G(y_1)P^{p_1}(y_1,y_2)G(y_2)\cdots \\
    &\hspace{4cm}\cdots P^{p_{s-1}}(y_{s-1},y_s)G(y_s)P^{p_s}(y_s,y) \\
    &=b^s\sqrt{\frac{\pi(x)}{\pi(T)}}
    \mathbb{E}_x\left[\prod_{j=1}^s G(Z_{r_j})\mathbf{1}_{\{Z_{m-s}\in T\}}\right].
\end{align}
Since $T_\eta\subseteq T$, we have $\bar{g}_\eta(G(y))\leq\mathbf{1}_{\{y\in T_\eta\}}\leq\mathbf{1}_{\{y\in T\}}$ for every $y$. We can therefore lower bound the indicator by the filtered $G$ operator:
\begin{equation}
    \langle T|w|x\rangle
    \geq b^s\sqrt{\frac{\pi(x)}{\pi(T)}}
    \mathbb{E}_x\left[\prod_{j=1}^s G(Z_{r_j})\bar{g}_\eta(G(Z_{m-s}))\right].
\end{equation}
Here the expectation is over the random walk $Z_t$. We now apply the bound $(PG)(z)\geq(1-L)G(z)$ iteratively. Let $q=1-L$. Since $P^tG\geq q^tG$ pointwise, Jensen's inequality and the monotonicity and convexity of $\bar{g}_\eta$ give
\begin{equation}
    P^t(\bar{g}_\eta(G))(x)\geq \bar{g}_\eta(P^tG(x))\geq \bar{g}_\eta(q^tG(x)).
\end{equation}
More generally, for every integer $a\geq0$ and every $c\in[0,1]$, the function $u\mapsto u^a\bar{g}_\eta(cu)$ is nonnegative, nondecreasing, and convex on $[0,1]$, so
\begin{equation}
    P^t\left(G^a\bar{g}_\eta(cG)\right)(x)
    \geq (q^tG(x))^a\bar{g}_\eta(cq^tG(x)).
\end{equation}
Applying this bound successively to the nested expectation gives
\begin{equation}
    \mathbb{E}_x\left[\prod_{j=1}^s G(Z_{r_j})\bar{g}_\eta(G(Z_{m-s}))\right]
    \geq q^{I_w}G(x)^s\bar{g}_\eta(q^{m-s}G(x)),
\end{equation}
where $I_w=\sum_{j=1}^sr_j$. Hence,
\begin{align}
    |\langle T| (D + bG)^m |x \rangle|
    &\geq \sqrt{\frac{\pi(x)}{\pi(T)}}
    \sum_{s=0}^m \sum_{w: |w|_G = s} b^s (1-L)^{I_w} G(x)^s\bar{g}_\eta((1-L)^{m-s}G(x)).
\end{align}
Averaging over $x\sim \pi_S$ gives 
\begin{align}
    |\langle T| (D + bG)^m |S \rangle|
    &\geq \sqrt{\frac{\pi(S)}{\pi(T)}}
    \sum_{s=0}^m \sum_{w: |w|_G = s} b^s (1-L)^{I_w} \mathbb{E}_{x\sim \pi_S}\left[G(x)^s\bar{g}_\eta((1-L)^{m-s}G(x))\right]
\end{align}
Combining with~\cref{lem:inversion-count} proves the theorem.
\end{proof}
The following lemma bounds the number $K_{m,s}(L)$ in a word $w$ with $s$ occurrences of $bG$ used in the theorem above.

\begin{lemma}
    \label{lem:inversion-count}
    For a word $w$, read the symbols from right to left, i.e., in the order in which they act on $|x\rangle$. Let $r_j$ be the number of $D$ steps that have occurred before the $j$-th occurrence of $bG$ in this order. Then
    \begin{equation}
    \sum_{w:|w|_G=s}(1-L)^{I_w}= \sum_{w:|w|_G=s}\prod_{j=1}^s (1-L)^{r_j} = K_{m,s}(L)
    \end{equation}
    where $K_{m,s}(L)$ is defined as
\begin{equation}
    K_{m,s}(L):= \begin{cases}
 \prod_{j=1}^s \frac{1-(1-L)^{m-s+j}}{1-(1-L)^j} & \text{if } 0< L\leq 1,\\
\binom{m}{s} & \text{if } L=0.
\end{cases}
\end{equation}
\end{lemma}
\begin{proof}
    If the final symbol in this right-to-left order is $bG$, then the preceding $m-1$ symbols contain $s-1$ copies of $bG$ and $m-s$ copies of $D$. The contribution in this case is
    \begin{equation}
    (1-L)^{m-s} K_{m-1,s-1}(L)
    \end{equation}
    since this final $bG$ contributes a factor of $(1-L)^{m-s}$ and the preceding $m-1$ symbols contribute $K_{m-1,s-1}$. If the final symbol in this right-to-left order is $D$, the contribution is simply $K_{m-1,s}$. Hence, we have the recursion
    \begin{equation}
    K_{m,s}(L) = (1-L)^{m-s} K_{m-1,s-1}(L) + K_{m-1,s}(L)
    \end{equation}
    where $K_{m,0}(L) = 1$ and $K_{m,m}(L) = 1$. Instead of solving the recurrence directly, we use induction. Set $q=1-L$ and define
    \begin{equation}
        F_{m,s}:=\prod_{j=1}^s\frac{1-q^{m-s+j}}{1-q^j}.
    \end{equation}
    The boundary conditions hold since $F_{m,0}=1$ by the empty product convention and
    \begin{equation}
        F_{m,m}=\prod_{j=1}^m\frac{1-q^j}{1-q^j}=1.
    \end{equation}
    Assume inductively that $K_{m-1,t}=F_{m-1,t}$ for all relevant $t$. For $1\leq s\leq m-1$, comparing factors gives
    \begin{equation}
        F_{m-1,s-1}=F_{m,s}\frac{1-q^s}{1-q^m},
        \qquad
        F_{m-1,s}=F_{m,s}\frac{1-q^{m-s}}{1-q^m}.
    \end{equation}
    Therefore the recurrence gives
    \begin{align}
        K_{m,s}
        &=q^{m-s}K_{m-1,s-1}+K_{m-1,s} \\
        &=q^{m-s}F_{m-1,s-1}+F_{m-1,s} \\
        &=F_{m,s}\left(q^{m-s}\frac{1-q^s}{1-q^m}+\frac{1-q^{m-s}}{1-q^m}\right) \\
        &=F_{m,s}\frac{q^{m-s}-q^m+1-q^{m-s}}{1-q^m}
        =F_{m,s}.
    \end{align}
    This proves the claimed closed form for $L\neq 0$. When $L=0$, the same induction reduces to Pascal's recurrence with $K_{m,0}=K_{m,m}=1$, and hence gives $K_{m,s}=\binom{m}{s}$.
\end{proof}

\subsection{Upper Bound on the Maximum Eigenvalue}
The preceding subsection lower-bounds the numerator of the normalized matrix element. We now upper-bound $E_b$, which controls the normalization cost accumulated by each power of $D_b$.

\begin{lemma}[Concentration from the log-Sobolev inequality]
\label{lem:lsi-concentration}
Let $P$ be a reversible Markov chain on $\mathcal X$, let $D(P)$ be its discriminant matrix, and let $\pi$ be its stationary distribution. Let $\mathcal L_2(G)$ be any upper bound on the one-step quadratic variation of $G$. Assume that $D$ satisfies the log-Sobolev inequality with constant $\omega$. Then, for every $\lambda\geq0$,
\begin{equation}
    \log\mathbb{E}_\pi\left[e^{\lambda(G-\mathbb{E}_\pi[G])}\right]
    \leq \frac{\lambda^2\mathcal{L}_2(G)^2}{8\omega}
    \leq \frac{\lambda^2\mathcal{L}_2(G)^2}{2\omega},
\end{equation}
\end{lemma}
\begin{proof}
For a nonnegative function $h$, apply the assumed inequality to the normalized vector with amplitudes proportional to $\sqrt{\pi(z)}h(z)$. By homogeneity,
\begin{equation}
    \operatorname{Ent}_\pi(h^2)
    \leq \frac{1}{\omega}\mathcal{D}_P(h,h),
    \qquad
        \mathcal{D}_P(h,h)
        :=\langle h,(I-P)h\rangle_\pi
        =\frac12\sum_{x,y\in\mathcal X}\pi(x)P(x,y)
            \left(h(x)-h(y)\right)^2.
\end{equation}
Set $h(z)=e^{\lambda G(z)/2}$. For all real $u,v$,
\begin{equation}
    \left(e^{u/2}-e^{v/2}\right)^2
    \leq \frac{(u-v)^2}{8}\left(e^u+e^v\right),
\end{equation}
which follows by writing the difference as an integral and applying Cauchy--Schwarz and convexity. Therefore, using reversibility,
\begin{align}
    \mathcal{D}_P(h,h)
        &\leq \frac{\lambda^2}{16}
            \sum_{x,y\in\mathcal X}\pi(x)P(x,y)
            \left(G(x)-G(y)\right)^2
            \left(e^{\lambda G(x)}+e^{\lambda G(y)}\right)\\
        &=\frac{\lambda^2}{8}\mathbb{E}_{x\sim\pi}\left[
            e^{\lambda G(x)}\sum_{y\in\mathcal X}P(x,y)
            \left(G(x)-G(y)\right)^2
            \right]\\
    &\leq \frac{\lambda^2\mathcal{L}_2(G)^2}{8}
      \mathbb{E}_\pi\left[e^{\lambda G}\right].
\end{align}
Let $Z(\lambda)=\mathbb{E}_\pi[e^{\lambda G}]$ and $F(\lambda)=\log Z(\lambda)$. Since
\begin{equation}
    \operatorname{Ent}_\pi(e^{\lambda G})
    =Z(\lambda)\left(\lambda F'(\lambda)-F(\lambda)\right),
\end{equation}
the preceding estimates imply
\begin{equation}
    \frac{d}{d\lambda}\left(\frac{F(\lambda)}{\lambda}\right)
    =\frac{\lambda F'(\lambda)-F(\lambda)}{\lambda^2}
    \leq\frac{\mathcal{L}_2(G)^2}{8\omega}.
\end{equation}
Integrating from $0$ to $\lambda$ and using $\lim_{\lambda\to0}F(\lambda)/\lambda=\mathbb{E}_\pi[G]$ proves the claim.
\end{proof}

\begin{theorem}[Maximum Eigenvalue Bound]
\label{thm:energy-bound}
Let $H_b=D+bG$, and let
\begin{equation}
    E_b:=\lambda_{\max}(H_b).
\end{equation}
Let $\mathcal L_2(G)$ be any upper bound on the one-step quadratic variation of $G$. Assume that $D$ satisfies the log-Sobolev inequality with constant $\omega>0$. Then, for any $b\geq 0$,
\begin{equation}
    E_b
    \leq
    1+b\mathbb{E}_\pi[G]+\frac{b^2\mathcal{L}_2(G)^2}{8\omega^2}
\end{equation}
\end{theorem}
\begin{proof}
Set $\mu = \mathbb{E}_\pi[G]$. Let $|\psi_b\rangle$ be a unit eigenvector of $H_b$ with eigenvalue $E_b$. Since $H_b$ is entrywise nonnegative, $|\psi_b\rangle$ may be chosen to be the Perron--Frobenius eigenvector. Set
\begin{equation}
    \mu_b(z):=|\psi_b(z)|^2,
    \qquad
    r_b:=\frac{\mu_b}{\pi}.
\end{equation}
Then $r_b\geq0$ and $\mathbb{E}_\pi r_b=1$. By the Rayleigh principle and the log-Sobolev inequality,
\begin{align}
    E_b
    &=\langle \psi_b|D+bG|\psi_b\rangle \\
    &\leq
    1-\omega\operatorname{Ent}_\pi(r_b)
    +b\mathbb{E}_\pi[r_bG].
\end{align}
We use the entropy (Young's) inequality
\begin{equation}
    \mathbb{E}_\pi[rg]
    \leq
    \operatorname{Ent}_\pi(r)+\log\mathbb{E}_\pi[e^g]
\end{equation}
valid for $r\geq0$ with $\mathbb{E}_\pi r=1$. Applying it with $g=bG/\omega$ gives
\begin{equation}
    b\mathbb{E}_\pi[r_bG]
    =\omega\mathbb{E}_\pi\left[r_b\frac{bG}{\omega}\right]
    \leq
    \omega\operatorname{Ent}_\pi(r_b)
    +\omega\log\mathbb{E}_\pi\left[\exp\left(\frac{bG}{\omega}\right)\right].
\end{equation}
The entropy terms cancel, and therefore
\begin{equation}
    E_b
    \leq
    1+\omega\log\mathbb{E}_\pi\left[\exp\left(\frac{bG}{\omega}\right)\right].
\end{equation}
By~\cref{lem:lsi-concentration}, for every $\lambda\geq0$,
\begin{equation}
    \log\mathbb{E}_\pi\left[\exp\left(\lambda(G-\mu)\right)\right]
    \leq \frac{\lambda^2 \mathcal{L}_2(G)^2}{8\omega}.
\end{equation}
Equivalently,
\begin{equation}
    \mathbb{E}_\pi\left[\exp(\lambda G)\right]
    \leq
    \exp\left(\lambda\mu+\frac{\lambda^2 \mathcal{L}_2(G)^2}{8\omega}\right).
\end{equation}
Taking $\lambda=b/\omega$ yields
\begin{equation}
    \omega\log\mathbb{E}_\pi\left[\exp\left(\frac{bG}{\omega}\right)\right]
    \leq
    b\mu+\frac{b^2\mathcal{L}_2(G)^2}{8\omega^2}.
\end{equation}
Substituting this into the variational bound proves
\begin{equation}
    E_b
    \leq
    1+b\mu+\frac{b^2\mathcal{L}_2(G)^2}{8\omega^2}.
\end{equation}
\end{proof}

\subsection{Combining the Amplitude and Normalization Bounds}
We now combine the lower bound on the unnormalized amplitude with the upper bound on $E_b$. The choice of $b$ balances the increase produced by the objective term against the normalization cost.

Now we can use this result to obtain the final bound for the amplitude gain.
\begin{theorem}
    \label{thm:normalized-amplitude-gain}
        Assume that $G$ is pseudo-Lipschitz with constant $L\in(0,1]$, let $\mathcal L_2(G)>0$ be an upper bound on the one-step quadratic variation of $G$, and assume that Condition~\ref{cond:log-Sobolev} is satisfied. Let $T\subseteq\mathcal X$ be nonempty, and let $\eta,\rho \in (0,1) = \Theta(1)$ satisfy $\eta-\rho-L>0$, $1-\rho-\mathbb{E}_\pi[G]-L>0$, and $T_\eta\subseteq T$. Suppose there exists a nonempty subset $S$ such that $G(x)\geq 1-\rho$ for all $x\in S$. Choose
    \begin{equation}
            M:=\max_{x\in S} \left\lceil \frac{\log(G(x)/(1-\eta))}{L}\right\rceil.
    \end{equation}
    and
    \begin{equation}
        b=b^*:=\min\left\{1,
        \frac{
        \sqrt{1+\frac{8\omega^2(1-\rho-L)(1-\rho-\mathbb{E}_\pi[G] -L)}
        {\mathcal{L}_2(G)^2}}-1}
        {1-\rho-L}
        \right\},
    \end{equation}
    Then,
    \begin{align}
    |\left\langle S\left| C_{b,M} \right|T \right\rangle|
    &\geq
    \left(\frac{\eta-\rho-L}{\eta}\right)\frac{1}{M+1}
    \sqrt{\frac{\pi(S)}{\pi(T)}}
    2^{\kappa n/2}.
    \end{align}
    where 
    \begin{equation}
        \kappa:=\frac{2b^*(1-\rho-\mathbb{E}_\pi[G]-L)}{3nL\log 2}
        \mathbb{E}_{x\sim\pi_S}\left[\log\left(\frac{G(x)}{(1-\rho)}\right)\right].
    \end{equation}
\end{theorem}
\begin{proof}
First, we consider a single point $x\in S$. For any $x\in S$, define
\begin{equation}
    m(x):= \left\lceil \frac{\log(G(x)/(1-\rho))}{L}\right\rceil.
\end{equation}
Since $\eta-\rho-L>0$, every $x\in S$ belongs to $T_\eta\subseteq T$. When $m(x)=0$, $\langle T|D_{b}^0|x\rangle= \sqrt{\pi(x)/\pi(T)}$. For $m(x)\geq 1$, the definition of $m(x)$ gives
\begin{equation}
    G(x)(1-L)^{m(x)}
    >(1-L)G(x)\left(\frac{1-\rho}{G(x)}\right)^{-\log(1-L)/L}
    \geq1-\rho-L.
\end{equation}
    Therefore, for every $0\leq s\leq m(x)$, monotonicity of $\bar{g}_\eta$ implies
    \begin{equation}
    \bar{g}_\eta((1-L)^{m(x)-s}G(x))
    \geq \bar{g}_\eta((1-L)^{m(x)}G(x))
    \geq \frac{\eta-\rho-L}{\eta}
    \end{equation}
    Combining this with~\cref{thm:amplitude-gain} gives
    \begin{equation}
        \langle T|(D+bG)^{m(x)}|x\rangle
        \geq \left(\frac{\eta-\rho-L}{\eta}\right)
        \sqrt{\frac{\pi(x)}{\pi(T)}}
        \sum_{s=0}^{m(x)}K_{m(x),s}(L)b^sG(x)^s.
    \end{equation}
For $0\leq s\leq m(x)$, the product formula in~\cref{lem:inversion-count} gives
\begin{equation}
    K_{m(x),s}(L)
    =\prod_{j=1}^s\frac{1-(1-L)^{m(x)-s+j}}{1-(1-L)^j}
    \geq
    \binom{m(x)}{s}\left(\frac{1-(1-L)^{m(x)}}{m(x)L}\right)^s.
\end{equation}
For every $1\leq k\leq m(x)$ we have
\begin{equation}
    1-(1-L)^k\geq \frac{k}{m(x)}(1-(1-L)^{m(x)}),
\end{equation}
because the averages $k^{-1}\sum_{i=0}^{k-1}(1-L)^i$ are decreasing in $k$. Hence,
\begin{equation}
    \prod_{j=1}^s(1-(1-L)^{m(x)-s+j})
    \geq
    \frac{s!\binom{m(x)}{s}}{{m(x)}^s}(1-(1-L)^{m(x)})^s,
\end{equation}
while $\prod_{j=1}^s(1-(1-L)^j)\leq L^s s!$. Therefore
\begin{align}
    \sum_{s=0}^{m(x)} K_{m(x),s}(L)b^{s}G(x)^{s}
    &\geq \sum_{s=0}^{m(x)}\binom{m(x)}{s}
    \left(\frac{bG(x)(1-(1-L)^{m(x)})}{m(x)L}\right)^s \\
    &= \left(1+\frac{bG(x)(1-(1-L)^{m(x)})}{m(x)L}\right)^{m(x)}\\
    &\geq  \left(1+b(1-\rho-L)\right)^{m(x)}.
\end{align}
We now consider averaging the points in $S$. Since $M$ is larger than the corresponding $m(x)$ for each $x$, we have $\frac{1}{M+1}\sum_{m=0}^M D_b^m$ includes the right $m$ for each $x$ at least once. By~\cref{thm:energy-bound},
\begin{equation}
    E_b\leq1+b\mathbb{E}_\pi[G]+\frac{b^2\mathcal{L}_2(G)^2}{8\omega^2}.
\end{equation}
Combining these bounds gives
\begin{equation}
|\left\langle T\left|C_{b, M}\right|S \right\rangle|
\geq
\left(\frac{\eta-\rho-L}{\eta}\right)\frac{1}{M+1}
\sqrt{\frac{\pi(S)}{\pi(T)}}
\mathbb{E}_{x\sim\pi_S}\left[
\left(
\frac{1+b(1-\rho-L)}
{1+b\mathbb{E}_\pi[G]+\frac{b^2\mathcal{L}_2(G)^2}{8\omega^2}}
\right)^{m(x)}
\right].
\end{equation}
Maximizing the base of the exponent over $0\leq b\leq1$ gives the $b^*$ in the theorem statement. The optimality equation for $b$, including the case where the optimum is attained at $b=1$, implies
\begin{equation}
    \frac{b^*\mathcal{L}_2(G)^2}{8\omega^2}
    \left(2+b^*(1-\rho-L)\right)\leq 1-\rho-\mathbb{E}_\pi[G] -L.
\end{equation}
Consequently,
\begin{equation}
    \log\left(
    \frac{1+b^*(1-\rho-L)}
    {1+b^*\mathbb{E}_\pi[G]+\frac{{b^*}^2\mathcal{L}_2(G)^2}{8\omega^2}}
    \right)
    \geq\frac{b^*( 1-\rho-\mathbb{E}_\pi[G] -L)}{3}.
\end{equation}
Jensen's inequality and the fact that $C_{b,M}$ is a symmetric operator yield
\begin{equation}
|\left\langle S\left|C_{b, M}\right|T \right\rangle|
\geq
\left(\frac{\eta-\rho-L}{\eta}\right)\frac{1}{M+1}
\sqrt{\frac{\pi(S)}{\pi(T)}}
\exp\left(\frac{b^* ( 1-\rho-\mathbb{E}_\pi[G] -L)\mathbb{E}_{\pi_S}[m(x)] }{3}\right).  
\end{equation}
\end{proof}
Theorem~\ref{thm:normalized-amplitude-gain} implies that if $\eta-\rho=\Theta(1)$, $1-\rho-\mathbb{E}_\pi[G]=\Theta(1)$, $L=O(1/n)$, and
\begin{equation}
\mathbb E_{x\sim\pi_S}\left[
\log\left(\frac{G(x)}{1-\rho}\right)
\right]=\Theta(1),
\end{equation}
then the amplitude $|\langle S|T\rangle|$ is amplified by a factor $\exp(\Theta(n))$. This gives an exponential improvement over direct search from the warm state $|T\rangle$ for an element of $S$. In Section~\ref{sec:application-to-optimization}, we specialize to $T=T_\eta$ and initialize the parameters from the problem instance.

\subsection{Discussion}
In the high $M$ limit, the average $\frac{1}{M+1}\sum_{m=0}^M D_b^m$ converges to the projector onto the eigenvalue-$1$ eigenspace of $D_b$, equivalently the projector onto the ground state of the tilted Hamiltonian. However, this does not mean that its matrix element $\langle T |\psi_b\rangle\langle \psi_b|S\rangle$ satisfies the finite-$M$ lower bound. In fact, this is likely false because $|\psi_b\rangle$ is very close to the $|+\rangle$ state, which is why the algorithms in~\cite{Dalzell_2023,hastings2018shortpathquantumalgorithm} include a $2^{n/2}$ dependence in the denominator. We avoid this term because, for polynomial $M$, the power average can retain a substantial contribution from the other eigenstates of $H_b$. The algorithm therefore relies on the walk-like behavior of the operator rather than on projection. This is also why a single term with large $m$, which more closely approximates an efficient projector, does not give the same result.



\section{Applications to Optimization}
\label{sec:application-to-optimization}
\subsection{Overview of the Conditioning-and-Search Algorithm}
Le Gall and Tamaki~\cite{gall2026dequantizingshortpathquantumalgorithms} introduced conditioning-and-search by turning the threshold-set and local-search arguments from the analysis of short-path algorithms into a classical algorithm. We state the algorithm in their minimization convention; replacing $H$ by $-H$ gives the maximization convention used for the tilted-walk algorithm.

Let $\mathcal X=\{0,1\}^n$, let $H:\mathcal X\to\mathbb R$ have optimum $H_{\min}<0$, and fix $\eta\in(0,1)$. The conditioning set is
\begin{equation}
    T_\eta=\{x\in\mathcal{X}:H(x)\leq(1-\eta)H_{\min}\}.
    \label{eq:condition-search-threshold}
\end{equation}
For each $x\in T_\eta$, let $\mathcal N_\eta(x)$ be a searchable neighborhood. A set $S_\eta\subseteq T_\eta$ is successful if there is an optimum $x^\star$ such that $x^\star\in\mathcal N_\eta(x)$ for every $x\in S_\eta$. The algorithm repeatedly samples $x$ uniformly from $T_\eta$ and searches $\mathcal N_\eta(x)$. The set $S_\eta$ is used only to analyze the success probability and need not be found by the algorithm.

If one sample from a distribution $\mu$ on $T$ costs $\tau_{\mathrm{samp}}$, one neighborhood search costs $\tau_{\mathrm{search}}$, and $\mu(S)\geq\beta$, repetition finds an optimum in expected time
\begin{equation}
    O\!\left(\frac{\tau_{\mathrm{samp}}+\tau_{\mathrm{search}}}
    {\beta}\right).
    \label{eq:condition-search-abstract-cost}
\end{equation}
This is the abstract conditioning-and-search principle \cite[Theorem~2.1]{gall2026dequantizingshortpathquantumalgorithms}. For the uniform distribution on $T_\eta$, we have $\beta=|S_\eta|/|T_\eta|$. Rejection sampling from $\{0,1\}^n$ then gives
\begin{equation}
    O^*\!\left(
        \frac{2^n}{|S_\eta|}
        +\frac{|T_\eta|}{|S_\eta|}\tau_{\mathrm{search}}
    \right)
    \label{eq:condition-search-rejection-cost}
\end{equation}
expected time~\cite[Corollaries~2.2 and~2.3]{gall2026dequantizingshortpathquantumalgorithms}. 
\subsection{Quantum Conditioning-and-Search}
We now specialize Theorem~\ref{thm:normalized-amplitude-gain} to the tight choice $T=T_\eta$. The theorem then compares the tilted amplitude on a successful set with the direct amplitude from the threshold state.

We take $D = X/n$ so that $\pi$ is uniform. The set states are therefore uniform superpositions, and the factor $\pi(S)/\pi(T)$ in the general amplitude bound becomes $|S|/|T|$.

Fix a set of coordinates $J\subseteq[n]$. For a center $x$ and radius $r$, define the restricted Hamming ball
\begin{equation}
    B_J(x,r)
    :=\left\{z\in\{0,1\}^n:
        d_H(x,z)\leq r,\ 
        \{i:x_i\neq z_i\}\subseteq J\right\}.
    \label{eq:restricted-hamming-ball}
\end{equation}
The choice $J=[n]$ gives the ordinary radius-$r$ Hamming ball. Let $x^\star$ be a maximizer of $G$. The successful centers are
\begin{equation}
    S_r:=\{x:x^\star\in B_J(x,r)\}=B_J(x^\star,r),
    \label{eq:successful-centers}
\end{equation}
where the equality follows from symmetry of the restricted Hamming distance. Define the local optimum
\begin{equation}
    g_{J,r}(x):=\max_{z\in B_J(x,r)}G(z).
\end{equation}
Every $x\in S_r$ satisfies $g_{J,r}(x)=1$.

\begin{algorithm}[H]
\caption{Quantum Conditioning-and-Search}
\label{alg:coherent-short-path-filter}
\small
\setlength{\abovedisplayskip}{3pt}
\setlength{\belowdisplayskip}{3pt}
\begin{algorithmic}[1]
\Require Binary oracle access to the objective $f$ and access to the transformation $C_{b,M}$; parameters $\eta$, $J$, and $r$; a lower bound $a$ on the marked amplitude; and failure probability $\delta$.
\State Prepare $|T_\eta\rangle$ from the uniform state by amplitude amplification, marking $x$ when $G(x)\geq1-\eta$.
\State Apply $C_{b,M}$ to $|T_\eta\rangle$.
\State Coherently compute $g_{J,r}(x)$ by quantum maximum finding over $B_J(x,r)$ and mark the state when $g_{J,r}(x)=1$.
\State Apply fixed-point amplitude amplification to the complete state preparation using the lower bound $a$ and failure probability $\delta$.
\State Measure a marked center $x$, run quantum maximum finding once more on $B_J(x,r)$, and output a maximizer.
\end{algorithmic}
\end{algorithm}

The following is the main theoretical guarantee for the Quantum Conditioning-and-Search algorithm.

\begin{theorem}[Optimization from a tilted-walk amplitude]
\label{thm:optimization-from-amplitude}
Given binary oracle access to $f$, let $T_\eta\subseteq\{0,1\}^n$ be nonempty, let $x^\star$ maximize $G$. Fix $J\subseteq[n]$ and $r$, and let $S_r$ be the successful-center set in \eqref{eq:successful-centers}. Let $b\in[0,1]$ and $M=\operatorname{poly}(n)$. If, for some $\kappa\geq0$ and polynomial $p$,
\begin{equation}
    \left|\langle S_r|C_{b,M}|T_\eta\rangle\right|
    \geq \frac{1}{p(n)}2^{\kappa n/2}
       \frac{|S_r|^{1/2}}{|T_\eta|^{1/2}},
    \label{eq:optimization-amplitude-hypothesis}
\end{equation}
then Algorithm~\ref{alg:coherent-short-path-filter} finds a maximizer of $G$ with bounded error using
\begin{equation}
    \operatorname{poly}(n)2^{-\kappa n/2}
    \frac{|T_\eta|^{1/2}}{|S_r|^{1/2}}
    \left(
        \frac{2^{n/2}}{|T_\eta|^{1/2}}+|S_r|^{1/2}
    \right)
    \label{eq:optimization-runtime}
\end{equation}
oracle calls.
\end{theorem}

\begin{proof}
Let $\Pi_{S_r}$ project onto the computational-basis states in $S_r$. Because $|S_r\rangle$ is a unit vector in its range, the hypothesis gives
\begin{align}
    \|\Pi_{S_r}C_{b,M}|T_\eta\rangle\|
    &\geq
      |\langle S_r|C_{b,M}|T_\eta\rangle|\\
    &\geq \frac{1}{p(n)}2^{\kappa n/2}
       \frac{|S_r|^{1/2}}{|T_\eta|^{1/2}}.
    \label{eq:optimization-good-amplitude}
\end{align}
By definition of $S_r$, every $x\in S_r$ satisfies $g_{J,r}(x)=1$. The marking step in Algorithm~\ref{alg:coherent-short-path-filter} therefore marks a subspace containing the component supported on $S_r$. By \eqref{eq:optimization-good-amplitude}, fixed-point amplitude amplification requires at most
\begin{equation}
    \operatorname{poly}(n)2^{-\kappa n/2}
    \frac{|T_\eta|^{1/2}}{|S_r|^{1/2}}
\end{equation}
applications of the complete preparation unitary and its inverse. On the other hand, preparing $|T_\eta\rangle$ from the uniform state costs
\begin{equation}
    \operatorname{poly}(n)\frac{2^{n/2}}{|T_\eta|^{1/2}}
\end{equation}
queries to the binary oracle for $f$. Every restricted ball $B_J(x,r)$ has cardinality $|S_r|$, so quantum maximum finding costs $\operatorname{poly}(n)|S_r|^{1/2}$ queries to the binary oracle for $f$. Applying $C_{b,M}$ has polynomial cost because $M=\operatorname{poly}(n)$. Multiplying these costs by the number of amplification iterations gives~\eqref{eq:optimization-runtime}. One final maximum search over $B_J(x,r)$ returns a maximizer and is absorbed into the same bound.
\end{proof}

\subsection{Application to \texorpdfstring{\textsc{MAX-E$k$-LIN2}}{MAX-Ek-LIN2}}
\label{subsec:eklin2-application}

A weighted \textsc{MAX-E$k$-LIN2} instance consists of $N$ parity equations
\begin{equation}
    \bigoplus_{i\in A_j}x_i=b_j,
    \qquad j\in[N],
\end{equation}
where $x\in\{0,1\}^n$, each $A_j\subseteq[n]$ has cardinality $k$, each $b_j\in\{0,1\}$, and equation $j$ has weight $w_j>0$. Here $\oplus$ denotes addition modulo two. The problem is to find an assignment maximizing the total weight of the satisfied equations,
\begin{equation}
    F(x)=\sum_{j=1}^N w_j
        \mathbf{1}_{\{\bigoplus_{i\in A_j}x_i=b_j\}}.
\end{equation}
The identity
\begin{equation}
    \mathbf{1}_{\{\bigoplus_{i\in A_j}x_i=b_j\}}
    =\frac{1+(-1)^{b_j+\sum_{i\in A_j}x_i}}{2}
\end{equation}
shows that maximizing $F$ is equivalent to maximizing the centered objective
\begin{equation}
    H(x)=\sum_{j=1}^Nw_j(-1)^{b_j+\sum_{i\in A_j}x_i}
\end{equation}
Thus $H$ is a homogeneous degree-$k$ Fourier polynomial on $\{0,1\}^n$. Its expectation under the uniform distribution is zero, which will be convenient for the analysis. Let $W=\sum_{j=1}^Nw_j$. We also define the quantity 
\begin{equation}
\Delta=\frac{kH_{\max}}{n d_{\max}}
\end{equation}
for convenience to characterize the problem runtime in the following theorem.

\begin{theorem}[Tilted-walk runtime for weighted \textup{\textsc{MAX-E$k$-LIN2}}]
\label{thm:tilted-walk-eklin2}
Fix $k\geq2$ and assume that $H_{\max}>0$, $\Delta=\Theta(1)$, and $1-\mathbb{E}[G]=\Theta(1)$. For weighted \textup{\textsc{MAX-E$k$-LIN2}}, implement Algorithm~\ref{alg:coherent-short-path-filter} with $D=X/n$, $G=\frac12\left(1+(H+W)/(H_{\max}+W)\right)$, write $\mu=\mathbb{E}[G]$, and choose
\begin{equation}
    \eta=\frac{1-\mu}{2},
    \qquad
    \rho=\frac{1-\mu}{4},
    \qquad
    r=\left\lfloor\frac{\Delta n}{16k}\right\rfloor,
    \label{eq:eklin2-parameters}
\end{equation}
Define the pseudo-Lipschitz constant
\begin{equation}
    L_A:=\frac{4k(1-\mu)}{n\Delta}.
    \label{eq:eklin2-L}
\end{equation}
and
\begin{equation}
    b^{*}=\min\left\{1,
    \frac{
    \sqrt{1+
    \dfrac{2\Delta^2(1-\rho-L_A)
    \left(3(1-\mu)/4-L_A\right)}{k^2(1-\mu)^2}}
    -1}{1-\rho-L_A}
    \right\}.
    \label{eq:eklin2-b}
\end{equation}
Then the algorithm finds an assignment maximizing the total satisfied weight with probability at least $2/3$ using
\begin{equation}
    Q\leq 2^{(1-\gamma_A-\kappa_A)n/2+o(n)}
      =2^{(1-c_A^{\mathrm{cl}}-\kappa_A)n/2+o(n)}
    \label{eq:tilted-walk-eklin2-runtime}
\end{equation}
queries, where $c_A^{\mathrm{cl}}=\gamma_A$ is the exponent of the classical conditioning-and-search algorithm in~\cite{gall2026dequantizingshortpathquantumalgorithms} for this parameter choice and
\begin{equation}
    \gamma_A=\frac{\Delta^2}{8k^2\log 2},
    \qquad
    \kappa_A=
    \frac{2b^{*}}{3nL_A\log 2}
    \left(\frac{3(1-\mu)}{4}-L_A\right)
    \log\left(\frac{1-\rho/2}{1-\rho}\right).
    \label{eq:eklin2-gains}
\end{equation}

\end{theorem}
\begin{proof}
For each variable $i$, let
\begin{equation}
        d_i=\sum_{j:i\in A_j}w_j,
        \qquad
        d_{\max}=\max_i d_i,
        \qquad
        R=H_{\max}+W.
\end{equation}

Since $-W\leq H(x)\leq H_{\max}$, the shifted and normalized objective satisfies $1/2\leq G\leq1$ and $\max_xG(x)=1$. Moreover,
\begin{equation}
    \mu=\mathbb{E}[G]=\frac12\left(1+\frac{W}{R}\right),
    \qquad
    (1-\mu)R=\frac{H_{\max}}{2}.
\end{equation}
Flipping coordinate $i$ changes $H$ by at most $2d_i$, and hence
\begin{equation}
    |H(x)-H(x^{(i)})|\leq2d_i.
\end{equation}
For this application, we use the extensive objective $f=H/d_{\max}$, which has the same maximizers as $H$, and take $A_c=W/d_{\max}$ and $R_c=R/d_{\max}$. The one-step quadratic variation of $f$ is at most
\begin{equation}
    \frac{4}{nd_{\max}^2}\sum_{i=1}^n d_i^2
    \leq4.
\end{equation}
We therefore take
\begin{equation}
    \mathcal L_2(f):=2=O(1),
\end{equation}
so Condition~\ref{cond:quadratic-variation} holds. Moreover, $H_{\max}/d_{\max}=n\Delta/k=\Theta(n)$, so the maximum of the chosen objective $f=H/d_{\max}$ is $\Theta(n)$. Lemma~\ref{lem:normalization-regularity} then gives
\begin{equation}
    \mathcal{L}_2(G)^2=\frac{d_{\max}^2}{R^2}.
    \label{eq:eklin2-l2-bound}
\end{equation}
By the same lemma, $G$ is pseudo-Lipschitz with constant
\begin{equation}
    L=2\mathcal L_2(G)=\frac{2d_{\max}}{R}=L_A.
\end{equation}
The last equality follows directly from
\begin{equation}
    \frac{4k(1-\mu)}{n\Delta}
    =\frac{4k}{n\Delta}\frac{H_{\max}}{2R}
    =\frac{2d_{\max}}{R},
\end{equation}
where we used $\Delta=kH_{\max}/(nd_{\max})$.
For the normalization $D=X/n$, the hypercube log-Sobolev constant in \cref{thm:energy-bound} is $\omega=1/n$.

For the value of $\eta$ stated in the theorem,
\begin{equation}
        1-\eta-\mathbb{E}[G]
        =\frac{1-\mathbb{E}[G]}{2}
        =\frac{H_{\max}}{4R}.
\end{equation}
The threshold condition can be written entirely in terms of $H$:
\begin{align}
    G(x)\geq1-\eta
    &\iff
    \frac12\left(1+\frac{H(x)+W}{R}\right)
       \geq1-\frac{1-\mu}{2}\\
    &\iff H(x)\geq\frac{H_{\max}}2.
\end{align}
Thus
\begin{equation}
        T_\eta
        =\left\{x:H(x)\geq\frac{H_{\max}}{2}\right\}.
\end{equation}
Under a uniformly random assignment, $\mathbb E[H]=0$, and changing coordinate $i$ changes $H$ by at most $2d_i$. McDiarmid's inequality with deviation $H_{\max}/2$ therefore gives
\begin{equation}
    \Pr\!\left[H\geq\frac{H_{\max}}2\right]
    \leq
    \exp\!\left(
       -\frac{2(H_{\max}/2)^2}{\sum_i(2d_i)^2}
    \right)
    =\exp\!\left(-\frac{H_{\max}^2}{8\sum_i d_i^2}\right).
\end{equation}
Multiplying by $2^n$ and using $\sum_i d_i^2\leq nd_{\max}^2$ yields
\begin{align}
        |T_\eta|
        &\leq2^n\exp\left(-\frac{H_{\max}^2}{8\sum_i d_i^2}\right)\notag\\
        &\leq\exp\left(
                n\log2-\frac{n\Delta^2}{8k^2}
        \right),
        \label{eq:eklin2-threshold-bound}
\end{align}
where the last inequality uses $H_{\max}/d_{\max}=n\Delta/k$. Equivalently,
\begin{equation}
    |T_\eta|\leq2^{(1-\gamma_A)n},
    \qquad
    \gamma_A=\frac{\Delta^2}{8k^2\log2}.
\end{equation}

If $H_{\max}=0$, then $H$ is identically zero and every assignment is optimal. Assume henceforth that $H_{\max}>0$. Since $\sum_i d_i=kW$ and $H_{\max}\leq W$, we have $d_{\max}\geq kW/n$ and therefore $0<\Delta\leq1$. Let
\begin{equation}
        S_r=\{x\in\{0,1\}^n:d_H(x,x^\star)\leq r\},
\end{equation}
where $r$ is given by~\eqref{eq:eklin2-parameters}. Along any path of at most $r$ bit flips from $x^\star$, the value of $H$ decreases by at most $2rd_{\max}$. Hence, for every $x\in S_r$,
\begin{equation}
    G(x)\geq1-\frac{rd_{\max}}{R}
    \geq1-\frac{H_{\max}}{16R}
    =1-\frac{1-\mu}{8}=1-\frac{\rho}{2}.
    \label{eq:eklin2-inner-threshold}
\end{equation}
In particular, $S_r\subseteq T_\eta$. The standard Hamming-ball estimate gives
\begin{equation}
    |S_r|=2^{h(\Delta/(16k))n+o(n)}.
    \label{eq:eklin2-successful-set-bound}
\end{equation}
Indeed, $r/n=\Delta/(16k)+O(1/n)$ and $\Delta/(16k)<1/2$, so
\begin{equation}
    \sum_{j=0}^r\binom nj
    =2^{h(\Delta/(16k))n+o(n)}.
\end{equation}
The inclusion $S_r\subseteq T_\eta$ follows because $G(x)\geq1-\rho/2$ on $S_r$ and $\rho/2<\eta$.

Since $\omega=1/n$,~\eqref{eq:eklin2-l2-bound} and the definition of $\Delta$ imply
\begin{equation}
    \frac{\mathcal{L}_2(G)^2}{8\omega^2}
    =\frac{k^2(1-\mu)^2}{2\Delta^2}.
\end{equation}
We now verify all hypotheses of \cref{thm:normalized-amplitude-gain}. Since $1-\mu=\Theta(1)$ and $L_A=O(1/n)$,
\begin{align}
    \eta-\rho-L_A
       &=\frac{1-\mu}{4}-L_A>0,\\
    1-\rho-\mu-L_A
       &=\frac{3(1-\mu)}4-L_A>0
\end{align}
for all sufficiently large $n$. The displayed value \eqref{eq:eklin2-b} is obtained from the theorem's formula for $b^{*}$ by substituting $L=L_A$, $\omega=1/n$, and
\begin{equation}
    \frac{\mathcal L_2(G)^2}{8\omega^2}
    =\frac{k^2(1-\mu)^2}{2\Delta^2}.
\end{equation}
Moreover, $x^\star\in S_r$ and $G(x^\star)=1$, so the value of $M$ in Theorem~\ref{thm:normalized-amplitude-gain} is
\begin{equation}
    M=\left\lceil\frac{\log(1/(1-\eta))}{L_A}\right\rceil=O(n).
\end{equation}
For every $x\in S_r$, \eqref{eq:eklin2-inner-threshold} gives
\begin{equation}
    m(x)\geq\frac{1}{L_A}
    \log\left(\frac{1-\rho/2}{1-\rho}\right)
    \qquad (x\in S_r).
\end{equation}
Equivalently,
\begin{equation}
    \mathbb E_{x\sim\pi_{S_r}}
       \left[\log\left(\frac{G(x)}{1-\rho}\right)\right]
    \geq
    \log\left(\frac{1-\rho/2}{1-\rho}\right).
\end{equation}
Substitution in the definition of $\kappa$ in \cref{thm:normalized-amplitude-gain} gives exactly $\kappa_A$ in \eqref{eq:eklin2-gains}. The remaining prefactor in that theorem is inverse polynomial because $\eta-\rho-L_A=\Theta(1)$ and $M=O(n)$. Therefore its amplitude bound has the form required in \eqref{eq:optimization-amplitude-hypothesis}, with $S_r$, $\kappa=\kappa_A$, and a polynomial $p$.

Finally, $h(t)\geq2t$ on $[0,1/2]$, so
\begin{equation}
    h\left(\frac{\Delta}{16k}\right)
    \geq\frac{\Delta}{8k}
    \geq\frac{\Delta^2}{8k^2\log2}=\gamma_A.
\end{equation}
Set $s_A=h(\Delta/(16k))$. Substituting $|T_\eta|\leq2^{(1-\gamma_A)n}$ and $|S_r|=2^{s_A n+o(n)}$ into \eqref{eq:optimization-runtime} gives two contributions. The threshold-state preparation contribution is
\begin{equation}
    2^{-\kappa_A n/2}
    \frac{|T_\eta|^{1/2}}{|S_r|^{1/2}}
    \frac{2^{n/2}}{|T_\eta|^{1/2}}
    =2^{(1-s_A-\kappa_A)n/2+o(n)},
\end{equation}
whereas the local-search contribution is
\begin{equation}
    2^{-\kappa_A n/2}|T_\eta|^{1/2}
    \leq2^{(1-\gamma_A-\kappa_A)n/2}.
\end{equation}
Since $s_A\geq\gamma_A$, both terms are at most $2^{(1-\gamma_A-\kappa_A)n/2+o(n)}$, proving the first equality in \eqref{eq:tilted-walk-eklin2-runtime}.

For the same relative threshold $H\geq H_{\max}/2$, the minimization convention of Le Gall and Tamaki has threshold parameter $1/2$. Their successful-set exponent is
\begin{equation}
    h\left(\frac{1-2^{-1/k}}2\right).
\end{equation}
Since $1-2^{-1/k}\geq1/(2k)$ and $h(t)\geq2t$ on $[0,1/2]$, this exponent is at least $1/(2k)$. Also $0<\Delta\leq1$ implies $\gamma_A<1/(2k)$. Hence their classical exponent is
\begin{equation}
    c_A^{\mathrm{cl}}
    =\min\left\{\gamma_A,
       h\left(\frac{1-2^{-1/k}}2\right)\right\}
    =\gamma_A.
\end{equation}
Finally, the assumptions imply that $b^{*}$, $nL_A$, and every other factor in $\kappa_A$ are bounded above and below by positive constants for sufficiently large $n$. Thus $\kappa_A>0$ is independent of $n$, and the extra factor $2^{-\kappa_A n/2}$ is a super-quadratic improvement over the classical conditioning-and-search runtime.
\end{proof}

\subsection{Application to Max-k-CSP}
\label{subsec:maxkcsp-application}
A weighted Max-k-CSP instance on variables $x_1,\dots, x_n \in \{0,1\}$ consists of $N$ constraints indexed by $j\in [N]$. The $j$-th constraint $C_j$ is specified by a set of variables of size $k_j\leq k$, a positive weight $w_j$, and a predicate $P_j: \{0,1\}^{k_j}\to \{0,1\}$. 

Let $s_j\in [1, 2^{k_j}-1]$ be the number of configurations of the $k_j$ variables that make $P_j=1$. We define the centered contribution of clause $j$ by 
\begin{equation}
C_j(x) = \begin{cases}
    w_j & \text{if } P_j(x)=1\\
    -\frac{s_j}{2^{k_j}-s_j}w_j &  \text{if } P_j(x)=0
\end{cases}
\end{equation}
so that
\begin{equation}
    \mathbb{E}[C_j]
    =\frac{s_j}{2^{k_j}}w_j
    -\frac{2^{k_j}-s_j}{2^{k_j}}
        \frac{s_j}{2^{k_j}-s_j}w_j
    =0.
\end{equation}
Then the centered objective is
\begin{equation}
H(x) = \sum_{j=1}^N C_j(x)
\end{equation}
with $\mathbb{E}[H]=0$ and $H_{\max}=\max_{x\in\{0,1\}^n}H(x)\geq0$. We define $W=\sum_{j=1}^N w_j$ and fix $x^\star$ such that $H(x^\star)=H_{\max}$. For the conditioning-and-search baseline, we apply the minimization convention to $-H$ and write $H_{\min}:=-H_{\max}$. For each variable $i\in[n]$, let
\begin{equation}
d_i = \sum_{j:i\in \text{vars}(j)}w_j
\end{equation}
be its weighted degree and let
\begin{equation}
    \Sigma=\sum_{i=1}^n d_i=\sum_{j=1}^N k_jw_j,
    \qquad
    d_{\mathrm{avg}}=\frac{\Sigma}{n},
    \qquad
    \nu=\frac{n\sum_{i=1}^n d_i^2}{\Sigma^2}\geq1,
\end{equation}
where the inequality follows from Cauchy--Schwarz. Define the change factor of constraint $j$ and its maximum by
\begin{equation}
    \Gamma_j
    =1+\frac{s_j}{2^{k_j}-s_j}
    =\frac{2^{k_j}}{2^{k_j}-s_j},
    \qquad
    \Gamma_{\max}=\max_{j\in[N]}\Gamma_j.
\end{equation}

Define the set of light coordinates by
\begin{equation}
    L_{\mathrm{light}}
    =\left\{i\in[n]:d_i\leq2d_{\mathrm{avg}}\right\}.
\end{equation}
The threshold slack determines the radius
\begin{equation}
        r_\eta^{\mathrm{lip}}
            =\left\lfloor
                \frac{\eta|H_{\min}|}
                         {2\Gamma_{\max}d_{\mathrm{avg}}}
                \right\rfloor
            =\left\lfloor\frac{\theta_\eta n}{2}\right\rfloor,
        \qquad
        \theta_\eta
            =\frac{\eta|H_{\min}|}{\Gamma_{\max}\Sigma}.
        \label{eq:condition-search-lip-radius}
\end{equation}
For $L\subseteq[n]$, define the restricted Hamming ball
\begin{equation}
    B_L(x,r)=\left\{y:d_H(x,y)\leq r,\
        \{i:x_i\neq y_i\}\subseteq L\right\}.
\end{equation}
Set
\begin{equation}
    S_\eta^{\mathrm{lip}}
        =B_{L_{\mathrm{light}}}(x^\star,r_\eta^{\mathrm{lip}}),
    \qquad
    \mathcal{N}_\eta(x)
        =B_{L_{\mathrm{light}}}(x,r_\eta^{\mathrm{lip}}).
\end{equation}
Thus $S_\eta^{\mathrm{lip}}\subseteq T_\eta$.  For every $x\in S_\eta^{\mathrm{lip}}$, the definition of the restricted ball gives $x^\star\in\mathcal{N}_\eta(x)$.  Hence $S_\eta^{\mathrm{lip}}$ is a successful set.  Counting assignments on at least $n/2$ light coordinates gives
\begin{equation}
        |S_\eta^{\mathrm{lip}}|
            \geq2^{\frac{1}{2}h(\theta_\eta)n-o(n)}.
\end{equation}
Every neighborhood $\mathcal{N}_\eta(x)$ has the same cardinality as $S_\eta^{\mathrm{lip}}$.  The successful-set and neighborhood-search exponents are therefore $\kappa_\eta=\frac{1}{2}h(\theta_\eta)$ \cite[Proposition~3.4 and Theorem~3.5]{gall2026dequantizingshortpathquantumalgorithms}.

The conditioning exponent follows from a separate concentration argument. Under a uniformly random assignment, the centered objective has mean zero and bounded-difference parameter at most $\Gamma_{\max}d_i$ in coordinate $i$. Since $\sum_i d_i^2=\nu\Sigma^2/n$, McDiarmid's inequality yields
\begin{equation}
        |T_\eta|
            \leq2^n\exp\!\left(
                -\frac{2(1-\eta)^2|H_{\min}|^2n}
                            {\Gamma_{\max}^2\nu\Sigma^2}
            \right)
            =2^{(1-\gamma_\eta)n},
        \label{eq:condition-search-mcdiarmid}
\end{equation}
which is Proposition~4.4 of Ref.~\cite{gall2026dequantizingshortpathquantumalgorithms}.

\begin{theorem}[Conditioning-and-search for weighted \textsc{MAX-$k$-CSP}]
\label{thm:conditioning-search-maxkcsp}
Assume that the optimum value $H_{\min}$ is known, and fix $\eta\in(0,1)$.  Conditioning-and-search finds an optimum of any weighted \textup{\textsc{MAX-$k$-CSP}} instance with high probability in time
\begin{equation}
    2^{(1-c_B^{\mathrm{cl}})n+o(n)},
    \label{eq:condition-search-csp-runtime}
\end{equation}
where
\begin{align}
    \gamma_\eta
      &=\frac{2(1-\eta)^2}{\ln 2}\,
                \frac{1}{\Gamma_{\max}^2\nu}
        \left(\frac{|H_{\min}|}{\Sigma}\right)^2,
    &
    \theta_\eta
            &=\frac{\eta|H_{\min}|}{\Gamma_{\max}\Sigma},
    \label{eq:condition-search-csp-parameters}\\
    c_B^{\mathrm{cl}}
      &=\min\!\left\{\gamma_\eta,\frac{1}{2}h(\theta_\eta)\right\}.
\end{align}
\end{theorem}

Theorem~\ref{thm:conditioning-search-maxkcsp} is Theorem~4.5 of Ref.~\cite{gall2026dequantizingshortpathquantumalgorithms}.

\begin{corollary}[Explicit weighted \textsc{MAX-$k$-CSP} exponent]
\label{cor:conditioning-search-maxkcsp-explicit}
Choose $\eta=1/2$ in Theorem~\ref{thm:conditioning-search-maxkcsp}. Using $\Sigma\leq kW$, the runtime exponent satisfies
\begin{equation}
    c_B^{\mathrm{cl}}
    \geq \frac{1}{2\ln 2}\,
    \frac{1}{\Gamma_{\max}^2k^2\nu}
       \left(\frac{|H_{\min}|}{W}\right)^2
    \geq \frac{1}{2\ln 2}\,
    \frac{1}{2^{2k}k^2\nu}
       \left(\frac{|H_{\min}|}{W}\right)^2.
\end{equation}
For an unweighted instance, $W=N$. Defining $\Delta_{\mathrm{CSP}}=|H_{\min}|/N$ gives
\begin{equation}
    c_B^{\mathrm{cl}}
    \geq 0.7213\,
    \frac{\Delta_{\mathrm{CSP}}^2}{2^{2k}k^2\nu},
    \label{eq:condition-search-csp-corollary}
\end{equation}
and conditioning-and-search runs in time $2^{(1-c_B^{\mathrm{cl}})n+o(n)}$ with high probability.
\end{corollary}

Corollary~\ref{cor:conditioning-search-maxkcsp-explicit} is Corollary~4.7 of Ref.~\cite{gall2026dequantizingshortpathquantumalgorithms}; its unweighted specialization is Theorem~1.2 in that reference.  A constant lower bound on $\Delta_{\mathrm{CSP}}$ and a constant upper bound on $\nu$ give a constant $c_B^{\mathrm{cl}}>0$.  For \textsc{MAX-E3-SAT}, the condition $\Delta_{\mathrm{CSP}}=\Omega(1)$ includes mildly regular instances in which an optimum satisfies at least a $(7/8+\delta)$ fraction of clauses for constant $\delta>0$~\cite{gall2026dequantizingshortpathquantumalgorithms}.

We now analyze the tilted-walk algorithm using the equivalent maximization convention.

The two possible values of $C_j$ differ by $\Gamma_jw_j$. To make the objective nonnegative, define
\begin{equation}
    B=\sum_{j=1}^N\frac{s_j}{2^{k_j}-s_j}w_j,
    \qquad
    R=H_{\max}+B,
    \qquad
    G(x)=\frac12\left(1+\frac{H(x)+B}{R}\right).
\end{equation}
Indeed,
\begin{equation}
    H(x)+B=\sum_{j=1}^N\Gamma_jw_jP_j(x),
\end{equation}
so $1/2\leq G(x)\leq1$ and $\max_xG(x)=1$. Moreover, if $\mu=\mathbb{E}[G]$, then
\begin{equation}
    {\mu=\frac12\left(1+\frac{B}{R}\right)},
    \qquad
    {(1-\mu)R=\frac{H_{\max}}{2}}.
\end{equation}
We apply the tilted-walk construction to this normalized objective:
\begin{equation}
    D_b=\frac{D+bG}{E_b},
    \qquad
    E_b=\lambda_{\max}(D+bG),
    \qquad
    C_{b,M}=\frac{1}{M+1}\sum_{\ell=0}^M D_b^\ell.
\end{equation}

{
Define the known local scale
\begin{equation}
    \Lambda_B:=\frac{\Gamma_{\max}\sqrt{\nu}\Sigma}{n}.
\end{equation}
For this application, we use the extensive objective $f=H/\Lambda_B$, which has the same maximizers as $H$, and take $A_c=B/\Lambda_B$ and $R_c=R/\Lambda_B$. Flipping coordinate $i$ changes $H$ by at most $\Gamma_{\max}d_i$. The one-step quadratic variation of $f$ is therefore at most
\begin{equation}
    \frac{\Gamma_{\max}^2}{n\Lambda_B^2}\sum_{i=1}^n d_i^2
    =1.
\end{equation}
We therefore take
\begin{equation}
    \mathcal L_2(f):=1=O(1),
\end{equation}
so Condition~\ref{cond:quadratic-variation} holds. For fixed $k$, the basic bounds on $B$, $\Sigma$, and $H_{\max}$ give $\Gamma_{\max}\sqrt{\nu}\Sigma/R=\Theta(1)$ under the theorem assumptions. Since $H_{\max}/R=2(1-\mu)=\Theta(1)$, it follows that
\begin{equation}
    \max_x f(x)=\frac{H_{\max}}{\Lambda_B}
    =\frac{nH_{\max}}{\Gamma_{\max}\sqrt\nu\Sigma}
    =\Theta(n).
\end{equation}
Lemma~\ref{lem:normalization-regularity} then gives
\begin{equation}
    \mathcal{L}_2(G)^2
    =\frac{\Gamma_{\max}^2\nu\Sigma^2}{4n^2R^2}.
\end{equation}
By the same lemma, $G$ is pseudo-Lipschitz with constant
\begin{equation}
    L=2\mathcal L_2(G)
    =\frac{\Gamma_{\max}\sqrt{\nu}\Sigma}{nR}=:L_B.
    \label{eq:maxkcsp-L}
\end{equation}
\par}
For the normalization $D=X/n$, the hypercube log-Sobolev constant is $\omega=1/n$.

\begin{theorem}[Tilted-walk runtime for weighted \textup{\textsc{MAX-$k$-CSP}}]
\label{thm:tilted-walk-maxkcsp}
Fix $k$ and assume that $H_{\max}>0$, $1-\mu=\Theta(1)$, and $\Gamma_{\max}^2\nu\Sigma^2/R^2=O(1)$. Choose
\begin{equation}
    \label{eq:maxkcsp-quantum-parameters}
\eta=\frac{1-\mu}{2},
\qquad
\rho=\frac{1-\mu}{4},
\qquad
r=\left\lfloor
\frac{H_{\max}n}{16\Gamma_{\max}\Sigma}
\right\rfloor,
\end{equation}
and
{
\begin{equation}
    b^{*}=\min\left\{1,
    \frac{
    \sqrt{1+
    \dfrac{32R^2(1-\rho-L_B)
    \left(3(1-\mu)/4-L_B\right)}
    {\Gamma_{\max}^2\nu\Sigma^2}}
    -1}{1-\rho-L_B}
    \right\}.
    \label{eq:maxkcsp-b}
\end{equation}
\par}
Define
\begin{align}
    \gamma_B
      &=\frac{H_{\max}^2}
    {2(\log2)\Gamma_{\max}^2\nu\Sigma^2},
    &
    \sigma_B
      &=\frac12 h\left(\frac{H_{\max}}
      {8\Gamma_{\max}\Sigma}\right),
    \label{eq:maxkcsp-exponents}\\
        {\kappa_B}
            &={\frac{2b^{*}}{3nL_B\log2}
            \left(\frac{3(1-\mu)}4-L_B\right)
            \log\left(\frac{1-\rho/2}{1-\rho}\right)}.
    \label{eq:maxkcsp-kappa}
\end{align}
Then Algorithm~\ref{alg:coherent-short-path-filter}, using restricted Hamming-ball search on the light coordinates, finds a maximizer with bounded error using
\begin{equation}
    Q\leq
    2^{\left(1-\min\{\gamma_B,\sigma_B\}-\kappa_B\right)n/2+o(n)}
    \label{eq:tilted-walk-maxkcsp-runtime}
\end{equation}
queries. If $\sigma_B\geq\gamma_B$, then $c_B^{\mathrm{cl}}=\gamma_B$ and
\begin{equation}
    Q\leq2^{(1-c_B^{\mathrm{cl}}-\kappa_B)n/2+o(n)},
    \label{eq:tilted-walk-maxkcsp-comparison}
\end{equation}
which is a super-quadratic speedup over conditioning-and-search.  The condition $\sigma_B\geq\gamma_B$ holds, in particular, for weighted exact-$k$ CSPs with $k\geq6$.
\end{theorem}

\begin{proof}
At least $n/2$ coordinates belong to $L_{\mathrm{light}}$. {Indeed, if more than $n/2$ coordinates had degree greater than $2d_{\mathrm{avg}}=2\Sigma/n$, then their degrees alone would sum to more than $\Sigma$, a contradiction.\par} Define $S_r=B_{L_{\mathrm{light}}}(x^\star,r)$ and use $B_{L_{\mathrm{light}}}(x,r)$ as the local-search neighborhood of $x$. Each light-coordinate flip decreases $H$ by at most $2\Gamma_{\max}\Sigma/n$. {Since $G=\frac12(1+(H+B)/R)$, the corresponding decrease in $G$ is at most $\Gamma_{\max}\Sigma/(nR)$ per flip. Therefore every $x\in S_r$ satisfies\par}
\begin{equation}
    {
    G(x)\geq1-\frac{\Gamma_{\max}\Sigma r}{nR}
    \geq1-\frac{H_{\max}}{16R}
    =1-\frac{\rho}{2}}.
    \label{eq:maxkcsp-inner-threshold}
\end{equation}
In particular, $S_r\subseteq T_\eta$. Counting assignments on any $n/2$ light coordinates gives
\begin{equation}
    |S_r|\geq2^{\sigma_B n-o(n)}.
    \label{eq:maxkcsp-quantum-ball-size}
\end{equation}
{
To justify the exponent, choose a subset $J'\subseteq L_{\mathrm{light}}$ of size $N=\lfloor n/2\rfloor$. The restricted ball on $J'$ is contained in $S_r$, and
\begin{equation}
    \frac{r}{N}
    =\frac{H_{\max}}{8\Gamma_{\max}\Sigma}+O(1/n).
\end{equation}
Since this ratio lies in $[0,1/2]$, the standard entropy estimate gives
\begin{align}
    |S_r|
    &\geq\sum_{j=0}^r\binom Nj\\
    &\geq
    2^{Nh(H_{\max}/(8\Gamma_{\max}\Sigma))-o(n)}
    =2^{\sigma_Bn-o(n)}.
\end{align}
The inclusion $S_r\subseteq T_\eta$ follows from $G(x)\geq1-\rho/2$ and $\rho/2<\eta$.
\par}

{
As in the preceding application, the threshold can be expressed in terms of the centered objective:
\begin{equation}
    G(x)\geq1-\eta
    \iff H(x)\geq\frac{H_{\max}}2.
\end{equation}
Under a uniformly random assignment, $\mathbb E[H]=0$, and changing coordinate $i$ changes $H$ by at most $\Gamma_{\max}d_i$. McDiarmid's inequality therefore gives
\begin{equation}
    \Pr\!\left[H\geq\frac{H_{\max}}2\right]
    \leq
    \exp\!\left(
       -\frac{2(H_{\max}/2)^2}
              {\Gamma_{\max}^2\sum_i d_i^2}
    \right)
    =\exp\!\left(
       -\frac{H_{\max}^2n}
              {2\Gamma_{\max}^2\nu\Sigma^2}
    \right).
\end{equation}
Multiplying by $2^n$ and converting to base two yields
\begin{equation}
    |T_\eta|\leq2^{(1-\gamma_B)n},
    \qquad
    \gamma_B=\frac{H_{\max}^2}
       {2(\log2)\Gamma_{\max}^2\nu\Sigma^2}.
    \label{eq:maxkcsp-quantum-threshold-size}
\end{equation}
\par}

The bounds on $\mathcal L_2(G)$ and $\omega=1/n$ imply
\begin{equation}
    \frac{\mathcal L_2(G)^2}{8\omega^2}
    {=\frac{\Gamma_{\max}^2\nu\Sigma^2}{32R^2}}.
\end{equation}
{
We next verify the hypotheses of \cref{thm:normalized-amplitude-gain}. The normalization estimates above give $nL_B=\Gamma_{\max}\sqrt\nu\Sigma/R=\Theta(1)$, while $1-\mu=\Theta(1)$. Hence $L_B=\Theta(1/n)$ and, for all sufficiently large $n$,
\begin{align}
    \eta-\rho-L_B
       &=\frac{1-\mu}{4}-L_B>0,\\
    1-\rho-\mu-L_B
       &=\frac{3(1-\mu)}4-L_B>0.
\end{align}
Substituting $L=L_B$, $\omega=1/n$, and
\begin{equation}
    \frac{\mathcal L_2(G)^2}{8\omega^2}
    =\frac{\Gamma_{\max}^2\nu\Sigma^2}{32R^2}
\end{equation}
into the formula for $b^{*}$ in that theorem gives exactly \eqref{eq:maxkcsp-b}. Since $x^\star\in S_r$ and $G(x^\star)=1$, the theorem uses
\begin{equation}
    M=\left\lceil\frac{\log(1/(1-\eta))}{L_B}\right\rceil=O(n).
\end{equation}
By \eqref{eq:maxkcsp-inner-threshold},
\begin{equation}
    m(x)\geq\frac{1}{L_B}
    \log\left(\frac{1-\rho/2}{1-\rho}\right)
    \qquad(x\in S_r).
\end{equation}
Thus
\begin{equation}
    \mathbb E_{x\sim\pi_{S_r}}
       \left[\log\left(\frac{G(x)}{1-\rho}\right)\right]
    \geq
    \log\left(\frac{1-\rho/2}{1-\rho}\right),
\end{equation}
and substitution in the theorem's expression for $\kappa$ gives exactly $\kappa_B$ in~\eqref{eq:maxkcsp-kappa}. Since $\eta-\rho-L_B=\Theta(1)$ and $M=O(n)$, the remaining prefactor is inverse polynomial. Therefore the amplitude hypothesis \eqref{eq:optimization-amplitude-hypothesis} holds with $S_r$, $\kappa=\kappa_B$, and a polynomial $p$.
\par}

{
Substituting~\eqref{eq:maxkcsp-quantum-ball-size} and \eqref{eq:maxkcsp-quantum-threshold-size} into \eqref{eq:optimization-runtime}, the threshold-state preparation term is at most
\begin{equation}
    2^{(1-\sigma_B-\kappa_B)n/2+o(n)},
\end{equation}
and the local-search term is at most
\begin{equation}
    2^{(1-\gamma_B-\kappa_B)n/2}.
\end{equation}
Their sum is therefore
\begin{equation}
    2^{(1-\min\{\gamma_B,\sigma_B\}-\kappa_B)n/2+o(n)},
\end{equation}
which proves~\eqref{eq:tilted-walk-maxkcsp-runtime}.
\par}

In the minimization convention of Ref.~\cite{gall2026dequantizingshortpathquantumalgorithms}, the same threshold has parameter $1/2$ and $|H_{\min}|=H_{\max}$. Hence
\begin{equation}
    c_B^{\mathrm{cl}}
    =\min\left\{\gamma_B,
      \frac12h\left(\frac{H_{\max}}
      {2\Gamma_{\max}\Sigma}\right)\right\}.
\end{equation}
If $\sigma_B\geq\gamma_B$, the second entry is also at least $\gamma_B$, which proves~\eqref{eq:tilted-walk-maxkcsp-comparison}. {Indeed, the argument of the entropy function in the second entry is four times the argument defining $\sigma_B$, and both lie in $[0,1/2]$; monotonicity of $h$ therefore gives
\begin{equation}
    \frac12h\left(\frac{H_{\max}}
       {2\Gamma_{\max}\Sigma}\right)
    \geq\sigma_B\geq\gamma_B.
\end{equation}
Thus $c_B^{\mathrm{cl}}=\gamma_B$, and the quantum exponent contains the additional improvement $\kappa_B$.\par} Finally, for an exact-$k$ instance, $\Sigma=kW$ and $H_{\max}\leq W$. Writing $x=H_{\max}/(\Gamma_{\max}\Sigma)$ gives $x\leq1/k$. Since $h(t)\geq2t$, $\nu\geq1$, and $1/k\leq(\log2)/4$ for $k\geq6$,
\begin{equation}
    \sigma_B=\frac12h(x/8)\geq\frac{x}{8}
    \geq\frac{x^2}{2(\log2)\nu}=\gamma_B.
\end{equation}
{
The middle inequality is equivalent to $x\leq(\log2)\nu/4$, which follows from $x\leq1/k\leq(\log2)/4$ and $\nu\geq1$. This proves $\sigma_B\geq\gamma_B$ for weighted exact-$k$ instances with $k\geq6$. Finally, $nL_B=\Gamma_{\max}\sqrt\nu\Sigma/R=\Theta(1)$ under the theorem assumptions, so $b^{*}$ and every remaining factor in $\kappa_B$ are bounded above and below by positive constants for sufficiently large $n$. Hence $\kappa_B>0$ is independent of $n$, completing the claimed super-quadratic comparison.
\par}
\end{proof}

\section{Tilted Walks from Nonstationary Initial States}
\label{sec:warm-starts}
Quantum tilted walks can start from an application-specific state that is not stationary for the available mixer. This section proves a lower bound on the marked amplitude obtained from a normalized initial state with nonnegative amplitudes while retaining the original mixer $D$. The hypotheses measure the minimum fraction of the initial amplitude preserved pointwise by $D$ and the decrease of the objective under the transition kernel induced by $D$ and the initial state. Theorem~\ref{thm:warm-start-gain} assumes neither a log-Sobolev inequality nor a spectral-gap bound, although each application must separately control the normalization $E_b$.

Theorem~\ref{thm:warm-start-gain} complements Theorem~\ref{thm:normalized-amplitude-gain}. The former applies to initial states that need not be conditioned stationary states, while the latter can apply to threshold states whose one-step retention vanishes. After proving the new amplitude bound, we express its hypotheses for stationary distributions conditioned on a set, relate pointwise amplitude preservation to the top-eigenvalue gap of the mixer, and construct an example in which the initial stationary overlap and both relevant gaps are exponentially small. The example shows that the tilted-walk amplitude bound can remain useful when the analysis of the generalized short-path algorithm by Chakrabarti et al.~\cite{chakrabarti2025generalizedshortpathalgorithms} cannot certify efficient ground-state reflection.

\subsection{How the Mixer Acts on the Initial State}

The amplitude bound requires two one-step estimates. First, we compare the initial amplitudes with their image under $D$ to measure the smallest fraction retained at any point in the support. Second, we use the same comparison to define a transition kernel that describes how the objective values appearing in the path expansion change under a mixer step. This subsection defines both quantities and states the objective condition used in the theorem.

Let
\begin{equation}
    |\varphi\rangle
    =\sum_{x\in\mathcal X}\varphi(x)|x\rangle
\end{equation}
be a normalized state with $\varphi(x)\geq0$, and define its support by
\begin{equation}
    \Omega_\varphi=\{x:\varphi(x)>0\}.
\end{equation}
We compare $D\varphi$ with $\varphi$ pointwise on this support. Define
\begin{equation}
    \lambda_\varphi
    =\min_{x\in\Omega_\varphi}
      \frac{(D\varphi)(x)}{\varphi(x)}.
    \label{eq:warm-compatibility}
\end{equation}
Equivalently, $\lambda_\varphi$ is the largest constant such that $D\varphi\geq\lambda_\varphi\varphi$ on $\Omega_\varphi$. Since $D$ is entrywise nonnegative and has spectral norm one,
\begin{equation}
    0\leq\lambda_\varphi
    \leq\langle\varphi|D|\varphi\rangle
    \leq1.
\end{equation}
Thus $\lambda_\varphi$ is the minimum fraction of the initial amplitude retained under one application of the mixer. The stationary state satisfies $D\sqrt\pi=\sqrt\pi$ and hence $\lambda_{\sqrt\pi}=1$. We assume below that $\lambda_\varphi>0$.

For any function $h$ on $\Omega_\varphi$ and any $x\in\Omega_\varphi$, one has
\begin{equation}
    \sum_{y\in\Omega_\varphi}
       D_{x,y}\varphi(y)h(y)
    =(D\varphi)(x)(K_\varphi h)(x),
    \label{eq:warm-reweighted-identity}
\end{equation}
where
\begin{equation}
    K_\varphi(x,y)
    =\frac{D_{x,y}\varphi(y)}{(D\varphi)(x)},
    \qquad x\in\Omega_\varphi.
    \label{eq:warm-reweighted-kernel}
\end{equation}
For each $x\in\Omega_\varphi$, the entries in this row are nonnegative and sum to one. Moreover, $K_\varphi(x,y)=0$ when $y\notin\Omega_\varphi$. Thus $K_\varphi$ is a stochastic kernel on $\Omega_\varphi$. It depends on both $D$ and $|\varphi\rangle$ and need not equal the original Markov chain $P$.

Choose $L_\varphi\in[0,1]$ such that
\begin{equation}
    (K_\varphi G)(x)\geq(1-L_\varphi)G(x)
    \qquad\text{for every }x\in\Omega_\varphi.
    \label{eq:warm-score-regularity}
\end{equation}
This condition controls the decrease in the objective under one step of $K_\varphi$. For every integer $t\geq0$ and every $x\in\Omega_\varphi$,
\begin{equation}
    (K_\varphi^tG)(x)
    \geq(1-L_\varphi)^tG(x).
\end{equation}
The condition is imposed only on the support of the initial state.

\subsection{Amplitude Gain from a Nonnegative Initial State}

The one-step estimates above imply a product lower bound for every power of $D+bG$. We first prove this bound pointwise for each configuration. We then obtain lower bounds on the marked amplitude for both the single power $D_b^M$ and the average $C_{b,M}$.

\begin{theorem}[Tilted-walk gain from a nonnegative initial state]
\label{thm:warm-start-gain}
Assume that $\lambda_\varphi>0$ and that~\eqref{eq:warm-score-regularity} holds. Let $b\geq0$, let $M\geq0$ be an integer, and let $S\subseteq\Omega_\varphi$ be nonempty. Let $g\in[0,1]$ and suppose $G(x)\geq g$ for every $x\in S$. Then, for every integer $m\geq0$ and every $x\in\mathcal X$,
\begin{equation}
    \langle x|(D+bG)^m|\varphi\rangle
    \geq
    \varphi(x)
    \prod_{t=0}^{m-1}
         \left(\lambda_\varphi+b(1-L_\varphi)^tG(x)\right).
    \label{eq:warm-pointwise-gain}
\end{equation}
Consequently, if $\Pi_S=\sum_{x\in S}|x\rangle\!\langle x|$, then
\begin{equation}
    \left\|\Pi_SC_{b,M}|\varphi\rangle\right\|
    \geq
    \frac{\left\|\Pi_S|\varphi\rangle\right\|}{M+1}
    \prod_{t=0}^{M-1}
         \frac{\lambda_\varphi+b(1-L_\varphi)^tg}{E_b}.
    \label{eq:warm-marked-norm-gain}
\end{equation}
The same pointwise argument also implies
\begin{equation}
    \left\|\Pi_SD_b^M|\varphi\rangle\right\|
    \geq
    \left\|\Pi_S|\varphi\rangle\right\|
    \prod_{t=0}^{M-1}
         \frac{\lambda_\varphi+b(1-L_\varphi)^tg}{E_b}.
    \label{eq:warm-single-power-gain}
\end{equation}
\end{theorem}

\begin{proof}
Let $F_m=(D+bG)^m\varphi$ and define
\begin{equation}
    q=1-L_\varphi,
    \qquad
    \Phi_m(u)=\prod_{t=0}^{m-1}
      \left(\lambda_\varphi+bq^tu\right).
\end{equation}
The polynomial $\Phi_m$ has nonnegative coefficients and is therefore nonnegative, nondecreasing, and convex on $[0,1]$. We prove by induction that
\begin{equation}
    F_m(x)\geq\varphi(x)\Phi_m(G(x))
    \label{eq:warm-induction}
\end{equation}
for every $x$. The claim is an equality at $m=0$. Suppose it holds at $m$. For $x\in\Omega_\varphi$, positivity and Jensen's inequality imply the first two bounds below. Equation~\eqref{eq:warm-score-regularity}, monotonicity of $\Phi_m$, and $D\varphi\geq\lambda_\varphi\varphi$ imply the third:
\begin{align}
    (DF_m)(x)
    &\geq(D\varphi)(x)
       (K_\varphi\Phi_m(G))(x)\\
    &\geq(D\varphi)(x)
       \Phi_m((K_\varphi G)(x))\\
    &\geq\lambda_\varphi\varphi(x)
       \Phi_m(qG(x)).
\end{align}
The induction hypothesis, followed by monotonicity of $\Phi_m$ and $q\leq1$, also implies
\begin{equation}
    bG(x)F_m(x)
    \geq bG(x)\varphi(x)\Phi_m(G(x))
    \geq bG(x)\varphi(x)\Phi_m(qG(x)).
\end{equation}
Adding these inequalities yields
\begin{equation}
    F_{m+1}(x)
    \geq\varphi(x)
       (\lambda_\varphi+bG(x))\Phi_m(qG(x))
    =\varphi(x)\Phi_{m+1}(G(x)).
\end{equation}
For $x\notin\Omega_\varphi$, the right-hand side of Equation~\eqref{eq:warm-induction} is zero and the claim follows from positivity. This proves Equation~\eqref{eq:warm-pointwise-gain}.

Applying Equation~\eqref{eq:warm-pointwise-gain} with $m=M$ and $G(x)\geq g$ for $x\in S$, then squaring and summing over $S$, proves~\eqref{eq:warm-single-power-gain}.

All entries of $D+bG$ are nonnegative. For every $x\in S$, we may therefore discard the terms of $C_{b,M}$ with powers $0,\ldots,M-1$ and retain only its $M$-th term. Equation~\eqref{eq:warm-pointwise-gain} and the bound $G(x)\geq g$ then imply
\begin{equation}
    \langle x|C_{b,M}|\varphi\rangle
    \geq
    \frac{\varphi(x)}{M+1}
    \prod_{t=0}^{M-1}
    \frac{\lambda_\varphi+b(1-L_\varphi)^tg}{E_b}.
\end{equation}
Squaring this inequality and summing over $x\in S$ proves Equation~\eqref{eq:warm-marked-norm-gain}.
\end{proof}

An exponential gain follows when every factor in Equation~\eqref{eq:warm-marked-norm-gain} is bounded below by $1+\delta$ for a constant $\delta>0$. For example, fix constants $\ell,\tau,\delta>0$ and suppose that, for all sufficiently large $n$,
\begin{equation}
    L_\varphi\leq\frac{\ell}{n},
    \qquad
    M=\lfloor\tau n\rfloor,
    \qquad
    \frac{\lambda_\varphi+bge^{-2\ell\tau}}{E_b}
       \geq1+\delta,
    \label{eq:warm-positive-rate}
\end{equation}
Then
\begin{equation}
    \left\|\Pi_SC_{b,M}|\varphi\rangle\right\|
    \geq
    2^{\Gamma n-o(n)}
    \left\|\Pi_S|\varphi\rangle\right\|,
    \qquad
    \Gamma=\tau\log_2(1+\delta)>0.
    \label{eq:warm-exponential-rate}
\end{equation}
Indeed, $(1-L_\varphi)^t\geq e^{-2\ell\tau}$ for $0\leq t<M$ and all sufficiently large $n$. The product in~\eqref{eq:warm-marked-norm-gain} is therefore at least $(1+\delta)^M$, while $1/(M+1)=2^{-o(n)}$.
Condition~\eqref{eq:warm-positive-rate} is a sufficient positive-rate condition. Verifying it requires an upper bound on $E_b$, but no lower bound on the spectral gap of the tilted Hamiltonian. Equation~\eqref{eq:warm-single-power-gain} yields the same exponential rate for $D_b^M$ without the factor $1/(M+1)$.

If an application supplies a polynomial-size list of pairs $(b,M)$ and a common running-time bound sufficient for at least one pair that attains the required marked-amplitude lower bound, Section~\ref{sec:implementation-details} proves that testing all pairs adds only polynomial overhead.

\subsection{Conditioned Stationary States}

The quantities $\lambda_\varphi$ and $K_\varphi$ have a direct interpretation when the initial state is the coherent encoding of $\pi$ conditioned on a set $A$. The next proposition expresses them using the original transition matrix $P$. These formulas identify the one-step retention condition and the required bound on objective values under the conditioned transition.

\begin{proposition}[Conditioned stationary states]
\label{prop:conditioned-warm-start}
Let $A\subseteq\mathcal X$ be nonempty, define $\pi(A)=\sum_{x\in A}\pi(x)$, and let
\begin{equation}
    |A\rangle
    =\frac{1}{\sqrt{\pi(A)}}
      \sum_{x\in A}\sqrt{\pi(x)}\,|x\rangle.
\end{equation}
Let $\lambda_A$ and $K_A$ denote the quantities in~\eqref{eq:warm-compatibility} and~\eqref{eq:warm-reweighted-kernel} for the initial state $|A\rangle$. Assume $P(x,A)>0$ for every $x\in A$. Then, for every $x\in A$,
\begin{align}
    \frac{\langle x|D|A\rangle}{\langle x|A\rangle}
       &=P(x,A),\\
    K_A(x,y)
       &=\frac{P(x,y)}{P(x,A)}\mathbf1_{\{y\in A\}}.
\end{align}
Consequently,
\begin{equation}
    \lambda_A=\min_{x\in A}P(x,A).
    \label{eq:conditioned-warm-compatibility}
\end{equation}
\end{proposition}

\begin{proof}
For $x\in A$, reversibility implies
\begin{align}
    \langle x|D|A\rangle
    &=\sum_{y\in A}
      \sqrt{\frac{\pi(x)}{\pi(y)}}P(x,y)
      \frac{\sqrt{\pi(y)}}{\sqrt{\pi(A)}}\\
    &=\frac{\sqrt{\pi(x)}}{\sqrt{\pi(A)}}P(x,A).
\end{align}
The identities follow from Equations~\eqref{eq:warm-compatibility} and \eqref{eq:warm-reweighted-kernel}.
\end{proof}

The proposition expresses the quantities in both hypotheses of Theorem~\ref{thm:warm-start-gain} using the original chain. The parameter $\lambda_A$ is the minimum probability of remaining in $A$ for one step, while $K_A$ is the transition kernel $P$ conditioned to remain in $A$ for that step. In particular, the objective condition~\eqref{eq:warm-score-regularity} becomes
\begin{equation}
    \frac{\sum_{y\in A}P(x,y)G(y)}{P(x,A)}
    \geq(1-L_A)G(x)
    \qquad\text{for every }x\in A.
    \label{eq:conditioned-warm-score-regularity}
\end{equation}
Thus the tilted walk can retain the original mixer while starting from $|A\rangle$. Applying the generalized short-path construction of Chakrabarti et al.~\cite{chakrabarti2025generalizedshortpathalgorithms} with $|A\rangle$ as its initial ground state would instead require a reversible chain with stationary distribution $\pi(\cdot\mid A)$, together with the spectral and regularity bounds required by that construction.

The two amplitude-gain theorems cover different cases. For the conditioned state $|A\rangle$, Theorem~\ref{thm:warm-start-gain} requires $P(x,A)>0$ throughout $A$ and is strongest when $\min_{x\in A}P(x,A)$ is close to one. Theorem~\ref{thm:normalized-amplitude-gain} can apply to a threshold state even when some configurations leave its support with probability one, because its proof controls paths that enter the threshold set rather than requiring uniform retention within it.

\subsection{Stationary Overlap and the Top-Eigenvalue Gap}

The next lemma shows that small stationary overlap and $\lambda_\varphi$ close to one can coexist only when the top-eigenvalue gap of $D$ is small.

\begin{lemma}[Stationary overlap and top-eigenvalue gap]
\label{lem:warm-gap-tradeoff}
Let $|\varphi\rangle=\sum_x\varphi(x)|x\rangle$ be normalized with $\varphi(x)\geq0$, and let $|\sqrt\pi\rangle=\sum_x\sqrt{\pi(x)}|x\rangle$. Suppose every eigenvalue of $D$ on $|\sqrt\pi\rangle^\perp$ is at most $1-\delta$ for some $\delta>0$. Then
\begin{equation}
    1-|\langle\sqrt\pi|\varphi\rangle|^2
    \leq\frac{1-\lambda_\varphi}{\delta}.
    \label{eq:warm-gap-tradeoff}
\end{equation}
\end{lemma}

\begin{proof}
The pointwise inequality $D\varphi\geq\lambda_\varphi\varphi$ on $\Omega_\varphi$ and the nonnegativity of $\varphi$ imply
\begin{equation}
    \langle\varphi|D|\varphi\rangle
    \geq\lambda_\varphi.
\end{equation}
Let $a=\langle\sqrt\pi|\varphi\rangle$. The spectral assumption also implies
\begin{equation}
    \langle\varphi|D|\varphi\rangle
    \leq |a|^2+(1-|a|^2)(1-\delta)
    =1-\delta(1-|a|^2).
\end{equation}
Combining the two inequalities and rearranging proves~\eqref{eq:warm-gap-tradeoff}.
\end{proof}

For example, suppose
\begin{equation}
    |\langle\sqrt\pi|\varphi\rangle|^2\leq\alpha<1,
    \qquad
    \lambda_\varphi\geq1-\varepsilon.
\end{equation}
Then Equation~\eqref{eq:warm-gap-tradeoff} implies
\begin{equation}
    \delta\leq\frac{\varepsilon}{1-\alpha}.
\end{equation}
Thus an initial state can have both small stationary overlap and little pointwise amplitude loss only when the top-eigenvalue gap is correspondingly small. This is a necessary condition, not a construction of such a state.

\subsection{A Concrete Comparison with the Generalized Short-Path Algorithm}

The following example realizes the regime identified by Lemma~\ref{lem:warm-gap-tradeoff}. The sector register creates an exponentially small top-eigenvalue gap and allows the initial state to remain concentrated in one sector. The work register carries a Hamming-weight objective that changes by only $O(1/n)$ under one lazy hypercube step. The work-register transitions will imply $L_{\varphi_n}=O(1/n)$ in~\eqref{eq:warm-score-regularity}, while the sector register will make both the spectral gap of the mixer and the spectral gap of the tilted Hamiltonian exponentially small.

Let
\begin{equation}
    \mathcal X_n
    =\{0,1\}^n_{\mathrm{sector}}
       \times\{0,1\}^n_{\mathrm{work}}.
\end{equation}
Write $|+\rangle=2^{-n/2}\sum_{z\in\{0,1\}^n}|z\rangle$. On the sector register, define
\begin{equation}
    R_n=(1-\epsilon_n)I+\epsilon_n|+\rangle\!\langle+|,
    \qquad
    \epsilon_n=2^{-n},
\end{equation}
so that a uniform resampling step occurs with probability $\epsilon_n$. On the work register, define the lazy hypercube walk
\begin{equation}
    D_{\mathrm{hc}}
    =\frac12I+\frac1{2n}\sum_{i=1}^nX_i.
\end{equation}
Both $R_n$ and $D_{\mathrm{hc}}$ are symmetric stochastic matrices. Set $\tau_n=1/n$ and combine the two walks as
\begin{equation}
    D_n
    =\tau_nR_n\otimes I
      +(1-\tau_n)I\otimes D_{\mathrm{hc}}.
    \label{eq:metastable-mixer}
\end{equation}
The matrix $D_n$ is therefore both the transition matrix and the discriminant of a reversible chain with uniform stationary distribution $\pi_n$. The spectral gaps associated with the sector and work registers are $\tau_n\epsilon_n$ and $(1-\tau_n)/n$, respectively. Hence the top-eigenvalue gap of $D_n$ is
\begin{equation}
    \delta_n=\tau_n\epsilon_n=\frac{2^{-n}}{n}.
\end{equation}

Instead of the stationary state, consider
\begin{equation}
    |\varphi_n\rangle
    =|0^n\rangle_{\mathrm{sector}}|+\rangle_{\mathrm{work}}.
\end{equation}
It is confined to one sector and uniform on the work register. Its squared overlap with the stationary state is
\begin{equation}
    |\langle\sqrt{\pi_n}|\varphi_n\rangle|^2=2^{-n}.
\end{equation}
Since $\langle0^n|R_n|0^n\rangle=1-\epsilon_n+\epsilon_n2^{-n}$, the ratio in~\eqref{eq:warm-compatibility} is the same at every point in the support and equals
\begin{equation}
    \lambda_{\varphi_n}
    =1-\tau_n\epsilon_n(1-2^{-n})
    =1-O(2^{-n}/n).
\end{equation}
Thus the initial state has exponentially small stationary overlap, while $\lambda_{\varphi_n}$ differs from one by only $O(2^{-n}/n)$.
In fact, the state saturates~\eqref{eq:warm-gap-tradeoff}:
\begin{equation}
    \frac{1-\lambda_{\varphi_n}}{\delta_n}
    =1-2^{-n}
    =1-|\langle\sqrt{\pi_n}|\varphi_n\rangle|^2.
    \label{eq:metastable-gap-saturation}
\end{equation}

Define the normalized work-register objective
\begin{equation}
    \widetilde G_n(w)=\frac12\left(1+\frac{|w|}{n}\right)
\end{equation}
and let $G_n=I\otimes\widetilde G_n$ be the objective on $\mathcal X_n$. This linear Hamming-weight objective changes by only $O(1/n)$ under one lazy hypercube step, which will yield $L_{\varphi_n}=O(1/n)$. The successful set is
\begin{equation}
    S_n=\{0^n\}_{\mathrm{sector}}
        \times\{w:|w|=3n/4\},
\end{equation}
where $4$ divides $n$.

\begin{proposition}[Exponential amplitude gain when the mixer and tilted Hamiltonian have exponentially small spectral gaps]
\label{prop:metastable-warm-separation}
For every sufficiently large $n$ divisible by $4$, let $b=1/2$ and $M=\lfloor n/20\rfloor$. There is a constant $\Gamma_0>0$ such that
\begin{equation}
    \left\|\Pi_{S_n}C_{b,M}|\varphi_n\rangle\right\|
    \geq
    2^{\Gamma_0n-o(n)}
    \left\|\Pi_{S_n}|\varphi_n\rangle\right\|.
    \label{eq:metastable-warm-gain}
\end{equation}
The spectral gap of the mixer $D_n$ is $2^{-n}/n$, and the spectral gap of the tilted Hamiltonian $D_n+bG_n$ is at most the same quantity. Thus the hypothesis that the spectral gap of the tilted Hamiltonian is inverse polynomial in the analysis of the generalized short-path algorithm by Chakrabarti et al.~\cite{chakrabarti2025generalizedshortpathalgorithms} fails, while the tilted-walk amplitude bound remains exponential.
\end{proposition}

\begin{proof}
Since $\lambda_{\varphi_n}\geq1-\tau_n$, define
\begin{equation}
    a_n:=\frac{1-\tau_n}{\lambda_{\varphi_n}}\leq1.
\end{equation}
After identifying the support of $|\varphi_n\rangle$ with the work register, the reweighted kernel is
\begin{equation}
    K_{\varphi_n}=(1-a_n)I+a_nD_{\mathrm{hc}}.
\end{equation}
For the objective on the work register, a direct calculation shows
\begin{equation}
    (D_{\mathrm{hc}}\widetilde G_n)(w)
    =\widetilde G_n(w)+\frac{1-2|w|/n}{4n}
    \geq\left(1-\frac1{4n}\right)\widetilde G_n(w).
\end{equation}
Consequently,
\begin{equation}
    (K_{\varphi_n}G_n)(s,w)
    \geq\left(1-\frac{a_n}{4n}\right)G_n(s,w)
    \geq\left(1-\frac1{4n}\right)G_n(s,w).
\end{equation}
Equation~\eqref{eq:warm-score-regularity} therefore holds with $L_{\varphi_n}=1/(4n)$. Moreover, $G_n=7/8$ throughout $S_n$.

The objective acts trivially on the sector register. On the work register,
\begin{equation}
    (1-\tau_n)D_{\mathrm{hc}}+b\widetilde G_n
    =\left(\frac{1-\tau_n}{2}+\frac{3b}{4}\right)I
     +\frac1n\sum_{i=1}^n
      \left(\frac{1-\tau_n}{2}X_i-\frac{b}{4}Z_i\right).
\end{equation}
The largest eigenvalue of $R_n$ is one. Diagonalizing the independent one-qubit terms above therefore yields
\begin{equation}
    E_b
    =\tau_n+\frac{1-\tau_n}{2}+\frac{3b}{4}
      +\frac12\sqrt{(1-\tau_n)^2+\frac{b^2}{4}}.
    \label{eq:metastable-normalization}
\end{equation}
For $b=1/2$, Equation~\eqref{eq:metastable-normalization} implies
\begin{equation}
    E_{1/2}=\frac{7+\sqrt{17}}{8}+O(1/n).
\end{equation}
Define
\begin{equation}
    \delta_0
    :=\frac{1+(7/16)e^{-1/40}}{(7+\sqrt{17})/8}-1>0.
\end{equation}
For $0\leq t<M$, we also have $(1-1/(4n))^t\geq e^{-1/40}$ for all sufficiently large $n$. Each factor in Equation~\eqref{eq:warm-marked-norm-gain} is therefore at least
\begin{equation}
    \frac{1+(7/16)e^{-1/40}-o(1)}{(7+\sqrt{17})/8+o(1)}
    =1+\delta_0-o(1).
\end{equation}
For all sufficiently large $n$, every factor is at least $1+\delta_0/2$. Since $M=\lfloor n/20\rfloor$, Theorem~\ref{thm:warm-start-gain} proves~\eqref{eq:metastable-warm-gain} with
\begin{equation}
    \Gamma_0
    =\frac1{20}\log_2\left(1+\frac{\delta_0}{2}\right)>0,
\end{equation}
after the inverse-polynomial factor $1/(M+1)$ is absorbed into the $2^{-o(n)}$ term.

It remains to bound the spectral gap of the tilted Hamiltonian. Write
\begin{equation}
    D_n+bG_n
    =\tau_nR_n\otimes I
        +I\otimes\bigl((1-\tau_n)D_{\mathrm{hc}}+b\widetilde G_n\bigr).
\end{equation}
Let $|u_b\rangle$ be a principal eigenvector of the work-register operator. Then $|+\rangle|u_b\rangle$ is a principal eigenvector of $D_n+bG_n$. If $|v\rangle$ is any sector state orthogonal to $|+\rangle$, then $R_n|v\rangle=(1-\epsilon_n)|v\rangle$, so $|v\rangle|u_b\rangle$ has eigenvalue $E_b-\tau_n\epsilon_n$. The gap between the two largest eigenvalues is consequently at most
\begin{equation}
    \tau_n\epsilon_n=\frac{2^{-n}}{n}.
\end{equation}
\end{proof}

We now compare the tilted operator in the proposition with the generalized short-path Hamiltonian. Set $H_n=-G_n$, so the optimum energy is $E^\star=-1$. Since $G_n\in[1/2,1]$, the filter with parameter $\eta=1/2$ satisfies
\begin{equation}
    g_{1/2}\left(\frac{H_n}{|E^\star|}\right)
    =g_{1/2}(-G_n)
    =I-2G_n.
\end{equation}
Writing $\beta$ for the generalized short-path tilt, its Hamiltonian is
\begin{equation}
    H_\beta^{\mathrm{sp}}
    =-D_n+\beta(I-2G_n).
\end{equation}
At $\beta=1/4$,
\begin{equation}
    D_n+\frac12G_n
    =-H_{1/4}^{\mathrm{sp}}+\frac14I.
    \label{eq:metastable-short-path-correspondence}
\end{equation}
Negation and addition of a scalar multiple of the identity preserve eigenvalue-gap magnitudes. The spectral gap of the tilted Hamiltonian $H_{1/4}^{\mathrm{sp}}$ is therefore at most $2^{-n}/n$.

The initial marked probability is
\begin{equation}
    \left\|\Pi_{S_n}|\varphi_n\rangle\right\|^2
    =\frac{\binom n{3n/4}}{2^n}.
\end{equation}
For the initial state $|\varphi_n\rangle$, Theorem~\ref{thm:warm-start-gain} certifies the exponential ratio in~\eqref{eq:metastable-warm-gain} using $M=\Theta(n)$. Equation~\eqref{eq:metastable-short-path-correspondence} identifies the corresponding generalized short-path Hamiltonian. Because both the spectral gap of the mixer and the spectral gap of the tilted Hamiltonian are exponentially small, the analysis by Chakrabarti et al.~\cite{chakrabarti2025generalizedshortpathalgorithms} cannot certify a polynomial-cost ground-state reflection. The example therefore compares the combined sufficient conditions of the two analyses; it does not prove that every generalized short-path implementation is slow. An element of $S_n$ can be written down directly by choosing the sector string $0^n$ and any work string of Hamming weight $3n/4$, so the proposition does not establish a quantum speedup for a hard optimization problem.

\section{Implementation Details}
\label{sec:implementation-details}
The goal of this section is to give implementation details about the algorithm by accounting for the errors introduced at each step of the implementation.
\subsection{Constructing the block encoding}
{
Fix an implementation precision $\epsilon_G>0$. We construct block encodings of $D$ and $G$ from the two access oracles in Equation~\eqref{eq:primitive-access-oracles}. For the discriminant, let
\begin{equation}
    U_D=U_P^\dagger\operatorname{SWAP}U_P.
\end{equation}
For every $x,y\in\mathcal X$,
\begin{equation}
    \langle x|\langle0|U_D|y\rangle|0\rangle
    =\sqrt{P(x,y)P(y,x)}
    =D_{x,y}.
\end{equation}
Thus $U_D$ is an exact unit-normalized block encoding of $D$ and uses one call each to $U_P$ and $U_P^\dagger$.

For the diagonal objective, query the exact oracle $O_f$ in Equation~\eqref{eq:primitive-access-oracles} and use reversible fixed-point arithmetic to compute
\begin{equation}
    G(x)=\frac12\left(1+\frac{f(x)+A_c}{R_c}\right)
\end{equation}
to an approximation $\widetilde G(x)$ satisfying
\begin{equation}
    |\widetilde G(x)-G(x)|\leq\epsilon_G.
\end{equation}
This computation and its inverse use a constant number of queries to $O_f$ and $O_f^\dagger$. A reversible controlled rotation then maps a signal qubit as
\begin{equation}
    |\widetilde G(x)\rangle|0\rangle
    \longmapsto
    |\widetilde G(x)\rangle
    \left(\widetilde G(x)|0\rangle
    +\sqrt{1-\widetilde G(x)^2}|1\rangle\right).
\end{equation}
Uncomputing the value register yields a $(1,p+1,O(\epsilon_G))$ block encoding of $G$. An LCU construction applied to the encodings of $D$ and $G$ then gives an $(1+b,a_D+p+O(1),b\epsilon_G)$ block encoding of $H_b=D+bG$.

Let $E_b^{\mathrm{ub}}\geq E_b$ be the known upper bound specified in Section~\ref{subsec:resolving-unknowns}, fix a slack parameter $\delta_{\mathrm{sv}}\in(0,1]$, and define
\begin{equation}
    \widehat E_b:=(1+\delta_{\mathrm{sv}})E_b^{\mathrm{ub}},
    \qquad
    D_b^{(\delta_{\mathrm{sv}})}
       :=\frac{D+bG}{\widehat E_b}.
    \label{eq:slack-normalized-Db}
\end{equation}
Then
\begin{equation}
    \bigl\|D_b^{(\delta_{\mathrm{sv}})}\bigr\|
    \leq\frac{1}{1+\delta_{\mathrm{sv}}}.
\end{equation}
The LCU circuit is an $(\alpha_b,a_D+p+O(1),b\epsilon_G/\widehat E_b)$ block encoding of $D_b^{(\delta_{\mathrm{sv}})}$, where
\begin{equation}
    \alpha_b:=\frac{1+b}{\widehat E_b}\leq2.
\end{equation}
Choose the arithmetic precision so that
\begin{equation}
    b\epsilon_G\leq\frac{\delta_{\mathrm{sv}}E_b^{\mathrm{ub}}}{4}.
\end{equation}
If $\alpha_b\leq1$, one additional controlled rotation attenuates the signal block by $\alpha_b$ and gives a unit-normalized block encoding. If $\alpha_b>1$, the singular values of the unamplified signal block are at most
\begin{equation}
    \frac{E_b+b\epsilon_G}{1+b}
    \leq
    \frac{(1+\delta_{\mathrm{sv}}/4)E_b^{\mathrm{ub}}}{1+b}
    =\frac{1+\delta_{\mathrm{sv}}/4}
      {(1+\delta_{\mathrm{sv}})\alpha_b}.
\end{equation}
The displayed bound leaves a relative margin $\Theta(\delta_{\mathrm{sv}})$. Uniform singular-value amplification produces a unit-normalized block encoding with degree
\begin{equation}
    d_{\mathrm{sv}}
    =O\!\left(
       \frac{\alpha_b}{\delta_{\mathrm{sv}}}
       \log\frac{\alpha_b}{\epsilon_b}
    \right)
    =O\!\left(
       \frac{1}{\delta_{\mathrm{sv}}}
       \log\frac{1}{\epsilon_b}
    \right).
    \label{eq:sv-amplification-degree}
\end{equation}
Choose the polynomial-approximation and remaining arithmetic errors so that their total is at most $\epsilon_b$. Robust singular-value transformation then gives a Hermitian $(1,a_b,\epsilon_b)$ block encoding $U_b$ of $D_b^{(\delta_{\mathrm{sv}})}$ using
\begin{equation}
    Q_b(\epsilon_b)
    =O\!\left(
       \frac{1}{\delta_{\mathrm{sv}}}
       \log\frac{1}{\epsilon_b}
    \right)
    \label{eq:Db-block-query-cost}
\end{equation}
calls to the LCU block encoding and its adjoint. Controlled versions and adjoints have the same asymptotic query cost~\cite{gilyen2019qsvt}.

For the remainder of this section, $D_b$ denotes the implemented contraction $D_b^{(\delta_{\mathrm{sv}})}$, and $C_{b,M}$ denotes the corresponding average of its powers. The comparison with the analytically normalized operator is given in Section~\ref{subsec:resolving-unknowns}.\par}

We start with the implementation of the block-encoding unitary $U_{C}$ for the operator $C_{b,M}$. In this subsection, we construct $U_C$ from a block encoding $U_b$ of $D_b$ by applying the polynomial
\begin{equation}
    p_M(x)=\frac{1}{M+1}\sum_{m=0}^{M}x^m,
    \label{eq:implementation-cbm}
\end{equation}
to the operator $D_b$ since $C_{b,M}=p_M(D_b)$. We recall that $(\alpha,a,\epsilon)$ block encoding of an operator $A$ is a unitary $U_A$ such that
\begin{equation}
    \left\|A-\alpha
    (\langle0^a|\otimes I)U_A(|0^a\rangle\otimes I)\right\|
    \leq\epsilon.
\end{equation}
{
For the fixed value of $b$ considered in this section, $D_b$ denotes the slack-normalized contraction in~\eqref{eq:slack-normalized-Db}; in particular, $\|D_b\|\leq(1+\delta_{\mathrm{sv}})^{-1}<1$. We use the Hermitian $(1,a_b,\epsilon_b)$ block encoding $U_b$ constructed above, meaning that
\begin{equation}
    \widetilde D_b
    =(\langle0^{a_b}|\otimes I)U_b(|0^{a_b}\rangle\otimes I)
    \label{eq:implementation-normalized-access}
\end{equation}
is Hermitian and satisfies $\|\widetilde D_b-D_b\|\leq\epsilon_b$. Since $U_b$ is unitary, we also have $\|\widetilde D_b\|\leq1$. We also assume that a controlled version of $U_b$ is available.\par}
Let $Q_b(\epsilon_b)$ denote the query cost of $U_b$. It satisfies the bound in~\eqref{eq:Db-block-query-cost}. We use the following consequence of quantum singular value transformation: given a block encoding of a Hermitian operator $A$ satisfying $\|A\|\leq1$ and a real polynomial $q$ containing only even powers or only odd powers and satisfying $|q(x)|\leq1$ for $x\in[-1,1]$, one can construct a block encoding of $q(A)$ using $O(\deg q)$ calls to the block encoding of $A$ and its inverse~\cite{gilyen2019qsvt}.

\begin{theorem}[Block encoding of $C_{b,M}$]
\label{thm:block-encoding-cbm}
Let $U_b$ be a $(1,a_b,\epsilon_b)$ block encoding of $D_b$. Given access to $U_b$ and $U_b^{\dagger}$, for every integer $M\geq 1$, there exists an algorithm that implements a  $(1,a_b+O(1),\epsilon_b M)$ block encoding of $C_{b,M}$, where the construction uses $O(M)$ calls to $U_b$ and $U_b^\dagger$. 
\end{theorem}

\begin{proof}
For $x\in[-1,1]$, it holds that 
\begin{equation}
    |p_M(x)|
    \leq\frac{1}{M+1}\sum_{m=0}^{M}|x|^m
    \leq1.
\end{equation}
Thus $p_M$ is a degree-$M$ real polynomial bounded by $1$ on the spectrum of $\widetilde D_b$. Since $D_b$ is Hermitian, singular value transformation and eigenvalue transformation coincide.  Quantum singular transformation in~\cite{gilyen2019qsvt} requires a polynomial containing only even powers or only odd powers, so we separate these two parts. For $s\in\{0,1\}$, let $N_s$ be the number of integers $m\in\{0,\ldots,M\}$ with $m\equiv s\pmod 2$, and define
\begin{equation}
    q_s(x)=\frac{1}{N_s}
    \sum_{\substack{0\leq m\leq M\\m\equiv s\ ({\rm mod}\ 2)}}x^m,
    \qquad
    \alpha_s=\frac{N_s}{M+1}.
    \label{eq:implementation-parity-decomposition}
\end{equation}
Then $q_0$ is even, $q_1$ is odd, $|q_s(x)|\leq1$ for $x\in[-1,1]$, and
\begin{equation}
    p_M(x)=\alpha_0q_0(x)+\alpha_1q_1(x),
    \qquad \alpha_0+\alpha_1=1.
\end{equation}
Apply polynomial eigenvalue transformation separately to $q_0$ and $q_1$, obtaining block encodings of $q_0(\widetilde D_b)$ and $q_1(\widetilde D_b)$. Prepare one ancilla qubit in $\sqrt{\alpha_0}|0\rangle+\sqrt{\alpha_1}|1\rangle$, apply the first block encoding when the ancilla is $|0\rangle$ and the second when it is $|1\rangle$, and then reverse the ancilla preparation. The encoded operator is
\begin{equation}
    \alpha_0q_0(\widetilde D_b)+\alpha_1q_1(\widetilde D_b)
    =p_M(\widetilde D_b).
\end{equation}
This uses one additional ancilla and $O(M)$ calls to $U_b$ and $U_b^\dagger$. Choose the two polynomial transformations so that their combined operator-norm error is at most $\frac{M\epsilon_b }{2}$.

It remains to compare $p_M(\widetilde D_b)$ with $p_M(D_b)$. Since $\|\widetilde D_b\|\leq1$, {$\|D_b\|\leq1$}, and $\|\widetilde D_b-D_b\|\leq\epsilon_b$, for every $m\geq1$,
\begin{align}
    \|\widetilde D_b^m-D_b^m\|
    &\leq\sum_{\ell=0}^{m-1}
      \|\widetilde D_b^{m-1-\ell}
      (\widetilde D_b-D_b)D_b^\ell\|\\
    &\leq m\epsilon_b.
\end{align}
Consequently,
\begin{align}
    \|p_M(\widetilde D_b)-p_M(D_b)\|
    &\leq\frac{1}{M+1}\sum_{m=0}^M m\epsilon_b
      =\frac{M\epsilon_b}{2}.
\end{align}
Adding the two sources of error proves the theorem.
\end{proof}

\subsection{Implementation of the Quantum Conditioning-and-Search}

Apply amplitude amplification to the uniform state $H^{\otimes n}|0^n\rangle$, marking the basis states in $T_\eta$. Let $U_{T_\eta}$ denote the resulting state-preparation unitary. Then
\begin{equation}
    U_{T_\eta}|0^n\rangle=|T_\eta\rangle,
\end{equation}
and $U_{T_\eta}^\dagger$ is implemented by reversing the same circuit. The number of queries to $O_f$ required to implement $U_{T_\eta}$ is
\begin{equation}
    Q_{T_\eta}=O\!\left(\frac{2^{n/2}}{|T_\eta|^{1/2}}\right)
\end{equation}
up to logarithmic factors. Let $U_C$ be the $(1,a_C,\epsilon_C)$ block encoding of $C_{b,M}$ from Theorem~\ref{thm:block-encoding-cbm}. Define
\begin{equation}
    \widehat C_{b,M}
    =(\langle0^{a_C}|\otimes I)U_C(|0^{a_C}\rangle\otimes I).
\end{equation}
By the definition of a block encoding and Theorem~\ref{thm:block-encoding-cbm} we have
\begin{equation}
    \|\widehat C_{b,M}-C_{b,M}\|\leq\epsilon_C = M\epsilon_b.
    \label{eq:implementation-cbm-error}
\end{equation}
First, apply $U_C$ to $|0^{a_C}\rangle|T_\eta\rangle$. The component for which the $a_C$ block-encoding ancillas are all zero is $|0^{a_C}\rangle\widehat C_{b,M}|T_\eta\rangle$.

Next, for each computational-basis state $|x\rangle$, we coherently compute
\begin{equation}
    g_r(x)=\max_{z:d_H(x,z)\leq r}G(z).
\end{equation}
For the restricted-ball application in Section~\ref{subsec:maxkcsp-application}, this maximization and the cardinality estimates below are instead taken over $B_{L_{\mathrm{light}}}(x,r)$; the implementation is otherwise unchanged. By definition of $S_r$, we have $g_r(x)=1$ for every $x\in S_r$. If $\Pi_{S_r}$ projects onto the span of the computational-basis states inside $S_r$, the amplitude on states satisfying the conditions that i) the $a_C$ block-encoding ancillas are all zero and ii) $g_r(x)=1$ is at least
\begin{equation}
    a_S=\|\Pi_{S_r}\widehat C_{b,M}|T_\eta\rangle\|
    \geq|\langle S_r|\widehat C_{b,M}|T_\eta\rangle|.
    \label{eq:implementation-good-amplitude}
\end{equation}
Since $|S_r\rangle$ and $|T_\eta\rangle$ are normalized, Equation \eqref{eq:implementation-cbm-error} implies
\begin{equation}
    |\langle S_r|\widehat C_{b,M}|T_\eta\rangle|
    \geq |\langle S_r|C_{b,M}|T_\eta\rangle|-\epsilon_C.
    \label{eq:implementation-amplitude-error}
\end{equation}

Let $a>0$ satisfy
\begin{equation}
    a\leq |\langle S_r|C_{b,M}|T_\eta\rangle|-\epsilon_C.
    \label{eq:implementation-amplitude-lower-bound}
\end{equation}
Fixed-point amplitude amplification applies to the unitary that prepares $|T_\eta\rangle$, applies $U_C$, and computes $g_r(x)$. It produces a state satisfying the two marking conditions with failure probability at most $\delta_{\mathrm{alg}}$ after
\begin{equation}
    O\!\left(a^{-1}\log\frac{1}{\delta_{\mathrm{alg}}}\right)
\end{equation}
iterations.

A radius-$r$ Hamming ball contains $|S_r|$ points. Quantum maximum finding computes $g_r(x)$ using
\begin{equation}
    O\!\left(
       |S_r|^{1/2}
       \log\frac{1}{\epsilon_{\mathrm{mark}}}
    \right)
\end{equation}
queries to $O_f$, where $\epsilon_{\mathrm{mark}}$ bounds the operator-norm difference between the implemented computation of $g_r(x)$ and an exact reversible computation of $g_r(x)$. By Theorem~\ref{thm:block-encoding-cbm}, applying $U_C$ costs $O(MQ_b(\epsilon_b))$ queries. The total query cost is therefore
\begin{equation}
    O\!\left(
       a^{-1}\log\frac{1}{\delta_{\mathrm{alg}}}
       \left(
          Q_{T_\eta}+M Q_b(\epsilon_b)
             +|S_r|^{1/2}
             \log\frac{1}{\epsilon_{\mathrm{mark}}}
       \right)
    \right).
    \label{eq:implementation-total-cost}
\end{equation}
After measuring a marked center $x$, one final application of quantum maximum finding to its radius-$r$ Hamming ball returns an optimum. This final cost is no larger than the corresponding term in \eqref{eq:implementation-total-cost}.

Finally, let $a_0$ denote the lower bound on $|\langle S_r|C_{b,M}|T_\eta\rangle|$ in \eqref{eq:implementation-amplitude-lower-bound}. Choose $\epsilon_C\leq a_0/2$. Equation \eqref{eq:implementation-amplitude-lower-bound} then allows the choice
\begin{equation}
    a=\frac{a_0}{2}.
\end{equation}
Preparing $|T_\eta\rangle$ from the uniform state uses $Q_{T_\eta}=\widetilde{O}(2^{n/2}/|T_\eta|^{1/2})$ queries to $O_f$. 

Before establishing the final runtime, we need to fix the error parameters. The error $M\epsilon_b$ includes both errors in the proof of Theorem~\ref{thm:block-encoding-cbm}: the two polynomial transformations have combined error at most $M\epsilon_b/2$, and replacing $D_b$ by $\widetilde D_b$ contributes at most $M\epsilon_b/2$.

Fix a target failure probability $\delta_{\mathrm{alg}}\in(0,1/2)$. Run fixed-point amplitude amplification with ideal failure probability at most $\delta_{\mathrm{alg}}/2$. Let $K$ be the total number of calls to the computation of $g_r(x)$ and its inverse. Since the marked amplitude is at least $a_0/2$,
\begin{equation}
    K=O\!\left(
        a_0^{-1}\log\frac{1}{\delta_{\mathrm{alg}}}
    \right).
\end{equation}
Choose
\begin{equation}
    \epsilon_{\mathrm{mark}}
    \leq\frac{\delta_{\mathrm{alg}}}{4K}.
\end{equation}
The triangle inequality then bounds the error accumulated over these calls by $K\epsilon_{\mathrm{mark}}\leq\delta_{\mathrm{alg}}/4$. Hence the total failure probability is at most $\delta_{\mathrm{alg}}$.

Assume that $M=\operatorname{poly}(n)$, $a_0\geq2^{-\operatorname{poly}(n)}$, and $\delta_{\mathrm{alg}}\geq2^{-\operatorname{poly}(n)}$. Then $\log(1/\epsilon_b)$ and $\log(1/\epsilon_{\mathrm{mark}})$ are polynomial in $n$. The two degree-$M$ polynomial transformations use $O(M)$ calls to $U_b$ and $U_b^\dagger$. If their single-qubit rotations are synthesized over a fixed gate set, accuracy $O(\epsilon_b)$ for each rotation gives combined operator-norm error $O(M\epsilon_b)$ and requires $O(M\log(1/\epsilon_b))$ elementary gates. Together with the assumed bound on $Q_b(\epsilon_b)$ and the logarithmic dependence of quantum maximum finding on $1/\epsilon_{\mathrm{mark}}$, these precision choices therefore add only polynomial overhead.

More generally, let $\mathcal P$ be a polynomial-size set of candidate pairs $(b,M)$ supplied by the application analysis, and let $T_{\max}$ be a common running-time bound. Run the algorithm separately for every pair for at most $T_{\max}$ steps, evaluate the objective value of every returned solution, and return the best one. Suppose one pair $(b^\star,M^\star)\in\mathcal P$ satisfies the marked-amplitude lower bound needed to succeed within $T_{\max}$. Repeating each candidate run $O(\bigl(\log(1/\delta_{\mathrm{alg}})\bigr))$ times makes the probability that the certified pair never returns an optimum at most $\delta_{\mathrm{alg}}$. No success guarantee is required for the other pairs because they terminate at the common cutoff. Whenever the certified pair returns an optimum, selecting the largest objective value among all returned solutions also returns an optimum. The total runtime is $O(|\mathcal P|T_{\max}\log(1/\delta_{\mathrm{alg}}))$. In the application theorems proved in this paper, $b^\star$ and $M^\star$ are given explicitly, so this enumeration is optional.

Substituting these expressions into \eqref{eq:implementation-total-cost} gives \eqref{eq:optimization-runtime}.

\subsection{Resolving unknown parameters}
\label{subsec:resolving-unknowns}

{
The implementation does not require $H_{\max}$ or $E_b$ as input. We remove the two quantities separately.

\paragraph{Replacing $E_b$ by a known normalization.}
Let $M_{\max}=\operatorname{poly}(n)$ be a common upper bound on every power used by the application algorithms, and set
\begin{equation}
    \delta_{\mathrm{sv}}:=\frac{1}{M_{\max}+1}.
    \label{eq:sv-slack-choice}
\end{equation}
Let $\overline\mu_G$ and $\overline{\mathcal L}_2(G)$ be known upper bounds on $\mathbb E_\pi[G]$ and $\mathcal L_2(G)$, respectively, and define
\begin{equation}
    E_b^{\mathrm{ub}}
    :=1+b\overline\mu_G
      +\frac{b^2\overline{\mathcal L}_2(G)^2}{8\omega^2}.
    \label{eq:known-energy-normalization}
\end{equation}
Theorem~\ref{thm:energy-bound} implies $E_b\leq E_b^{\mathrm{ub}}$. Define the slack normalization
\begin{equation}
    \widehat E_b:=(1+\delta_{\mathrm{sv}})E_b^{\mathrm{ub}}.
\end{equation}
The implemented operators are
\begin{equation}
    D_b^{(\delta_{\mathrm{sv}})}
       :=\frac{D+bG}{\widehat E_b},
    \qquad
    C_{b,M}^{(\delta_{\mathrm{sv}})}
       :=\frac{1}{M+1}\sum_{m=0}^M
          (D_b^{(\delta_{\mathrm{sv}})})^m
    \label{eq:known-normalized-operators}
\end{equation}
and satisfy $\|D_b^{(\delta_{\mathrm{sv}})}\| \leq(1+\delta_{\mathrm{sv}})^{-1}$.

To compare the implemented amplitude with the analytical bound, define
\begin{equation}
    D_b^{\mathrm{ub}}:=\frac{D+bG}{E_b^{\mathrm{ub}}},
    \qquad
    C_{b,M}^{\mathrm{ub}}
       :=\frac{1}{M+1}\sum_{m=0}^M(D_b^{\mathrm{ub}})^m.
\end{equation}
The amplitude proof already replaces every denominator $E_b^m$ by $(E_b^{\mathrm{ub}})^m$, so it applies directly to $C_{b,M}^{\mathrm{ub}}$. Moreover,
\begin{equation}
    C_{b,M}^{(\delta_{\mathrm{sv}})}
    =\frac{1}{M+1}\sum_{m=0}^M
       (1+\delta_{\mathrm{sv}})^{-m}(D_b^{\mathrm{ub}})^m.
\end{equation}
The matrices $D+bG$ and $(D_b^{\mathrm{ub}})^m$ are entrywise nonnegative, and the set states have nonnegative amplitudes. Consequently,
\begin{equation}
    \left\langle S\left|C_{b,M}^{(\delta_{\mathrm{sv}})}\right|T\right\rangle
    \geq
    (1+\delta_{\mathrm{sv}})^{-M}
    \left\langle S\left|C_{b,M}^{\mathrm{ub}}\right|T\right\rangle.
    \label{eq:slack-amplitude-loss}
\end{equation}
For every $M\leq M_{\max}$, the choice \eqref{eq:sv-slack-choice} gives
\begin{equation}
    (1+\delta_{\mathrm{sv}})^{-M}
    =\exp\!\left(-M\log(1+\delta_{\mathrm{sv}})\right)
    \geq\exp(-M\delta_{\mathrm{sv}})
    \geq e^{-1}.
    \label{eq:slack-amplitude-constant}
\end{equation}
Thus the slack normalization decreases the marked amplitude by at most a constant factor and leaves every exponential runtime bound unchanged. More generally, any inverse-polynomial slack satisfying $M\delta_{\mathrm{sv}}=o(n)$ incurs only a factor $2^{-o(n)}$.

Finally, substituting~\eqref{eq:sv-slack-choice} into \eqref{eq:sv-amplification-degree} gives
\begin{equation}
    d_{\mathrm{sv}}
    =O\!\left(
       M_{\max}\log\frac{1}{\epsilon_b}
    \right)=\operatorname{poly}(n),
\end{equation}
because $M_{\max}=\operatorname{poly}(n)$ and $\log(1/\epsilon_b)=\operatorname{poly}(n)$. Hence the spectral slack both justifies uniform singular-value amplification and preserves the claimed polynomial implementation overhead.

For both applications, all terms in \eqref{eq:known-energy-normalization} are explicit. The walk is the hypercube walk, so $\omega=1/n$. The mean $\mathbb E[G]$ and the bounds on $\mathcal L_2(G)$ are given in the corresponding application proofs. No largest-eigenvalue estimation is performed.

\paragraph{Recovering $H_{\max}$.}
We now describe an outer procedure that determines $H_{\max}$ without exhaustive search. {In both applications, $f$ is a known positive rescaling of the centered objective $H$, so the oracle $O_f$ equivalently evaluates $H$ exactly in binary.\par} Let $\delta_H$ be the spacing of its fixed-point representation, and let $W$ be the known total constraint weight. In both applications,
\begin{equation}
    \mathbb E[H]=0,
    \qquad
    0\leq H_{\max}\leq W,
    \qquad
    \log\frac{W}{\delta_H}=\operatorname{poly}(n).
    \label{eq:objective-range-and-spacing}
\end{equation}
If $H_{\max}=0$, then $H$ is identically zero and every assignment is optimal. Assume below that $H_{\max}>0$.

Using the common upper bound $M_{\max}$ fixed above, choose
\begin{equation}
    \xi:=\frac{1}{n^2(M_{\max}+1)^2}.
\end{equation}
Starting with $U_0=W$, form adjacent intervals
\begin{equation}
    U_j:=\frac{W}{(1+\xi)^j},
    \qquad
    L_j:=\frac{U_j}{1+\xi},
    \label{eq:optimum-candidate-grid}
\end{equation}
until $L_j<\delta_H$. The number of intervals is
\begin{equation}
    O\!\left(\xi^{-1}\log\frac{W}{\delta_H}\right)
    =\operatorname{poly}(n).
\end{equation}
Let $j^\star$ be the first index for which $L_j\leq H_{\max}$. Then
\begin{equation}
    L_{j^\star}\leq H_{\max}\leq U_{j^\star}.
    \label{eq:correct-optimum-bracket}
\end{equation}

The candidate routine processes the intervals in decreasing order. Let $A=W$ for \textup{\textsc{MAX-E$k$-LIN2}} and let $A=B$ for \textup{\textsc{MAX-$k$-CSP}}, where $B$ is the shift defined in Section~\ref{subsec:maxkcsp-application}. For interval $j$, use
\begin{equation}
    {G_j(x):=\frac12\left(1+\min\!\left\{1,\frac{H(x)+A}{U_j+A}\right\}\right)},
    \qquad
    T_j:=\{x:H(x)\geq L_j/2\}.
    \label{eq:candidate-objective-threshold}
\end{equation}
This clipping makes every candidate circuit well-defined. For every interval up to and including $j^\star$, one has $U_j\geq H_{\max}$, so the clipping is inactive. For later intervals, clipping can only decrease the mean and cannot increase any one-step difference. Thus a known upper bound on the mean and the required quadratic-variation bound are obtained from the application proofs by replacing $H_{\max}+A$ with $U_j+A$. Consequently, \eqref{eq:known-energy-normalization} gives a valid known normalization for every candidate routine.

The local-search radius is chosen conservatively from the lower endpoint:
\begin{equation}
    r_j:=
    \begin{cases}
      \left\lfloor L_j/(16d_{\max})\right\rfloor,
          & \textup{\textsc{MAX-E$k$-LIN2}},\\[2mm]
      \left\lfloor L_jn/(16\Gamma_{\max}\Sigma)\right\rfloor,
          & \textup{\textsc{MAX-$k$-CSP}}.
    \end{cases}
    \label{eq:candidate-local-radius}
\end{equation}
At the correct interval, the corresponding ball around $x^\star$ remains inside $T_j$. The proof is the same as in Section~\ref{sec:tilted-walk-analysis}: moving at most $r_j$ steps decreases $H$ by at most $L_j/8$, while $H_{\max}\geq L_j$.

All remaining parameters are computed from $L_j$ and $U_j$. For example, the two levels used in the amplitude theorem can be chosen as
\begin{equation}
    {1-\eta_j:=\frac12\left(1+\frac{A+L_j/2}{A+U_j}\right)},
    \qquad
    {1-\rho_j:=\frac12\left(1+\frac{A+3L_j/4}{A+U_j}\right)}.
    \label{eq:candidate-levels}
\end{equation}
The candidate routine computes $b$ and $M$ from the bounds associated with $L_j$ and $U_j$ using the formulas in Theorem~\ref{thm:normalized-amplitude-gain} and the corresponding application proof. At $j=j^\star$, the ratio $U_j/L_j=1+\xi$. Repeating the application proofs with \eqref{eq:candidate-objective-threshold}--\eqref{eq:candidate-levels} changes each constant in the amplitude exponent by $O(\xi)$ and changes each selected power by at most $O(1+M_{\max}\xi)$. Therefore the marked-amplitude lower bound at the correct interval differs from the known-$H_{\max}$ bound by at most a factor $2^{o(n)}$. In particular, the candidate routine has the same exponential running-time bound.

Each candidate routine is run for its candidate-dependent time bound and then terminated. Threshold-state preparation uses the standard doubling version of amplitude amplification: if no point of $T_j$ is found within this budget, the routine returns \textsc{null} and proceeds to the next interval. It therefore terminates even when $T_j$ is empty. For the intervals preceding $j^\star$, the candidate lower endpoint is larger, so the successful-ball and amplitude exponents give a running-time bound no larger than the bound at $j^\star$. Hence all stages up to the correct interval cost at most a polynomial factor times the correct-stage runtime.

It remains to identify the correct interval without knowing $H_{\max}$. For a center $x$, compute the exact local value
\begin{equation}
    h_j(x):=\max_{z\in B_J(x,r_j)}H(z).
\end{equation}
First test whether $h_j(x)\geq L_j$ by amplitude detection. For every $j<j^\star$, one has $L_j>H_{\max}$, so the marked amplitude is zero. At $j=j^\star$, every successful center satisfies $h_j(x)=H_{\max}\geq L_j$, and the marked amplitude has the lower bound just described. Amplitude detection therefore distinguishes the two cases within the candidate running time.

Once the first nonempty interval is found, use the same test with thresholds inside $[L_j,U_j]$. If a threshold $v$ satisfies $v\leq H_{\max}$, all successful centers are marked; if $v>H_{\max}$, no center is marked. Since $H$ has an exact binary representation, a bitwise binary search using $O(\log(W/\delta_H))$ amplitude tests recovers $H_{\max}$ exactly. The algorithm then constructs the exact normalization of $G$, runs the main tilted-walk algorithm once, and verifies the returned assignment by evaluating $H$.

The candidate grid and the bitwise search both have polynomial size. Assigning inverse-polynomial failure probability to each test and taking a union bound gives bounded total error. Their combined overhead is polynomial, while the $O(\xi)$ perturbation contributes only $2^{o(n)}$. Thus removing prior knowledge of $H_{\max}$ does not change any of the exponential runtime bounds.\par}

\section{Conclusion}
\label{sec:conclusion}
We introduced quantum tilted walks as a framework for transferring amplitude between a chosen initial state and a chosen successful set. At zero bias, $C_{b,M}$ reduces to an average of ordinary walk powers. When a sufficiently large power approximates the projector onto the ground state of the tilted Hamiltonian, the construction recovers generalized short-path algorithms. Quantum tilted walks instead bound the successful-set amplitude produced by $C_{b,M}$ directly, so the transformation is not required to approximate this projector and the analysis does not require a lower bound on the spectral gap of the tilted Hamiltonian. Its performance is determined by a positive path gain and an upper normalization certificate for the tilted operator.

For the threshold states arising in conditioning-and-search, $C_{b,M}$ gives an additional exponential amplitude gain. We used this gain to obtain runtimes
\begin{equation}
    2^{(1-c^{\mathrm{cl}}-\kappa)n/2+o(n)}
\end{equation}
for weighted \textsc{MAX-E$k$-LIN2} and weighted \textsc{MAX-$k$-CSP}, where $2^{(1-c^{\mathrm{cl}})n+o(n)}$ is the corresponding classical conditioning-and-search runtime and $\kappa>0$ under the stated assumptions. Thus tilted walks improve on the ordinary quadratic quantum speedup of the classical algorithm.

The framework also decouples the initial state from the mixer. Theorem~\ref{thm:warm-start-gain} replaces stationarity by the pointwise factor $\lambda_\varphi$ and the lower bound~\eqref{eq:warm-score-regularity} on $K_\varphi G$. The example in Section~\ref{sec:warm-starts} shows that these conditions can hold even when both the spectral gap of the mixer and the spectral gap of the tilted Hamiltonian are exponentially small. It separates the tilted-walk guarantee from the generalized short-path guarantee proved by Chakrabarti et al.~\cite{chakrabarti2025generalizedshortpathalgorithms} for the same mixer; its target is easy to construct, so it is not a hard optimization result.

\section*{AI Usage Disclosure}
\phantomsection\label{sec:ai_use}
The main ideas, methods, and proofs in this paper are due to the authors. Large Language Models have been used to help with writing and grammar, and to audit the technical material for rigor and completeness. All outputs generated by AI have been separately checked by the authors, who are responsible for the correctness of all statements in the draft.

\section*{Acknowledgment}
We thank Rob Otter and Ruslan Shaydulin for executive support and valuable feedback on this project. We thank Dylan Herman for reviewing the earlier drafts of this work. We also acknowledge our colleagues at the Global Technology Applied Research Center of JPMorganChase.

\section*{Disclaimer}
This paper was prepared for informational purposes by the Global Technology Applied Research center of JPMorgan Chase \& Co. This paper is not a product of the Research Department of JPMorgan Chase \& Co. or its affiliates. Neither JPMorgan Chase \& Co. nor any of its affiliates makes any explicit or implied representation or warranty and none of them accept any liability in connection with this paper, including, without limitation, with respect to the completeness, accuracy, or reliability of the information contained herein and the potential legal, compliance, tax, or accounting effects thereof. This document is not intended as investment research or investment advice, or as a recommendation, offer, or solicitation for the purchase or sale of any security, financial instrument, financial product or service, or to be used in any way for evaluating the merits of participating in any transaction.

\printbibliography

\end{document}